\documentclass[12pt,a4paper,oneside]{amsart}
\pdfoutput=1
\usepackage[hmarginratio={1:1},vmarginratio={1:1},textwidth=16cm,textheight=23cm,heightrounded]{geometry}

\usepackage{amsaddr}
\usepackage{stmaryrd}
\usepackage[]{amsmath}
	\numberwithin{equation}{section}
\usepackage[]{amssymb}
\usepackage{amsfonts}
\usepackage{mathrsfs}
\usepackage{amsthm}
	\newtheorem{thm}{Theorem}[section]
	\newtheorem{lem}[thm]{Lemma}
	\newtheorem{prop}[thm]{Proposition}

	\newtheorem{cj}[thm]{Conjecture}

	\theoremstyle{definition}
	\newtheorem{defin}[thm]{Definition}
	\theoremstyle{remark}
	\newtheorem{rmk}[thm]{Remark}

\DeclareMathOperator{\R}{\mathbb{R}}
\DeclareMathOperator{\C}{\mathbb{C}}
\DeclareMathOperator{\N}{\mathbb{N}}
\DeclareMathOperator{\Z}{\mathbb{Z}}
\DeclareMathOperator{\dd}{d}
\DeclareMathOperator{\tr}{Tr}

\DeclareMathOperator{\Div}{Div}
\DeclareMathOperator{\Arg}{Arg}

\newcommand{\hs}{\mathrm{H}\mathbb{S}^2}

\newcommand{\cP}{\mathbb{P}}
\newcommand{\cB}{\mathbb{B}}
\newcommand{\mz}{\mathcal{Z}}

\newcommand{\mM}{\mathcal{M}}
\newcommand{\mE}{\mathcal{E}}

\newcommand{\mA}{\mathcal{A}}
\newcommand{\mC}{\mathcal{C}}
\newcommand{\mV}{\mathcal{V}}
\newcommand{\bN}{\mathbf{N}}
\newcommand{\bM}{\mathbf{M}}

\newcommand{\sR}{\mathsf{R}}

\newcommand{\sG}{\mathsf{G}}
\newcommand{\sT}{\mathsf{T}}

\newcommand{\Br}{\mathcal{B}}
\newcommand{\Ync}{Y_{\mathrm{nc}}}
\newcommand{\Xpre}{X_{\mathrm{pre}}}

\newcommand{\pz}{\widehat{\mathcal{Z}}}

\newcommand{\MF}{\mathrm{MF}}
\newcommand{\DC}{\mathbf{D}}
\newcommand{\Eop}{\mathsf{E}}

\newcommand{\ii}{\mathtt{i}}

\newcommand{\tnc}{t_{\mathrm{nc}}}

\usepackage[utf8]{inputenc}
\usepackage[english]{babel}

\usepackage[pdftex]{graphicx}
\usepackage[dvipsnames]{xcolor}

\usepackage[linktocpage=true,colorlinks=true,citecolor=Blue,linkcolor=purple,urlcolor=Blue]{hyperref}
\usepackage{tikz-cd}

\usepackage{caption}
\usepackage{enumerate}
\usepackage[numbers,compress,square,comma]{natbib}

\DeclareRobustCommand{\SkipTocEntry}[5]{}
\begin{document}
\bibliographystyle{myJHEP}
\captionsetup[figure]{labelfont={sc,small},labelformat={default},labelsep=period,font=small}
\captionsetup[table]{labelfont={sc,small},labelformat={default},labelsep=period,font=small}

{\linespread{1.1}
\title[NC resolutions and pre-quotients of CY double covers]{Non-commutative resolutions and pre-quotients of Calabi--Yau double covers}

\author[]{Tsung-Ju Lee$^{\sharp}$, Bong H.~Lian$^{\ast}$, Mauricio Romo$^{\ddag}$,\protect\\ Leonardo Santilli$^{\flat}$}
\address[]{$^{\sharp}$\small Department of Mathematics,\protect\\ \small National Cheng Kung University, Tainan, Taiwan.}
\address[]{$^{\ast,\ddag}$\small Center for Mathematics and Interdisciplinary Sciences,\protect\\ \small Fudan University, Shanghai,
200433, China.}
\address[]{$^{\ast,\ddag}$\small Shanghai Institute for Mathematics and Interdisciplinary Sciences (SIMIS),\protect\\ \small Shanghai,
200433, China.}
\address[]{$^{\ast}$\small Department of Mathematics,\protect\\ \small Brandeis University, Waltham, MA 02453, USA.}
\address[]{$^{\flat}$\small Yau Mathematical Sciences Center,\protect\\ \small Tsinghua University, Beijing, 100084, China.}
\email{$^{\sharp}$tsungju@gs.ncku.edu.tw}
\email{$^{\ast}$lianbong@gmail.com}
\email{$^{\ddag}$mromoj@simis.cn}
\email{$^{\flat}$santilli@tsinghua.edu.cn}

\linespread{1.12}
\begin{abstract}
	Following an earlier proposal \cite{Lee:2023piu} to apply the GLSM formalism to understand the so-called non-commutative resolution, this paper takes one important step further to extend this formalism to a much larger class of non-commutative resolutions. The proposal was initially motivated by the discovery of a new class of mirror pairs singular Calabi--Yau varieties \cite{Hosono:2020gpj}, given by certain branched double covers over toric varieties of MPCP type. The overarching problem was to understand these mirror pairs from the viewpoint of homological mirror symmetry \cite{Kontsevich:1994dn}. In the present paper, we propose two main results along this line. First, one new insight is that the `gauge-fixing' condition on the branching locus of the double cover used in \cite{Hosono:2020gpj} can be relaxed in an interesting way. This turns out to produce GLSMs that describe a much larger class of non-commutative resolutions, leading to $A$-periods for a larger class of non-commutative resolutions, as well as the GKZ systems for their $A$-periods. Second, we show that the $A$-periods can also be realized as $A$-periods of a certain smooth CICY family in a toric variety of MPCP type, such that a suitable finite quotient of this family recovers the double cover CY we have started with. We call this CICY family the `pre-quotient' of the double cover CY. This realization strongly suggests that pre-quotient may provide an important approach for understanding homological mirror symmetry for singular double cover CY varieties and non-commutative resolutions.
\end{abstract}
}

\maketitle
\clearpage
\tableofcontents
\clearpage 

\section{Introduction}

The main focus of the present paper is to explore some aspects of homological mirror symmetry (HMS) \cite{Kontsevich:1994dn} for certain classes of singular Calabi--Yau. HMS is formulated as an equivalence between a bounded derived category of coherent sheaves and a Fukaya category. Much effort has been made over the past three decades to understand HMS for smooth Calabi--Yau manifolds. The insight from considerations of marginal deformations of superconformal field theories \cite{Lerche:1989uy,Cecotti:1991me} indicates that a special class of quantities, such as chiral/anti-chiral rings or supersymmetric boundary conditions, on a singular Calabi--Yau variety should be equivalent to their counterparts on its crepant resolution, whenever one exists. Mathematically, this idea has been articulated in the crepant resolution conjecture \cite{Ruan:2001xm}. 
Therefore, with the goal of formulating HMS for a singular Calabi--Yau variety, it is natural to look for a notion of categorical resolution of its bounded derived category of coherent sheaves, as well as a suitable replacement for the Fukaya category in the singular setting. For the former, one possibility is to consider a non-commutative resolution, whose origins we can trace physically to \cite{Caldararu:2010ljp} and mathematically to \cite{VandenBergh:2002,Kuznetsov:06}.\par

Following \cite{Lee:2023piu}, this paper represents a further attempt to test these ideas by studying non-commutative resolutions for a larger class of
recently discovered mirror pairs of singular Calabi--Yau varieties \cite{Hosono:2020gpj}. The core idea is that a non-commutative resolution would provide a good model for the derived category of such a singular Calabi--Yau variety. Let us first recall the strategy proposed in \cite{Lee:2023piu}. We began with the initial data: a toric variety $\cB$ of MPCP type, equipped with a nef partition of $-K_{\cB}$. These data describe a particular family of equisingular branched double cover Calabi--Yau varieties $Y$ over $\cB$. The putative mirror Calabi--Yau varieties are likewise double covers over a `mirror' toric variety $\cB^\vee$. In addition, choosing the branching locus to satisfy a certain gauge-fixing condition, following \cite{Hosono:2018kqt,Hosono:2019jet}, would imply that the $B$-periods of this mirror Calabi--Yau family are solutions to a GKZ system. This led to an explicit determination of the $B$-periods. 

Therefore, to test the non-commutative resolution model, the strategy was to find a way to realize the $A$-periods of the singular double covers, starting from the hypothetical non-commutative resolution of $Y$. To this end, we proposed a new $\mathcal{N}=(2,2)$ gauged linear sigma model (GLSM), constructed its so-called hemisphere partition functions following \cite{Hori:2013ika}, and used them to realize the $A$-periods of objects in the category of $B$-branes of the non-commutative resolution. Finally, to connect to the double cover Calabi--Yau family we started with, we showed that the $A$-periods so realized are, in fact, $B$-periods of the mirror Calabi--Yau family \cite{Lee:2023piu}. In the present work, we generalize the construction of Calabi--Yau double covers by allowing the branching loci to be defined by a general nef partition of $-2K_{\mathbb{B}}$. This violates the 'gauge-fixing' condition for the branching locus, hence we cannot apply the mirror construction of \cite{Hosono:2020gpj}, but the definition of non-commutative resolution generalizes straightforwardly. This allows us to compute the $A$-periods following the same prescription; however, in the generic case, we do not have a mirror to compare with. Nevertheless we propose a connection with $A$-periods of smooth Calabi--Yau varieties as we summarize in more detail below.

\subsection{Summary of results}\label{intro:results}
Given the data determining a singular Calabi--Yau double cover $Y\longrightarrow \mathbb{B}$ of a smooth, compact and toric Fano variety $\mathbb{B}$, as explained in \S\ref{sec:doublecover}, we introduce the following indexed lists of vectors in $\mathbb{Z}^s$, with positive entries:
\begin{equation*}
\Theta:=\{ \underline{\theta}^{(1)},\ldots,\underline{\theta}^{(N)}\},\qquad \mathcal{D}:=\{\underline{d}^{(1)},\ldots,\underline{d}^{(\hat{r})}\} .
\end{equation*}
The vectors $\underline{\theta}^{(i)}$ correspond to the $(\mathbb{C}^{*})^{s}$ weights of the GIT quotient representation of $\mathbb{B}$, while the vectors $\underline{d}^{(\alpha)}$ to the weights of the sections $s_{\alpha}$ determining the branching locus of $Y$. More precisely, if we denote $(z_{1},\ldots,z_N)$ the homogeneous coordinates of $\mathbb{B}$, we have
\begin{equation*}
\begin{aligned}
    \lambda\cdot z_i &= \left(\prod_{a=1}^{s}\lambda_{a}^{\theta^{(i)}_{a}}\right) z_i , \\
    s_{\alpha}(\lambda\cdot z)&=\left(\prod_{a=1}^{s}\lambda_{a}^{d^{(\alpha)}_{a}}\right) s_{\alpha}(z),\qquad \lambda\in(\mathbb{C}^{*})^{s}.
\end{aligned}
\end{equation*}
The cardinality of these lists is given by $|\Theta|=N$ and $|\mathcal{D}|=\hat{r}$, i.e.~$\underline{\theta}^{(i)}$ and $\underline{\theta}^{(j)}$ are considered distinct if $i\neq j$, and the same holds between the elements of $\mathcal{D}$. Next, from these vectors we form the set:
\begin{equation*}
\mathscr{J}:=\{ (\underline{\theta}^{(i)},\underline{d}^{(\alpha_{i})})\in \Theta \times \mathcal{D}~|~\underline{d}^{(\alpha_{i})}=\underline{\theta}^{(i)}~\text{and}~\underline{\theta}^{(i)}\neq \underline{\theta}^{(j)}, \underline{d}^{(\alpha_{i})}\neq \underline{d}^{(\alpha_{j})}~\text{for any two pairs}\} ,
\end{equation*}
and define the natural maps $p_j\colon\Theta \times \mathcal{D}$, $j=1,2$ to $\Theta$ and $\mathcal{D}$, respectively. Now define
\begin{eqnarray*}
\mathscr{H}:=\left\{(\underline{d}^{(\alpha)},\underline{d}^{(a_{\alpha})})\in \mathcal{D}\times  \mathcal{D}~\middle|~
\begin{aligned}
    &\underline{d}^{(a_{\alpha})}\in\mathcal{D}\setminus p_2(\mathscr{J}) ~\text{satisfies}~2\underline{d}^{(a_{\alpha})}=d^{(\alpha)}~\text{and}
    \nonumber\\
    &\underline{d}^{(\alpha)}\neq \underline{d}^{(\beta)},\underline{d}^{(a_{\alpha})}\neq \underline{d}^{(a_{\beta})}~\text{for any two pairs}.
\end{aligned}
\right\}
\end{eqnarray*}
and we define the maps $\pi_{j}\colon \mathcal{D} \times \mathcal{D}$, $j=1,2$ to the first and second factor of $\mathcal{D}$, respectively.\par 
Our main result is given in Theorem \ref{thm:smoothX}, which establishes that the $A$-periods of the noncommutative resolution of $Y$, denoted $\Ync$ and defined in \S\ref{sec:NCmath}, and the $A$-periods of the (smooth) Calabi--Yau complete intersection $X$ satisfy the same GKZ system, after a simple change of variables. The variety $X$ of Theorem \ref{thm:smoothX} is given explicitly as a complete intersection inside the GIT quotient:
\begin{eqnarray*}
\mathbb{C}^{N-\lvert\mathscr{J}\rvert +\hat{r}- \lvert \mathscr{H}\rvert} /\!\!/ (\mathbb{C}^{*})^{s} ,
\end{eqnarray*}
where $(\mathbb{C}^{*})^{s}$ acts with weights $2\underline{\theta}^{(i)}$ for each $\underline{\theta}^{(i)} \in \Theta\setminus p_{1}(\mathscr{J})$ on the first $(N-\lvert\mathscr{J}\rvert)$ coordinates, and with weights $\underline{d}^{(\alpha)} \in \mathcal{D}\setminus\pi_{1}(\mathscr{H})$ on the remaining $(\hat{r}- \lvert \mathscr{H}\rvert)$ coordinates. Then $X$ is given by the complete intersection of the vanishing loci of generic sections of weight $2\underline{d}^{(\alpha)}$, for $\underline{d}^{(\alpha)}\in\mathcal{D}\setminus(p_2(\mathscr{J})\cup\pi_2(\mathscr{H}))$. Furthermore, whenever any set of weights $d^{(\alpha)}_{a}$ is divisible by $2$, for all $\alpha=1,\ldots \hat{r}$, the model can be simplified a bit further by dividing all the weights under the action of the $a^{\text{th}}$ torus by $2$.\par
As an example, consider a double cover of the weighted projective space $\cP^{N-1}_{\vec{\theta}}$ branched over the $N$ coordinate hyperplanes and $r$ hypersurfaces of degrees $\vec{d}=(d^{(1)}, \dots, d^{(r)})$, subject to $\sum_{\alpha=1}^r d^{(\alpha)} = \sum_{i=1}^{N} \theta^{(i)}$. In this case, $\Theta=\vec{\theta}$, $\mathcal{D}=\{d^{(1)}, \dots, d^{(r)}, \theta^{(1)}, \dots, \theta^{(N)}\}$, and $\Theta\setminus p_{1}(\mathscr{J})=\emptyset$. We have that $X=\cP^{r+N-1}_{\vec{d},\vec{\theta}} \left[ d^{(1)}, \dots, d^{(r)}\right] $ is a complete intersection of $r$ hypersurfaces of even degree in the weighted projective space $\cP^{r+N-1}_{\vec{d},\vec{\theta}}$, described as a toric variety $\C^{r+N} /\!\!/\C^{\ast}$ where the torus acts on the affine coordinates with weights $\mathcal{D}$.\par
\medskip
Moreover, we show in \S\ref{sec:prequotient} that in the case gauge fixing (as defined in \S\ref{sec:LLY}) is available, the singular variety $Y$ can be written as a global quotient $\Xpre/\Gamma$ for some finite abelian group $\Gamma$. We conjecture that 
\begin{equation*}
X=\Xpre .
\end{equation*}

\subsection{Outline}
The rest of this work is organized as follows.\par
We start by setting the stage and reviewing the relevant constructions in 
\S\ref{sec:CYpreliminaries} and in \S\ref{sec:GLSM}. The main results are collected in \S\ref{sec:maingeom} and \S\ref{sec:mainGLSM}.\par
In more detail, after briefly introducing the necessary toolkit from toric geometry (\S\ref{sec:toric}), we define the Calabi--Yau double covers of \cite{Hosono:2020gpj} in \S\ref{sec:doublecover}. Their non-commutative resolutions are defined in \S\ref{sec:NCmath}, and their associated GKZ systems are given in \S\ref{sec:GKZdef}. 

In \S\ref{sec:GLSM}, we review the GLSM description of the non-commutative resolutions of singular double covers, and define the corresponding $A$-periods.\par
In \S\ref{sec:maingeom}, we construct the pre-quotient geometries for the singular double covers. Conjectures \ref{conj1} and \ref{conj:prequotient} predict an explicit relation between the $A$-periods of the non-commutative resolution and the $A$-periods of the pre-quotient. The conjecture is substantiated in \S\ref{sec:mathEx} in families of selected examples. 

In \S\ref{sec:mainGLSM}, to every Calabi--Yau double cover of any toric Fano variety we associate a Calabi--Yau complete intersection whose $A$-periods satisfy the same GKZ system, using GLSM techniques. The central result Theorem \ref{thm:smoothX} is stated and explained in \S\ref{sec:GLSMGKZderivation}. Examples of Picard rank one are worked out explicitly in \S\ref{sec:GLSMEx}.\par
The text is supplemented with technical appendices. Appendix \ref{app:Notation} summarizes our notation. Appendix \ref{app:CICYperiod} contains details and proofs of computations for Calabi--Yau complete intersections. Finally, Appendix \ref{app:FImoduli} explains the origin of the GKZ variables in the GLSM framework, and unambiguously fixes the relation between hemisphere partition functions and $A$-periods in full generality.

\addtocontents{toc}{\SkipTocEntry}
\subsection*{Acknowledgements}
We would like to thank Shinobu Hosono, Johanna Knapp and Kentaro Hori for enlightening discussions. MR and LS are grateful to the Department of Mathematics at National Cheng Kung University for hospitality during the completion of this work. LS also thanks SIMIS for hospitality at various stages of this project. TJL is supported by
NSTC grant 112-2115-M-006-016-MY3.
The work of LS was supported by the Shuimu Scholars program of Tsinghua University and by the National Natural Science Foundation of China grant W2433005 ``String theory, supersymmetry, and applications to quantum geometry.''

\section{Calabi--Yau double covers and non-commutative resolutions}
\label{sec:CYpreliminaries}
This preliminary section introduces the main characters. Ingredients from toric geometry are collected in \S\ref{sec:toric}, and we refer to \cite{Borisov:1993mirror,Cox:2011tor} for more extensive treatments. We then introduce Calabi--Yau branched double covers in \S\ref{sec:doublecover}, closely following \cite{Hosono:2020gpj,Lee:2020oyt}. Non-commutative resolutions are introduced in \S\ref{sec:NCmath}. In \S\ref{sec:GKZdef} we define the GKZ system for the Calabi--Yau varieties of interest in this work.
For convenience, we summarize our notation in Appendix \ref{app:Notation}.

\subsection{Toric geometry}
\label{sec:toric}
In this subsection, we briefly recall the Batyrev--Borisov duality for toric complete intersections \cite{Batyrev:1994pg}.\par

\subsubsection{Toric varieties}
Let $\bN =\Z^{n}$ be a lattice of rank $n$
and $\bM := \mathrm{Hom}_{\Z} (\bN, \Z )$ be its dual lattice. Denote by $\bN_{\R} := \bN \otimes_{\Z} \R$ and $\bM_{\R} := \bM \otimes_{\Z} \R$ their real scalar extensions. The canonical pairing between $\bM_{\R}$ and $\bN_{\R}$
will be denoted by the bracket 
$$\langle -, -\rangle\colon \bN_{\R} \times \bM_{\R} \to \R.$$
Recall that a polyhedron is an intersection of a finite number of
half-spaces in Euclidean space.
A bounded polyhedron is called a polytope; it is
the convex hull of a finite number of points in Euclidean space.
A polyhedron (polytope) is called a lattice polyhedron (polytope)
if all of its vertices belong to the relevant integral lattice.
Let $\Delta$ be a lattice polytope in \(\mathbf{M}_{\mathbb{R}}\). 
If \(\mathbf{0}\in\Delta\), 
the polar dual $\Delta^{\ast}$ is defined as the set
\begin{equation*}
	\Delta^{\ast} := \left\{ x \in \bN_{\R}~|~\langle x,y \rangle \ge -1,~\forall y \in \Delta  \right\}.
\end{equation*}

A reflexive polytope is a lattice polytope 
whose polar dual is again a lattice polytope.
A reflexive polytope \(\Delta\) in \(\mathbf{M}_{\mathbb{R}}\) 
determines two distinguished fans:
the normal fan \(\mathcal{N}(\Delta)\) in \(\mathbf{N}_{\mathbb{R}}\)
and the face fan \(\mathcal{F}(\Delta)\) in \(\mathbf{M}_{\mathbb{R}}\).
In particular, we have
\begin{equation*}
    \mathcal{N}(\Delta)=\mathcal{F}(\Delta^{\ast})~\mbox{and}~
    \mathcal{F}(\Delta)=\mathcal{N}(\Delta^{\ast}).
\end{equation*}
Given a fan \(\Sigma\), one can construct a toric variety \(X_{\Sigma}\)
in the usual way. We denote by \(\Sigma(k)\)
the set of all \(k\)-dimensional cones in \(\Sigma\).
According to the construction, the toric divisors in \(X_{\Sigma}\)
are in one-to-one correspondence with \(\Sigma(1)\).

Each polytope \(\Delta\) determines a projective 
toric variety by
\(\mathbb{P}_{\Delta}:=X_{\mathcal{N}(\Delta)}\).

The toric variety \(\mathbb{P}_{\Delta}\)
is often singular. However, in the present circumstance, one can take a
so-called \emph{maximal projective crepant partial (MPCP) desingularization};
it is defined as the projective toric variety \(\cB_{\Delta}\) whose defining fan
is a refinement of \(\mathcal{N}(\Delta)=\mathcal{F}(\Delta^{\ast})\) obtained
from firstly adding new one-cones whose primitive generators belong to 
\(\Delta^{\ast}\cap \mathbf{N}\), and then taking a simplicialization
for the resulting fan. MPCP desingularizations always exist
but are usually not unique. 
\begin{rmk}
    In \cite{Lee:2020oyt,Lee:2023piu}, an MPCP desingularization
    of \(\mathbb{P}_{\Delta}\) 
    is denoted by \(X\). 
    The choice of letter \(\mathbb{B}\)
    in the present work 
    is because later on this will be the `base', and 
    also $\cB_{\Delta}$ is really a blow-up of $\cP_{\Delta}$.
\end{rmk}

\subsubsection{Nef partitions}
\label{sec:nefpartition}
Let \(\Delta\) be a 
reflexive polytope in \(\mathbf{M}_{\mathbb{R}}\).
\begin{defin}\label{def:BBnef}
A \emph{nef partition on $\cP_{\Delta}$} is a partition of the set of \(1\)-cones
\begin{equation*}
\mathcal{N}(\Delta)(1)= \mathcal{I}_1 \sqcup \cdots \sqcup \mathcal{I}_r ,
\end{equation*}
such that there exist integral convex piece-wise 
linear functions $\varphi_1, \dots, \varphi_r$ on $\mathcal{N}(\Delta)$ satisfying the conditions
\begin{equation}
\label{eq:Borisovphi}
	\varphi_{\alpha} \left(  \nu \right) = \begin{cases} 1 & \text{if } \nu \in \mathcal{I}_{\alpha}, \\ 0 & \text{otherwise} \end{cases}
\end{equation}
for all \(\alpha=1,\ldots,r\).
\end{defin}
For a \(1\)-cone \(\nu\in\mathcal{N}(\Delta)(1)\), we denote by \(D_{\nu}\) the 
corresponding Weil divisor.
Following \cite{Borisov:1993mirror}, to a nef partition we associate the 
collection $\{E_1, \dots, E_r\}$ of nef divisors 
\begin{equation*}
	E_{\alpha} := \sum_{\nu \in \mathcal{I}_{\alpha}} D_{\nu} 
\end{equation*}
on \(\mathbb{P}_{\Delta}\).
The definition already implies that 
each \(E_{\alpha}\) is Cartier and numerically effective (nef), and
certainly,
\begin{equation*}
    E_{1}+\cdots+E_{r}=-K_{\mathbb{P}_{\Delta}}.
\end{equation*}
Since \(\mathbb{B}_{\Delta}\to \mathbb{P}_{\Delta}\) is crepant,
the pull-back of the decomposition to $\cB_{\Delta}$ also gives
rise to a decomposition of the anti-canonical divisor of \(\mathbb{B}_{\Delta}\)
and it will be denoted by the same symbol.
Moreover, a nef partition on $\cP_{\Delta}$ determines a 
collection $\left\{ \Delta_1, \dots ,\Delta_r \right\}$ of lattice polytopes in \(\bM_{\R}\) such that
\begin{equation*}
	\Delta_1 + \cdots + \Delta_r = \Delta,
\end{equation*}
where \(\Delta_{\alpha}\) is the section
polytope of the divisor \(E_{\alpha}\).
By construction, if we set 
\begin{equation*}
    \nabla_{\alpha}=\mathrm{Conv}\left(\{\mathbf{0}\}\cup\{\nu,~\nu\in \mathcal{I}_{\alpha}\}\right),
\end{equation*}
then one can show that 
\begin{equation*}
    \nabla :=\nabla_{1}+\cdots+\nabla_{r}
\end{equation*}
is again a reflexive polytope in \(\mathbf{N}_{\mathbb{R}}\).
Here, we slightly abuse the notation; for a \(1\)-cone \(\nu\), 
we also use the same symbol \(\nu\)
to denote its primitive generator. The decomposition 
\(\nabla =\nabla_{1}+\cdots+\nabla_{r}\) determines a
nef partition
\begin{equation*}
    F_{1}+\cdots+F_{r}=-K_{\mathbb{P}_{\nabla}}
\end{equation*}
on the toric variety \(\mathbb{P}_{\nabla}\).
We can choose a MPCP desingularization \(\mathbb{B}_{\nabla}\to\mathbb{P}_{\nabla}\)
and the pullback of the decomposition gives rise to a nef-partition on \(\mathbb{B}_{\nabla}\).

\subsubsection{Calabi--Yau complete intersections}
The toric data of $\cB_{\nabla}$ equipped with a nef partition $F_1 + \cdots + F_r$ determine a family of 
Calabi--Yau complete intersections in 
$\cB_{\nabla}$ \cite{Batyrev:1994pg}.
\begin{prop}[{\cite[Proposition 4.15]{Batyrev:1994pg}}] 
Let $f^{(\alpha)}$ be a generic Laurent polynomial whose Newton polyhedron is $\nabla_{\alpha}$, $\forall \alpha=1, \dots, r$, and define 
\begin{equation*}
	X:=\overline{ \left\{ f^{(1)}=\cdots=f^{(r)}=0\right\}} \subset \cB_{\nabla} .
\end{equation*}
Then $X$ is Calabi--Yau variety of dimension $\dim (X)=n-r$. Moreover,
\(X\) has at most Gorenstein terminal abelian quotient singularities.
\end{prop}
We denote by $d_a^{(\alpha)} $ the degree of $f^{(\alpha)}$ under the action of the $a$\textsuperscript{th} component
of the gauge group $(\C^{\ast})^{s}$. 
We will sometimes collect them into vectors $\underline{d}^{(\alpha)} \in \Z^s$ to lighten the expressions.

\subsubsection{Lattices of relations and toric GITs}
\label{sec:latticeL}
The toric variety \(\mathbb{B}_{\nabla}\) can be viewed as a GIT quotient which
we shall recall now.
Let us denote by \(s=\dim\mathrm{H}^{2}(\cB_{\nabla};\mathbb{Q})\).
Then there exists a short exact sequence 
\begin{equation*}
\begin{tikzcd}
    &0\ar[r] &\mathbf{N}\ar[r] &\mathbb{Z}^{N}
    \ar[r] &\mathrm{Cl}(\mathbb{B}_{\nabla})\ar[r] &0
\end{tikzcd}
\end{equation*}
where \(N\) is the number of 1-cones in the defining fan of \(\mathbb{B}_{\nabla}\).

Applying \(\mathrm{Hom}_{\mathbb{Z}}(-,\mathbb{C}^{\ast})\) to the short exact sequence, we obtain
\begin{equation*}
\begin{tikzcd}
    &1\ar[r] &\mathrm{Hom}_{\mathbb{Z}}(\mathrm{Cl}(\mathbb{B}_{\nabla}),\mathbb{C}^{\ast})
    \ar[r] &(\mathbb{C}^{\ast})^{N}\ar[r] 
    &T_{N}\ar[r]&1.
\end{tikzcd}
\end{equation*}
We assume that \(\mathbb{B}_{\nabla}\)
is \emph{smooth}. In such case, we have
\begin{equation*}
    \mathrm{Cl}(\mathbb{B}_{\nabla})=\mathrm{Pic}(\mathbb{B}_{\nabla})
\end{equation*}
and the latter group is torsion free (cf.~\cite[Proposition 4.2.5]{Cox:2011tor}).
We infer that the `gauge group' 
\(\mathrm{Hom}_{\mathbb{Z}}(\mathrm{Cl}(\mathbb{B}_{\nabla}),\mathbb{C}^{\ast})\) is an algebraic torus 
of dimension \(s=N-n\).

Let \(\chi\colon \mathrm{Hom}_{\mathbb{Z}}(\mathrm{Cl}(\mathbb{B}_{\nabla}),
\mathbb{C}^{\ast})\to \mathbb{C}^{\ast}\) be the character 
induced by an ample divisor on \(\mathbb{B}_{\nabla}\). Then there is an isomorphism
\begin{equation*}
    (\mathbb{C}^{N})^{\mathsf{s}}_{\chi}~\slash~\mathrm{Hom}_{\mathbb{Z}}(\mathrm{Cl}(\mathbb{B}_{\nabla}),\mathbb{C}^{\ast})
    \cong \mathbb{B}_{\nabla}.
\end{equation*}
Here \((\mathbb{C}^{N})^{\mathsf{s}}_{\chi}\) denotes
the set of stable points determined by
the linearized line bundle \(\mathscr{L}_{\chi}:=\mathbb{C}^{N}\times \mathbb{C}_{\chi}\)
(cf.~\cite[Chapter 14]{Cox:2011tor}).

Let \(\{\nu_{1},\ldots,\nu_{N}\}\) be
the set of all non-zero integral points in a reflexive polytope
\(\Delta\subset\bM_{\mathbb{R}}\).
Thus, they live in \(\mathbf{M}\) and they are
primitive generators of the \(1\)-cones 
in the fan defining the dual toric variety \(\mathbb{B}_{\nabla}\).
Regarding 
\begin{equation*}
    \begin{bmatrix}
        \nu_{1} & \cdots & \nu_{N}
    \end{bmatrix}
\end{equation*}
as a matrix, we obtain a surjection \(\mathbb{Z}^{N}\to\mathbb{Z}^{n}\).
Let \(\mathbb{L}\) be the kernel. Then 
\(\mathbb{L}\) is a lattice of rank \(s:=N-n\).
Every element \(\vec{\ell}\in\mathbb{L}\) can be written as an integral vector
\begin{equation*}
    \vec{\ell} = (\ell^{(1)},\ldots,\ell^{(N)})\in\mathbb{Z}^{N}.
\end{equation*}
By the construction (see in particular \cite{Hosono:1995bm}), 
there exist $s$ vectors
\begin{equation*}
	\vec{\ell}_a = (\ell_a ^{(1)} , \dots, \ell_a ^{(N)} ) \in \Z^{N}~\mbox{with}~a=1, \dots, s=N-n
\end{equation*}
which generate \(\mathbb{L}\), and such that 
\begin{equation*}
	\sum_{i=1}^{N} \ell_a ^{(i)} \nu_{i} =0  , \qquad \forall a=1, \dots, s .
\end{equation*}
In other words, the collection $\{ \vec{\ell}_1,\ldots,\vec{\ell}_s\}$ forms
a basis for \(\mathbb{L}\). \par
For later convenience, we introduce the collection of $N$ vectors of $s$ entries
\begin{equation}
\label{eq:ellweights}
	\underline{\ell}^{(i)} = (\ell_1 ^{(i)} , \dots, \ell_s ^{(i)} ) \in \Z^{s}~\mbox{with}~i=1, \dots, N .
\end{equation}
That is, given the rectangular matrix 
$$\left[\ell^{(i)}_{a} \right]_{\substack{i=1, \dots, N \\ a=1, \dots, s}}\in\mathrm{Mat}_{s\times N}(\mathbb{Z}),$$ 
$\vec{\ell}_a$ and $\underline{\ell}^{(i)}$ denote, respectively, row and column vectors.\par
Using the fact that $\mathcal{I}_{\alpha} \cap \mathcal{I}_{\beta} = \emptyset$ for $\beta \ne \alpha$, the integrality of $\ell_a^{(i)}$ and $d_a ^{(\alpha)}$, and the fact that $\underline{d} ^{(\alpha)}$ originate from the nef partition, up to relabeling one shows 
\begin{lem}It holds that 
\begin{equation}
\label{eq:CYnefcharges}
	\sum_{i \in \mathcal{I}_{\alpha}} \ell_a ^{(i)} = d_a ^{(\alpha)} .
\end{equation}
\end{lem}
In particular, \eqref{eq:CYnefcharges} implies the identity 
\begin{equation*}
	\sum_{i=1}^{N} \ell_a ^{(i)} = \sum_{\alpha=1}^{r} d_a ^{(\alpha)} , \qquad \forall a=1, \dots, s , 
\end{equation*}
equivalent to the condition that $X$ is Calabi--Yau, which in turn is guaranteed by the adjunction formula.\par
We consider the embedding $\mathbf{M} \hookrightarrow \Z^{r} \oplus~\mathbf{M} $ and 
define the liftings \(\tilde{\nu}_{1},\ldots,\tilde{\nu}_{r+N}
\in \Z^{r} \oplus~\mathbf{M}\) according to the formula
\begin{equation}
\begin{aligned}
	\tilde{\nu}_1 &= (\underbrace{1, 0, \dots, 0}_{r} , \underbrace{0, \dots, 0}_{n}) \\
	\tilde{\nu}_2 &= (\underbrace{0, 1 ,\dots, 0}_{r} ,\underbrace{0, \dots, 0}_{n}) \\
	& \hspace{1.8cm}\vdots \\
	\tilde{\nu}_{r} &= (\underbrace{0, 0, \dots, 1}_{r} ,\underbrace{0, \dots, 0}_{n}) \\
	\tilde{\nu}_{r+i} &= (\underbrace{ \varphi_1 (\nu_i), \varphi_2 (\nu_i), \dots, \varphi_r (\nu_i)}_{r} ,\underbrace{ \nu_{i,1}, \dots, \nu_{i,n}}_{n})
\end{aligned}
\label{eq:tildenuI}
\end{equation}
where $\varphi_{\alpha}$, which originate from the nef partition, have been defined in \eqref{eq:Borisovphi}.\par

Again, we regard
\begin{equation}\label{eq:tildenutildeL}
    \begin{bmatrix}
        \tilde{\nu}_{1} & \cdots & \tilde{\nu}_{r+N}
    \end{bmatrix}
\end{equation}
as a matrix, so we obtain a surjection 
\(\mathbb{Z}^{r+N}\to\mathbb{Z}^{r+n}\).
Denote by \(\widetilde{\mathbb{L}}\) the kernel of the surjection above.
It can be shown that \(\widetilde{\mathbb{L}}\cong\mathbb{L}\).

We can regard the integral vectors $\underline{\ell}^{(i)} $ in \eqref{eq:ellweights} as \emph{weights}
in the quotient construction of toric varieties. In which case, we use
\(\underline{\theta}^{(i)}\) to denote the integral vectors \(\underline{\ell}^{(i)}\)
to emphasize that they are regarded as weights in a GIT quotient.

Given \(N\) integral vectors \(\underline{\theta}^{(1)},\ldots,\underline{\theta}^{(N)}\in \mathbb{Z}^{s}\)
(with our convention \(\underline{\ell}^{(i)}=\underline{\theta}^{(i)}\)), 
we can define an action by letting $\lambda=(\lambda_1, \dots, \lambda_s) \in 
(\C^{\ast})^s$ act on $\C^N$ according to
\begin{equation*}
    \lambda \cdot (z_1, \dots, z_N) = \left( \prod_{a=1}^{s} \lambda_a ^{\theta_a ^{(1)}} z_1 , \dots, \prod_{a=1}^{s} \lambda_a ^{\theta_a ^{(N)}} z_N \right).
\end{equation*}
To each character \(\chi\colon(\mathbb{C}^{\ast})^{s}\to \mathbb{C}^{\ast}\) we 
associate a \(1\)-dimensional representation \(\mathbb{C}_{\chi}\).
We then have a trivial line bundle \(\mathbb{C}^{N}\times\mathbb{C}_{\chi}\)
and we may talk about the stability.
Denote by $Z_{\chi} \subset \C^N$ the unstable locus
with respect to \(\chi\). We can
form a geometric quotient
\cite{Cox:1993fz}
\begin{equation}
\label{eq:toricquotient}
    \left( \C^N \setminus Z_{\chi}  \right) / (\C^{\ast})^s 
\end{equation}
and it turns out that this is a toric variety.
Conversely, beginning with a projective simplicial toric variety $\mathcal{X}$,
if \(\chi\) is a character corresponding to an ample divisor on $\mathcal{X}$, then 
$\mathcal{X}$ is a geometric quotient, i.e.
\begin{equation*}
    \mathcal{X} \cong \left( \C^N \setminus Z_{\chi}  \right) / (\C^{\ast})^s 
\end{equation*}
For instance, a weighted projective space is the geometric quotient (with \(s=1\))
\begin{equation*}
    \cP^{N-1}_{\vec{\theta}} = \left( \C^N \setminus \{0\} \right) / \C^{\ast},
\end{equation*}
where the homogeneous coordinates have weights $\vec{\theta}=(\theta^{(1)}, \dots, \theta^{(N)})$.\par
In a nutshell, the inclusion $ (\C^{\ast})^s \subset  (\C^{\ast})^N$ induces a surjective map on the character groups 
\begin{equation*}
    \mathrm{Hom}_{\Z} ( (\C^{\ast})^N, \C^{\ast} ) \cong \Z^N \longrightarrow \mathrm{Hom}_{\Z} ( (\C^{\ast})^s, \C^{\ast} ) \cong \Z^s .
\end{equation*}
For any choice of character $\chi \in \Z^s $, one forms the GIT quotient 
\begin{equation}
    \C^N /\!\!/ _{\chi} (\C^{\ast})^s:=(\mathbb{C}^{N})_{\chi}^{\mathsf{s}}\slash 
(\mathbb{C}^{\ast})^{s}.
\end{equation}
It is known 
that the underlying variety is toric, and agrees with \eqref{eq:toricquotient}.

\subsection{Calabi--Yau double covers}
\label{sec:doublecover}
With the setup as in \S\ref{sec:toric}, let $\cB_{\nabla}$ be an MPCP desingularization of the toric Fano variety $\mathbb{P}_{\nabla}$, and let
\begin{equation*}
p\colon Y \xrightarrow{~2:1~}\cB_{\nabla} 
\end{equation*}
be a branched double cover. The branching locus is a divisor $D_{\mathrm{br}}$, and $Y$ has trivial canonical bundle if and only if the linear equivalence
\begin{equation*}
	D_{\mathrm{br}} \sim -2K_{\cB_{\nabla}} 
\end{equation*}
holds.\par

\subsubsection{General nef partitions}
We now provide an analogue of the nef partition of Definition \ref{def:BBnef} for double covers. 
The branching locus is now a divisor 
linearly equivalent to \(-2K_{\mathbb{B}_{\nabla}}\).
Following \S\ref{sec:nefpartition}, we consider a 
decomposition of twice the anti-canonical divisor:
\begin{equation*}
    F_1 + \cdots + F_{\hat{r}} = -2K_{\cB_{\nabla}}.
\end{equation*}
This choice corresponds to a partition of the sum \(
-2K_{\cB_{\nabla}}=2\sum_{\nu\in\Sigma(1)}D_{\nu}\)
into nef Cartier divisors, where \(\Sigma\) is the 
fan defining \(\mathbb{B}_{\nabla}\). Let us write
\begin{equation}
\label{eq:general-nef-partition}
-2K_{\cB_{\nabla}}=2\sum_{\nu\in\Sigma(1)}D_{\nu}
= F_1 + \cdots + F_{\hat{r}} ,
\end{equation}
where each \(F_{\alpha}\) is a nef Cartier divisor 
of the form \(F_{\alpha}=\sum_{\nu\in\Sigma(1)}a_{\alpha, \nu}D_{\nu}\)
with \(a_{\alpha, \nu}\in\{0,1,2\}\). To distinguish from 
the classical nef-partition, which is a decomposition
of the anticanonical divisor, we shall call such
a decomposition \eqref{eq:general-nef-partition}
a \emph{general nef partition}.

\subsubsection{Calabi--Yau double covers with gauge fixing}
\label{sec:LLY}
Let $F_1 + \cdots + F_{r}$ be a nef partition of $-K_{\cB_{\nabla}}$ as in 
\S\ref{sec:nefpartition}. The works \cite{Hosono:2018kqt,Hosono:2019jet,Hosono:2020gpj,Lee:2020oyt} restrict to 
\begin{equation}
\label{eq:rplusN}
    \hat{r}=r+N
\end{equation}
and consider the singular double cover $Y \longrightarrow \cB_{\nabla} $, where $D_{\mathrm{br}}$ is constructed from the section $s=s_1 \cdots s_{r+N}$ with 
\begin{equation}
\label{eq:branchingsec}
	s_{\alpha} \in \mathrm{H}^{0} (\cB_{\nabla} , F_{\alpha}),~\alpha=1,\ldots,r.
\end{equation}
These sections are assumed smooth with 
\begin{equation}
\label{eq:gaugefixing}
	\Div \left( s_{r + i} \right) = D_i,~i=1, \ldots, N.
\end{equation}
\begin{defin}
    In this situation, we will say that the last $N$ entries of the nef partition are \emph{gauge fixed}.
\end{defin}
\par
Under such hypotheses on the branching locus, \cite{Hosono:2020gpj,Lee:2020oyt} 
proposed a mirror Calabi--Yau variety for $Y$.
One can show that the \(B\)-periods for the Calabi--Yau double covers with 
gauge fixing considered
in \cite{Lee:2020oyt} satisfy a certain GKZ system.

In this work we remove the restrictions \eqref{eq:rplusN}--\eqref{eq:gaugefixing} and allow for $\hat{r} \le N$, significantly enlarging the family of double covers considered, compared to the existing literature.

\subsection{Non-commutative resolutions}
\label{sec:NCmath}

The data above (\S\ref{sec:doublecover}) induce a sheaf $\mC_0$ of even parts of Clifford algebras on $\cB_{\nabla}$ \cite{Kuznetsov:08}. We let $\mC_0-\mathrm{mod}$ be the category of sheaves of finitely generated $\mC_0$-modules over $\cB_{\nabla}$.
\begin{defin}The non-commutative resolution $\Ync$ of $Y$ is the non-commutative algebraic variety defined by the pair $(\cB_{\nabla}, \mC_0)$, such that
\begin{equation*}
	\DC (\Ync) := D^b (\cB_{\nabla}, \mC_0\mathrm{-mod}).
\end{equation*}
\end{defin}
where $D^b (\cB_{\nabla}, \mC_0\mathrm{-mod})$ stands for the derived category of sheaves of finitely generated $ \mC_0$ modules over $\cB_{\nabla}$.
\begin{defin}
	Let $\Ync$ be the non-commutative resolution of a Calabi--Yau singular double cover of a toric Fano variety $\cB_{\nabla}$, and $X$ a complete intersection in another toric variety. If they belong to distinct chambers of the same K\"ahler moduli space, then $(\Ync, X)$ is said to be a \emph{Clifford double mirror} pair.
\end{defin}
Borisov--Li \cite{BorisovLi} provide a constructive way of identifying Clifford double mirror pairs $(\Ync, X)$. We summarize the key idea and refer to \cite{BorisovLi} for the details.\par
Let $X$ be a Calabi--Yau complete intersection specified by the toric data in \S\ref{sec:toric}. If the following relation holds in $\Z^r \oplus \mathbf{M}$,
\begin{equation*}
	\sum_{\alpha=1}^{r} \tilde{\nu}_{\alpha} = \frac{1}{2} \sum_{i=1}^{N} \tilde{\nu}_{r+i},
\end{equation*}
then there exists $\Ync$, specified by the same toric data, such that $(\Ync, X)$ is a double mirror pair. In this situation, Borisov--Li prove the existence of a derived equivalence \cite[Theorem~6.3]{BorisovLi}
\begin{equation*}
	\DC (\Ync) \xrightarrow{ \ \sim \ } D^b \mathrm{Coh} (X) .
\end{equation*}\par
    We emphasize that the hypotheses on $\Ync$ to admit a double mirror are very stringent, as will be evident below. In particular, a necessary condition for $\Ync$ to satisfy the hypotheses of \cite[Theorem~6.3]{BorisovLi} is $\hat{r} =2N$.

\subsection{GKZ systems and Picard--Fuchs equations}
\label{sec:GKZdef}
This subsection introduces the notion of GKZ systems \cite{Gelfand:1990bua}. 
The input data consists of an integral matrix $\mathsf{A}$ and a complex vector \(\beta\), called the exponent, while the output is a system of
partial differential equations on a certain affine space.
There are two collections of differential operators in a GKZ system, 
referred to as box operators and Euler operators; the box operators are determined by 
integral vectors in the lattice relation of columns of $\mathsf{A}$, 
while the Euler operators are determined by rows of $\mathsf{A}$ 
and $\beta$.\par

To keep the presentation clean, we do not aim at generality, and directly define the GKZ systems in the situation of interest in this work. See \cite{Gelfand:1994,Hosono:1995bm} for more details.\par
We first define these data for complete intersections in toric varieties in \S\ref{sec:GKZCICY}, 
and then pass to the definitions for double covers in \S\ref{sec:GKZnc}.\par

\subsubsection{GKZ systems for Calabi--Yau complete intersections}
\label{sec:GKZCICY}
In the setup of \S\ref{sec:toric}, recall the integer vectors $\tilde{\nu}_I \in \Z^{r} \oplus~\mathbf{M}$ from \eqref{eq:tildenuI}. Let
\begin{equation}
\label{eq:gkz-ci}
    \mathsf{A}_X:=
    \begin{bmatrix}
        \tilde{\nu}_{1} & \cdots & \tilde{\nu}_{r+N}
    \end{bmatrix}\in \mathrm{Mat}_{(r+n) \times (r+N) } (\Z )~\mbox{and}~\beta \in \mathbb{Q}^{r+n}.
\end{equation}
Denote by \( \{ c_{I} \}_{I=1,\ldots,r+N}\)
the coordinates on the affine space \(\mathbb{C}^{r+N}\).
Each element \(\tilde{\ell}\in\mathrm{Ker}(\mathsf{A}_X)\)
(recall our conventions around \eqref{eq:tildenutildeL})
can be regarded as an element in \(\mathbb{Z}^{r+N}\).
We can write \(\tilde{\ell}=\tilde{\ell}^{+}-\tilde{\ell}^{-}\) with 
\(\tilde{\ell}^{\pm}\in\mathbb{Z}^{r+N}_{\ge 0}\)
and \(\mathrm{supp}(\tilde{\ell}^{+})\cap \mathrm{supp}(\tilde{\ell}^{-})=\emptyset\).
\begin{defin}The GKZ system associated to these data is the set of partial differential equations by the following two types of operators:
\begin{itemize}
	\item The box operators 
	\begin{equation}
	\label{eq:GKZCICYBox}
	\Box_{\tilde{\ell}} := \prod_{\tilde{\ell}_{I}^{+}>0}\left(\frac{\partial}{\partial c_{I}}\right)^{ \tilde{\ell}_{I}^{+}} - \prod_{\tilde{\ell}_{I}^{-}>0}\left(\frac{\partial}{\partial c_{I}}\right)^{ \tilde{\ell}_{I}^{-}}~\qquad \mbox{for}~\tilde{\ell}\in\widetilde{\mathbb{L}}.
	\end{equation}
	\item The Euler operators
	\begin{equation}
	\label{eq:GKZCICYEuler}
	\Eop_j - \beta_j  := \sum_{I=1}^{r+N} [ \mathsf{A}_X]_{I,j} ~ c_{I}\frac{\partial}{\partial c_{I}} - \beta_j  , \qquad \mbox{for}~j=1, \dots, r+n .
	\end{equation}
\end{itemize}
\end{defin}
GKZ systems of this kind were also studied in \cite{Batyrev:1993V,Hosono:1995bm};
they naturally arise as a system of 
Picard--Fuchs equations for Calabi--Yau hypersurfaces or complete intersections in 
toric varieties.

In fact, we have
\begin{prop}
The \(B\)-periods for the Calabi--Yau complete intersections 
in \(\mathbb{B}_{\Delta}\) arising from
a nef partition satisfy a GKZ system
with 
\(\mathsf{A}_{X}\) given by \eqref{eq:gkz-ci} and 
\begin{equation*}
    \beta = (-1,\ldots,-1,0,\ldots,0)\in\mathbb{Q}^{r+n} ,
\end{equation*}
with the first \(r\) coordinates being \(-1\) and the rest being zero.
\end{prop}

\subsubsection{GKZ systems for Calabi--Yau double covers with gauge fixing}
\label{sec:GKZnc}
It turns out that the \(B\)-periods for Calabi--Yau double covers 
of toric varieties with gauge fixing also 
satisfy a certain type of GKZ systems which is pretty 
much similar to the previous one.
Indeed, the following proposition is proven in \cite{Lee:2020oyt}.
\begin{prop}
The \(B\)-periods for the Calabi--Yau
double covers of \(\mathbb{B}_{\Delta}\) with gauge fixing
satisfy a GKZ system with 
\(\mathsf{A}_{X}\) given by \eqref{eq:gkz-ci} and 
\begin{equation*}
    \beta = (-1/2,\ldots,-1/2,0,\ldots,0)\in\mathbb{Q}^{r+n}
\end{equation*}
with the first \(r\) coordinates being \(-1/2\) and the rest being zero.
\end{prop}

\begin{rmk}
    It can be shown that the GKZ system for Calabi--Yau double covers
    with gauge fixing is complete,
    meaning all the solutions are period integrals 
    (cf.~\cite{2023-Lee-period-CYFCI}).
\end{rmk}

\section{GLSMs, non-commutative resolutions, and \texorpdfstring{$A$}{A}-periods}
\label{sec:GLSM}
Our main tools to explore non-commutative resolutions are GLSMs \cite{Witten:1993yc}, which are 2d $\mathcal{N}=(2,2)$ gauge theories. The models of interest to us pertain to the class of so-called hybrid models \cite{Bertolini:2013xga}, discussed in relation to non-commutative resolutions in \cite{Caldararu:2010ljp} (see also e.g. \cite{Sharpe:2012ji,Sharpe:2013bwa}). Recent examples of hybrid GLSM computations include \cite{Erkinger:2020cjt,Erkinger:2022sqs,Katz:2022lyl,Katz:2023zan,Knapp:2024wwh,Knapp:2025hnf}.\par

We define the GLSM data in \S\ref{sec:GLSMgen}, and review the connection between GLSM B-branes and derived categories. In \S\ref{sec:GLSMnc} we review and generalize the construction of \cite{Lee:2023piu} for non-commutative resolution of Calabi--Yau singular double covers using GLSMs. Our prescription in \S\ref{sec:GLSMnc} yields non-commutative resolutions for Calabi--Yau singular double covers with arbitrary general nef partitions. In \S\ref{sec:GLSMHS} we spell out the relation between the hemisphere partition function \cite{Hori:2013ika,Honda:2013uca} of the GLSMs of \S\ref{sec:GLSMnc} and the $A$-period of objects in the category $\DC (\Ync)$.

\subsection{GLSMs and categories of B-branes}
\label{sec:GLSMgen}
\begin{defin}The GLSM data consist of a tuple $(\sG,\rho, \sR, W, t)$, where $\sG$ is a compact Lie group, $\rho \colon\sG \to GL(V)$ is a faithful unitary representation of $\sG$ on a complex vector space $V $, and $\sR\colon U(1)_V \to GL(V)$ is a representation of the group $U(1)_V$ that commutes with the action of $\sG$ on $V$. Moreover, $W\colon V \to \C$ is a holomorphic $\sG$-invariant polynomial, which has weight $2$ under the action of $U(1)_V$ on $V$; and $t$ is such that 
\begin{equation*}
	\exp (t) \in \mathrm{Hom} \left( \pi_1 (\sG), \C^{\ast} \right)^{\pi_0 (\sG)} .
\end{equation*}
\end{defin}
Throughout this work, we will focus on the case $\sG= U(1)^s$ for some $s \in \N$, whence we have $t \in [\R \oplus \ii \left( \R / 2\pi \Z \right)]^s$. Furthermore, we will assume that $\rho $ factors through $SL(V)$.\par
\begin{rmk}In the physics literature, $\sG$ is the gauge group, $\rho $ specifies the chiral matter fields, $U(1)_V$ is the vector R-symmetry, $W$ is the superpotential, and $t$ is the FI-$\theta$ parameter. The working assumption on $\rho $ corresponds to the requirement that the GLSM is `non-anomalous'.
\end{rmk}

\subsubsection{B-branes and derived categories}
We focus on topologically A-twisted GLSMs \cite{Witten:1988xj} on the hemisphere $\hs$, with B-type boundary conditions. Such boundary conditions at fixed $t$ and $W$ are termed B-branes, and denoted $\Br$. They form a triangulated category $\MF_{\sG} (W)$, whose objects $\Br =(\Br_{\text{\rm alg}}, \gamma_t)$ consist of the algebraic data $\Br_{\text{\rm alg}}$ \cite{Herbst:2008jq,Ballard:2016ncw} and an admissible contour $\gamma_t$ \cite{Hori:2013ika}.
\begin{defin}
The algebraic data of the B-brane $\Br$ consist of the quadruple $\Br_{\text{\rm alg}}= (M, \rho_M, \sR_M, \sT)$, where $M=M_0\oplus M_1$ is a $\Z_2$-graded, finite-rank, free $\mathrm{Sym}(V^{\vee})$-module; $\rho_{M} : \sG \to GL(M)$ is a $\Z_2$-even representation of $\sG$ on $M$; and $\sR_{M} : U(1)_V \to GL(M)$ is a $\Z_2$-even representation that commutes with $\rho_M$ and is allowed to have rational weights. Finally, $\sT \in \mathrm{End}_{\mathrm{Sym}(V^{\vee})} (M)$, called a matrix factorization of $W$, is a $\Z_2$-odd endomorphism satisfying $\sT^2 = W \cdot 1_M$.\par
Furthermore, $\rho_M, \sR_M$ must be compatible with $\rho , \sR$, in the sense that for each $\mu\in M$
\begin{equation*}
\begin{aligned}
	\sR_M (\lambda) \sT (\sR (\lambda) \mu ) \sR_M (\lambda)^{-1} &= \lambda \sT (\mu) , \qquad & \forall \lambda \in U(1)_V , \\ 
	\rho_M (g)^{-1} \sT (\rho  (g) \mu ) \rho_M (g) &= \sT (\mu) , \qquad & \forall g \in \sG . 
\end{aligned}
\end{equation*}
\end{defin}
\begin{rmk}In the physics literature, $M$ is called the Chan--Paton vector space, and the matrix factorization $\sT$ is called tachyon condensation profile.
\end{rmk}
To reduce clutter, we postpone the definition of admissible contour after we introduce the models of interest (cf.~Definition \ref{def:gammat}). The K\"ahler moduli space $\mM_{\text{\rm K}}$ of the GLSMs (or more precisely, a real projection of it) is partitioned in chambers \cite{Aspinwall:1993nu}. For $e^{-t}$ kept in the interior of one such chamber, we can neglect this piece of information and focus on $\Br_{\text{\rm alg}}$. The latter was defined more generally in the mathematics literature in \cite{Ballard:2016ncw}.\par
Assume $e^{-t}$ belongs to the interior of a chamber of $\mM_{\text{\rm K}}$. There exist a distinguished triangulated category $\DC_{t} \subset \MF_{\sG} (W)$, which only depends on the chamber, and a projector functor \cite{Herbst:2008jq,Ballard:2016ncw,HalpernLeistner:2014}
\begin{equation}
\label{eq:projfunctor}
	\pi_{t} \colon \MF_{\sG} (W) \longrightarrow \DC_{t} .
\end{equation}
The specifics of $\DC_{t} $ vary with the GLSM under consideration and with the choice of chamber: the prototypical example is $\DC_{t} \cong D^b \mathrm{Coh} (X)$ for a smooth Calabi--Yau variety $X$; in the rest, we will be interested in $\DC_{t}\cong \DC (\Ync)$.
\begin{defin}\label{def:piprojfunct}
We will henceforth denote by $\mathfrak{C} \subset \mM_{\text{\rm K}}$ be the chamber containing an open neighborhood of $e^{-t}=0$ (we set our models so that $e^{-t}=0$ corresponds to the non-commutative resolution). For $e^{-t} \in \mathfrak{C}$, the projection functor \eqref{eq:projfunctor} will be denoted simply 
\begin{equation*}
	\pi \colon \MF_{\sG} (W) \longrightarrow \DC_{t} \vert_{e^{-t} \in \mathfrak{C}} .
\end{equation*}
\end{defin}

\subsection{GLSMs and non-commutative resolutions}
\label{sec:GLSMnc}
The setup is as in \S\ref{sec:CYpreliminaries}. We put $\cB\equiv \cB_{\nabla}$ to lighten the expressions, and denote the weight of its $i^{\text{th}}$ homogeneous coordinate \cite{Cox:1993fz} by $\underline{\theta}^{(i)}$ (cf.~\S\ref{sec:latticeL}).\par 
Consider a $\Z_2$-gerbe $\widetilde{\cB}$ over the toric variety $\cB$ \cite{Pantev:2005zs}. Denoting $\left\{ \phi_i \right\}_{i=1, \dots, N}$ the lift to $\widetilde{\cB}$ of the homogeneous coordinates of $\cB$, these have weights $2\theta^{(i)}_a$ under $\C^{\ast}_a \subset (\C^{\ast})^s$.\par 
Throughout this section, we allow for an arbitrary nef partition 
\begin{equation}
\label{eq:nefgeneral}
    -2 K_{\cB} = \sum_{\alpha=1}^{\hat{r}} F_{\alpha} ,
\end{equation}
$1 \le \hat{r} \le 2N$ (possibly $\hat{r}<N$).\par
The line bundles $F_{\alpha} \longrightarrow \cB$ appearing in \eqref{eq:nefgeneral} lift to line bundles $F_{\alpha} \longrightarrow \widetilde{\cB}$. 
Denoting $z_{\alpha}$ the fibre coordinate on $F_{\alpha}^{-1} \longrightarrow \widetilde{\cB}$, we have that 
\begin{equation*}
	z_{\alpha}^2 s_{\alpha} (\phi) \in \mathrm{H}^0 (\mathcal{O}_{F_{\alpha}^{-1}}) ,
\end{equation*}
where we are denoting with the same symbol the lift to $\widetilde{\cB}$ of the sections \eqref{eq:branchingsec}. We define 
\begin{equation}\label{eq:bundletilde}
	\mV := \mathrm{Tot} \left( \bigoplus_{\alpha=1}^{\hat{r}} F_{\alpha}^{-1}  \longrightarrow \widetilde{\cB} \right) .
\end{equation}
The superpotential $W$ is \cite{Lee:2023piu}
\begin{equation*}
	W= \sum_{\alpha=1}^{\hat{r}} z_{\alpha}^2 s_{\alpha} (\phi)\in \mathrm{H}^{0}(\mathcal{O}_{\mathcal{V}}).
\end{equation*}
This construction provides a GLSM description of the non-commutative resolution of the singular double covers of \S\ref{sec:doublecover} \cite{Caldararu:2010ljp} and their generalization from \eqref{eq:nefgeneral}. Fields and charges of the GLSMs are summarized in \textsc{Table} \ref{tab:GLSMdoublenew}.\par

\begin{table}[th]
\centering
\begin{tabular}{l |c |c | c}
	field & \# & $U(1)_a$ gauge charge & $U(1)_V$ R-charge \\
	\hline
	$\phi_i$ & $i=1, \dots, N$ & $2 \theta_a ^{(i)}$ & $4\varepsilon$ \\
	$z_{\alpha}$ & $\alpha=1, \dots, \hat{r}$ & $- d_a^{(\alpha)}$ & $1-2\varepsilon$\\
	\hline
\end{tabular}
\caption{Field content and charges of the GLSM for non-commutative resolutions of Calabi--Yau double covers. Here $\varepsilon\in \left(0,\frac{1}{2}\right)$.}
\label{tab:GLSMdoublenew}
\end{table}\par

\begin{rmk} The Calabi--Yau double covers of \S\ref{sec:LLY} are recovered by setting $\hat{r}=r+N$ and imposing \eqref{eq:gaugefixing}. The GLSMs just constructed need not be subject to this constraint.
 \end{rmk}

\subsubsection{Sheaves of modules of Clifford algebras}
The construction of $\mV$ in \eqref{eq:bundletilde} is equivalently presented as an orbi-bundle \cite{Adem:2007zz} over the algebraic stack
\begin{equation*}
	\bigoplus_{\alpha=1}^{\hat{r}} \mathcal{O}_{\cB} \left( - \frac{d_1^{(\alpha)}}{2} , \dots,  - \frac{d_s^{(\alpha)}}{2} \right) \ \longrightarrow \ [ \cB / \Z_2 ] .
\end{equation*}
In this setup, it is possible to define a category of matrix coherent factorizations \cite{buchweitz1987maximal,Tu:2010nh,ballard2012resolutions,orlov2004triangulated,efimovcoherent}, denoted $\MF_{\Z_2} (\mV, W)$. Let $ 0 \in \mathfrak{C} \subset \mM_{\text{\rm K}}$ as in Definition \ref{def:piprojfunct}. Then, our conventions are such that 
\begin{equation*}
	\left. \DC_t \right\rvert_{e^{-t} \in \mathfrak{C}} = \MF_{\Z_2} (\mV, W) .
\end{equation*}
To relate this construction with the category $\DC (\Ync)$ defined in \S\ref{sec:NCmath}, we use a result of \cite{Guo:2021aqj}\footnote{See also \cite{Caldararu:2010ljp,Addington:2012zv}, \cite[\S2.4]{Katz:2022lyl}, \cite[\S4]{Katz:2023zan} for related works.}. It is shown therein that there exist 
\begin{enumerate}[(i)]
    \item a sheaf of Clifford algebras $\mA_0$ on $\cB$, and 
    \item a derived equivalence 
    \begin{equation*}
	    \MF_{\Z_2} (\mV, W)  \ \xrightarrow{ \ \sim \ } \ D \left( \cB, \mA_0 \sharp \Z_2 \right) .
    \end{equation*}
\end{enumerate} 
The right-hand side is the derived category of sheaves of $(\mA_0 \sharp \Z_2)$-modules on $\cB$, where the symbol $\sharp$ stands for the smash product. In particular, $\mA_0 \sharp \Z_2$ is a sheaf of even parts of Clifford algebras, providing an explicit realization of $\mC_0$ in \S\ref{sec:NCmath}.\par

\subsection{\texorpdfstring{$A$}{A}-periods of non-commutative resolutions}
\label{sec:GLSMHS}
We now introduce the $A$-periods, which are realized starting with the $\hs$ partition function of the GLSMs \cite{Hori:2013ika}. We restrict our attention to the GLSMs of \S\ref{sec:GLSMnc}, and refer to \cite{hori2019notes} for a thorough overview of hemisphere partition functions.

\subsubsection{Admissible contours}
To define $\gamma_t$, we denote by $\mathfrak{t}_{\C}$ the complexified Cartan subalgebra of $\sG$, with coordinates $\sigma =(\sigma_1, \dots, \sigma_s)$, and let $\mathcal{H}$ be the infinite collection of hyperplanes in $\mathfrak{t}_{\mathbb{C}}$
\begin{equation*}
    \mathcal{H} = \left( \bigcup_{i=1}^{N} \mathcal{H}^{(i)} \right) \cup \left( \bigcup_{\alpha=1}^{\hat{r}}  \tilde{\mathcal{H}}^{(\alpha)} \right) , 
\end{equation*}
where
\begin{equation*}
\begin{aligned}
    \mathcal{H}^{(i)} & := \bigcup_{n^{(i)} \ge 0} \left\{ \langle \underline{\theta}^{(i)}, \sigma \rangle =\ii \left( \frac{n^{(i)}}{2} +\varepsilon\right) \right\} , \\
    \tilde{\mathcal{H}}^{(\alpha)} & :=\bigcup_{n^{(\alpha)} \ge 0} \left\{ \langle \underline{d}^{(\alpha)}, \sigma \rangle = -\ii \left( n^{(\alpha)} + \frac{1}{2}-\varepsilon\right) \right\} .
\end{aligned}
\end{equation*}
Additionally, we define a function $\widetilde{W}_{\text{\rm eff}} : \mathfrak{t}_{\C} \longrightarrow \C $, called effective twisted superpotential, as
\begin{equation*}
\begin{aligned}
	\widetilde{W}_{\text{\rm eff}} (\sigma) = - \langle t, \sigma\rangle  - 2 & \sum_{i=1}^{N} \langle \underline{\theta}^{(i)}, \sigma \rangle \left[ \log \left( \ii 2 \langle \underline{\theta}^{(i)}, \sigma \rangle  \right) -1 \right]  + \sum_{\alpha=1}^{\hat{r}} \langle \underline{d}^{(\alpha)}, \sigma \rangle \left[ \log \left( \ii \langle \underline{d}^{(\alpha)}, \sigma \rangle  \right) -1 \right] .
\end{aligned}
\end{equation*}
\begin{defin}\label{def:gammat}
An admissible contour $\gamma_t$ is a middle-dimensional, Weyl-invariant\footnote{In the present work, since we will only consider $\sG$ abelian, the Weyl group is trivial. For nonabelian $\sG$ this is a delicate condition, and in examples it can also be the case that for strongly coupled phases, Weyl-invariance is not compatible with absolute convergence \cite{EHKR}, making necessary a modification of the definition of admissibility.} cycle in $\mathfrak{t}_{\C} \setminus \mathcal{H} $ such that:
\begin{enumerate}[(i)]
	\item It is a smooth deformation of the real cycle $\left\{ \Im (\sigma)=0, \sigma \in \mathfrak{t}_{\C} \right\}$;
	\item $\widetilde{W}_{\text{\rm eff}} (\sigma) + \ii 2 \pi q (\sigma) \to + \infty$ in all asymptotic directions of $\gamma_t$, for every $q \in \mathfrak{t}^{\vee}_{\C} \cap \Z$.
\end{enumerate}
\end{defin}
The last ingredient we need to define is the brane factor $f_{\Br} : \mathfrak{t}_{\C} \to \C$, for each $\Br \in \MF_{\sG} (W)$, 
\begin{equation*}
	f_{\Br} (\sigma) = \tr_M \left( \sR_M \left( e^{\ii \pi} \right) \rho_M \left( e^{2 \pi \sigma} \right)\right) .
\end{equation*}
It may be written in the form
\begin{equation}
\label{eq:fBsum}
	f_{\Br} (\sigma) = \sum_{k \in \mathfrak{N}_{\Br}} (-1)^{\mathsf{r}_{k}}e^{2 \pi \langle q_k, \sigma \rangle} ,
\end{equation}
where $\mathsf{r}_{k} \in \Z$, $q_k \in \mathfrak{t}^{\vee}$ are integral points, and $\mathfrak{N}_{\Br} \subset \Z$ is a finite set of integers, depending on $\Br$. The finiteness of $\mathfrak{N}_{\Br}$ is guaranteed by fact that $M$ has finite rank, or $\Br$ is quasi-isomorphic to a matrix factorization of finite rank.

\subsubsection{\texorpdfstring{$A$}{A}-periods from GLSMs}
Consider a GLSM as in \S\ref{sec:GLSMnc}. 
\begin{defin}
    For every object $\Br \in \MF_{\sG} (W)$ we define its \emph{partition function} to be
\begin{equation}
\label{eq:Znc}
	\mz_{\Ync} \left( \Br; t \right) = \int_{\gamma_t} \dd^s \sigma~e^{\ii \langle t, \sigma\rangle} ~\prod_{i=1}^{N} \Gamma \left( 2\varepsilon + \ii 2 \langle \underline{\theta}^{(i)}, \sigma \rangle \right) \prod_{\alpha=1}^{\hat{r}} \Gamma \left( \frac{1}{2} - \varepsilon - \ii \langle\underline{d}^{(\alpha)}, \sigma \rangle \right) ~f_{\Br} (\sigma) .
\end{equation}
\end{defin}
Our conventions have been conveniently set up in such a way that the GLSM realizes $\Ync$ in the chamber $\mathfrak{C}$ containing $e^{-t}=0$. In this chamber it is possible to take $\varepsilon \to 0^{+}$ without loss of generality.\par
\begin{defin}
    Introduce redundant variables $\{ x_I\}_{I=1,\dots , \hat{r}+N}$ satisfying 
    \begin{equation}
    \label{eq:xtotnc}
        \prod_{\alpha=1}^{\hat{r}} (-x_I )^{-d^{(\alpha)}_a}\prod_{i=1}^{N} (-x_{\hat{r}+i} )^{\theta^{(i)}_a} = e^{-t_a - \log (4) \sum_{i=1}^{N}\theta^{(i)}_a } , \qquad a=1, \dots, s .
    \end{equation}
    For every $\Br \in \MF_{\sG} (W)$ we define the $A$-period of $\Br$ to be 
    \begin{equation}\label{eq:Anc}
        \pz_{\Ync} \left( \Br ; x \right) := \left(\prod_{\alpha=1}^{\hat{r}} \left. \frac{1}{\sqrt{x_{\alpha}}} \right)\mz_{\Ync} \left( \Br; t \right) \right\rvert_{\text{\eqref{eq:xtotnc}}}
    \end{equation}
    with the right-hand side given in terms of \eqref{eq:Znc} at $\varepsilon \to 0^{+}$, and subject to the replacement \eqref{eq:xtotnc}.
\end{defin}

\begin{rmk}\label{rem:K0}
    The partition functions \eqref{eq:Znc}, and hence the $A$-periods \eqref{eq:Anc}, depend on $\Br \in \MF_{\sG} (W)$ only through the function $f_{\Br}$. In particular, within the chamber $e^{-t} \in \mathfrak{C} \subset \mM_{\text{\rm K}}$, 
    the $A$-periods depend only on the class $[\pi (\Br)] \in K_0 (\Ync)$, where $\pi$ is given in Definition \ref{def:piprojfunct} and $K_0 (\Ync)$ is the Grothendieck group of $\DC (\Ync)$.
    (See, for instance, \cite{Neeman} for definitions and properties of the Grothendieck group of a triangulated category.)
\end{rmk}

\section{Non-commutative resolutions and sheaves on Calabi--Yau complete intersections}
\label{sec:maingeom}

Throughout the current section we impose the assumption in \S\ref{sec:LLY}; only
Calabi--Yau double covers with gauge fixing will be treated. 
In particular, we demand $\hat{r}=r+N$ and impose the condition \eqref{eq:gaugefixing} on the branching locus. This will be relaxed in the next section.

\subsection{From non-commutative resolutions to complete intersections}
We conjecture that for every $\Ync$ constructed above there exists a Calabi--Yau complete intersection $X$, with $\dim (X)= \dim (Y)$, such that the $A$-periods of $\Ync$ and the $A$-periods of $X$ are annihilated by the same GKZ system.\par
The rationale behind Conjecture \ref{conj1} is explained in \S\ref{sec:prequotient}, and strong evidence will be provided in \S\ref{sec:mainGLSM}.

\begin{cj}\label{conj1}
    Fix $e^{-t} \in \mathfrak{C} \subset \mM_{\text{\rm K}}$ and define $\tnc := t + 4\log (2) \left( \sum_{i=1}^{N} \underline{\theta}^{(i)} \right)$.
    There exist a smooth Calabi--Yau complete intersection $X$, together with an isomorphism of abelian groups
    \begin{equation*}
        \Phi \colon K_0 \left( \Ync \right) \longrightarrow K_0 \left( X \right)
    \end{equation*}
    and an overall constant $c \in \mathbb{C}$, such that the equality 
    \begin{equation*}
        \mz_{\Ync} (\Br; t) = c \mz_{X} (\mE ; \tnc) 
    \end{equation*}
    holds for all $\Br$ and for every $\mE \in D^b \mathrm{Coh} (X)$ with 
    \begin{equation*} 
        \Phi \left( [\pi (\Br)] \right) = [\mE].
    \end{equation*}\par
    In particular, there exists a GKZ system annihilating all the $A$-periods $\pz_{\Ync}$ and $\pz_{X}$.
\end{cj}
Some explanations are in order. 
In the statement of the conjecture, $ K_0 \left( \Ync \right):= K_0 \left( \DC (\Ync) \right)$ and $K_0 \left(X \right) := K_0 \left( D^b \mathrm{Coh} (X) \right)$. $[ \cdot ]$ means the image of an object of a triangulated category in the corresponding Grothendieck group, and $\pi (\Br)$ is as in Definition \ref{def:piprojfunct}.
By Remark \ref{rem:K0}, we have that $\mz_{\Ync} (\Br; t)$ is only sensitive to $[\pi (\Br)]$ in the Grothendieck group $K_0 (\Ync)$; likewise $\mz_{X} (\mE ; \tnc)$ is only sensitive to $[\mE] \in K_0 (X)$. The relation between the various categories entering in the statement of Conjecture \ref{conj1} is summarized in the diagram:
\begin{equation*}
    \begin{tikzcd}
        \MF_{\sG} (W) \arrow[d,"\pi"] & \MF_{\sG} (W_X) \arrow[d] \\
        \DC (\Ync) \arrow[d] & D^b \mathrm{Coh} (X) \arrow[d] \\
        K_0(\Ync) \arrow[r, "\Phi"] & K_0(X)
    \end{tikzcd}
\end{equation*}
with the upper vertical arrows given by \eqref{eq:projfunctor} and the lower vertical arrows take the zeroth K-group of a triangulated category.\par
\begin{rmk}
 Generically, we \emph{do not} expect $\Phi$ to lift to a derived equivalence between $\DC (\Ync)$ and $D^b \mathrm{Coh} (X)$.
\end{rmk}

Next, we shall give a geometric construction of a Calabi--Yau complete intersection $\Xpre$, which we conjecture to be the $X$ we seek.

\subsection{Pre-quotient of Calabi--Yau double covers}
\label{sec:prequotient}
For brevity, we put \(\mathbb{B}\equiv \mathbb{B}_{\nabla}\) in the following discussion.
Let us retain the notation in \S\ref{sec:latticeL}. The
toric variety \(\mathbb{B}\) can be represented by a quotient
\begin{equation*}
    \mathbb{B}\cong \mathbb{C}^{N}\sslash_{\chi} (\mathbb{C}^{\ast})^{s}
\end{equation*}
where \(\chi\) is a character corresponding to an ample divisor on \(\mathbb{B}\).

\subsubsection{Pre-quotient construction}

Note that the sections \(s_{1},\ldots,s_{r}\) determine a \((\mathbb{C}^{\ast})^{s}\)-equivariant map
\begin{equation*}
\begin{aligned}
    \mathbb{C}^{N} &\to \mathbb{C}^{r+N}\\
    (z_{1},\ldots,z_{N})&\mapsto  (s_{1}(z),\ldots,s_{r}(z),z_{1},\ldots,z_{N}) ,
\end{aligned}
\end{equation*}
and hence a morphism on their GIT quotients
\begin{equation*}
    \mathbb{B}\cong \mathbb{C}^{N}\sslash_{\chi} (\mathbb{C}^{\ast})^{s}\to \mathbb{C}^{r+N}\sslash_{\chi} (\mathbb{C}^{\ast})^{s}.
\end{equation*}
Here \((\mathbb{C}^{\ast})^{s}\) acts on \(\mathbb{C}^{r+N}\) through
the weights \((\underline{d}^{(1)}, \dots, \underline{d}^{(r)},\underline{\theta}^{(1)}, \dots, \underline{\theta}^{(N)})\).
One can check that this is an injection.

According to \cite[Proposition 14.3.10]{Cox:2011tor}, the GIT quotient
\begin{equation*}
    \mathbb{P}_{(\vec{\underline{d}}, \vec{\underline{\theta}})}:=
    \mathbb{C}^{r+N}\sslash_{\chi} (\mathbb{C}^{\ast})^{s}
\end{equation*}
is a projective toric variety. 

To construct the double covers, we use the ``coordinate squaring map.''
More specifically, denote by \(z_{\alpha}\) (recall
that \(\alpha=1,\ldots,r+N\)) the homogeneous coordinates
of \(\mathbb{P}_{(\vec{\underline{d}}, \vec{\underline{\theta}})}\).
Consider the map
\begin{equation*}
    \mathbb{C}^{r+N}\to \mathbb{C}^{r+N}~\mbox{by}~z_{\alpha}\mapsto z_{\alpha}^{2}.
\end{equation*}
This is obviously \((\mathbb{C}^{\ast})^{s}\) equivariant and yields a finite Galois cover
\begin{equation}
\label{eq:squaring}
    \mathbb{P}_{(\vec{\underline{d}}, \vec{\underline{\theta}})}
    \to \mathbb{P}_{(\vec{\underline{d}}, \vec{\underline{\theta}})}.
\end{equation}
The Galois group of the covering
\eqref{eq:squaring} is given by a quotient 
\((\mathbb{Z}_{2})^{r+N}\slash \mathsf{K}\), where \(\mathsf{K}\) is the subgroup
consisting of elements in \((\mathbb{Z}_{2})^{r+N}\)
that can be realized by an element in \((\mathbb{C}^{\ast})^{s}\).
Here we identify the action of the $\alpha^{\text{th}}$ $\mathbb{Z}_{2}$ 
with multiplying by \(-1\) on \(z_{\alpha}\), i.e.~\(z_{\alpha}\mapsto -z_{\alpha}\).

Consider the fiber product
\begin{equation*}
    \begin{tikzcd}
        &\Xpre\ar[r]\ar[d] &\mathbb{P}_{(\vec{\underline{d}}, \vec{\underline{\theta}})}\ar[d]\\
        &\mathbb{B}\ar[r]&\mathbb{P}_{(\vec{\underline{d}}, \vec{\underline{\theta}})}
    \end{tikzcd}
\end{equation*}
and the right vertical map is the map given in \eqref{eq:squaring}.
Then \(\Xpre\) is a Calabi--Yau complete intersection
of degree \((2\underline{d}^{(1)}, \dots, 2\underline{d}^{(r)})\) in 
\(\mathbb{P}_{(\vec{\underline{d}}, \vec{\underline{\theta}})}\).

Note that the summation map 
\begin{equation*}
    (\mathbb{Z}_{2})^{r+N}\to \mathbb{Z}_{2},~(\omega_{\alpha})\mapsto 
    \sum_{\alpha=1}^{r+N} \omega_{\alpha}
\end{equation*}
descends to the quotient
\begin{equation*}
    (\mathbb{Z}_{2})^{r+N}\slash \mathsf{K}\to \mathbb{Z}_{2}.
\end{equation*}
Let \(\Gamma\) be the kernel of this map. Then \(\Xpre\slash \Gamma\)
is a double cover over \(\mathbb{B}\), and
by the following lemma, the branching locus of
\(\Xpre\slash\Gamma\to \mathbb{B}\) is
equal to the branching locus of the Calabi--Yau double cover
\(Y\to \mathbb{B}\). Due to the absence of torsion
line bundles, we have \(Y\cong \Xpre\slash\Gamma\).
\begin{lem}
    The branching locus of \(\Xpre\slash\Gamma\to \mathbb{B}\)
    is equal to the branching locus of the Calabi--Yau double cover
    \(Y\to \mathbb{B}\).
\end{lem}
\begin{proof}
    It suffices to check that the branching locus of
    \begin{equation*}
        \Psi\colon\mathbb{P}_{(\vec{\underline{d}}, \vec{\underline{\theta}})}
        \slash \Gamma \to 
        \mathbb{P}_{(\vec{\underline{d}}, \vec{\underline{\theta}})}
    \end{equation*}
    is equal to the union of toric divisors.

    To facilitate the argument, we use homogeneous coordinates.
    For a point \(p\in (\mathbb{C}^{r+N})^{\mathsf{s}}_{\chi}\), we denote by 
    \([p]\) the class of \(p\) under the \((\mathbb{C}^{\ast})^{s}\) action. Now suppose that \(p=(z_{1},\ldots,z_{r+N})\in \mathbb{C}^{r+N}\) represents a point in the union of all toric divisors in $\mathbb{P}_{(\vec{\underline{d}}, \vec{\underline{\theta}})}$, i.e.~\(z_{\alpha}=0\) for some \(\alpha\). We need to check that if 
    \begin{equation*}
        \Psi([(z_{1},\ldots,z_{r+N})])=\Psi([(y_{1},\ldots,y_{r+N})])
    \end{equation*}
    then there is an element \(g\in\Gamma\) such that
    \begin{equation*}
        g\cdot [(z_{1},\ldots,z_{r+N})]=[(y_{1},\ldots,y_{r+N})].
    \end{equation*}
    It is clear that there exists \(h\in (\mathbb{Z}_{2})^{r+N}\slash \mathsf{K}\)
    such that 
    \begin{equation*}
        h\cdot [(z_{1},\ldots,z_{r+N})]=[(y_{1},\ldots,y_{r+N})].
    \end{equation*}
    If \(h\in \Gamma\), then we are done. If \(h\notin \Gamma\), 
    we can choose \(k=(k_{1},\ldots,k_{r+N})\in \mathbb{Z}_{2}^{r+N}\)
    with \(k_{\alpha}=1\) and \(k_{\beta}=0\) for \(\beta\ne\alpha\).
    Then \(h\cdot\bar{k}\in \Gamma\) (here \(\bar{k}\) is the image of \(k\)
    in the quotient \((\mathbb{Z}_{2})^{r+N}\slash \mathsf{K}\)) and
    \begin{equation*}
        (h\cdot\bar{k})\cdot [(z_{1},\ldots,z_{r+N})]=
        h\cdot [(z_{1},\ldots,z_{r+N})]=[(y_{1},\ldots,y_{r+N})]
    \end{equation*}
    and the proof is completed.
\end{proof}

\begin{rmk}
    The toric variety \( \mathbb{P}_{(\vec{\underline{d}}, \vec{\underline{\theta}})}\) constructed above may be very singular, even non-Gorenstein. Here is an example. Consider \(\mathbb{B}=\mathbb{P}^{3}\) and the nef-partition
    \begin{equation*}
        -2K_{\cP^3} = 3H + H +H+H+H+H ;
    \end{equation*}
    namely we consider a double cover of \(\mathbb{P}^{3}\)
    whose branching divisor is a union of five hyperplanes and a cubic in general position.
    Then the toric variety \(\mathbb{P}_{(\vec{\underline{d}}, \vec{\underline{\theta}})}\)
    is the weighted projective space \(\mathbb{P}^{5}_{(3,1^5)}\).
    Moreover, 
    \begin{equation*}
        \Xpre\cong\mathbb{P}^{5}_{(3,1^5)}[2,6] ,
    \end{equation*}
    which is a smooth Calabi--Yau threefold.
\end{rmk}

\subsubsection{Pre-quotient and $A$-periods}
The divisor $\sum_{\alpha=1}^{r+N} F_{\alpha}$ is a simple normal crossing divisor, which guarantees that $Y$ is a normal variety with at worst quotient singularities. 
From our previous discussion, we have the following.
\begin{prop}
    Let notations be as above.
    There exist a projective variety $\Xpre$ and a finite group $\Gamma$ such that 
    \begin{equation*}
        Y \cong \Xpre / \Gamma 
    \end{equation*}
    as coarse moduli spaces.
\end{prop}

\begin{cj}\label{conj:prequotient}
\(\Xpre\) is smooth and $\Xpre=X$.
\end{cj}

\begin{rmk}
    One expects that the untwisted part of
    the genus-zero orbifold Gromov--Witten invariants of \(Y\)
    can be computed from the genus-zero Gromov--Witten invariants of \(\Xpre\)
    and the knowledge of the group action \(\Gamma\). 
    Computing genus-zero Gromov--Witten invariants of \(\Xpre\)
    is rather standard; one could compute genus-zero (orbifold) GW invariants of 
    \(\mathbb{P}_{(\vec{\underline{d}}, \vec{\underline{\theta}})}\)
    first \cite{Coates_Corti_Iritani_Tseng_2015} and then apply the quantum hyperplane section theorem
    \cite{2010Tseng}.

    We shall also point out that the toric variety 
    \begin{equation*}
        \mathbb{P}_{(\vec{\underline{d}}, \vec{\underline{\theta}})}\slash \Gamma
    \end{equation*}
    is also an orbifold. However, the subvariety
    \begin{equation*}
        \Xpre\slash\Gamma\subset 
        \mathbb{P}_{(\vec{\underline{d}}, \vec{\underline{\theta}})}\slash \Gamma
    \end{equation*}
    is not a complete intersection in general.
\end{rmk}

Clearly, it is possible to deform both immersions away from this special point in the complex structure moduli space, and repeat the argument. Then $\cP_{\Delta}$ embeds as the intersection of $r$ degree-$\underline{d}^{(\alpha)}$ hypersurfaces, and $X$ is cut out by generic algebraic functions.

\subsection{Examples}
\label{sec:mathEx}

We discuss two families of examples:
\begin{itemize}
	\item Double covers of weighted projective spaces, $Y \xrightarrow{ \ 2:1 \ } \cP_{\vec{w}}^{N-1}$;
	\item Double covers of products of projective spaces, $Y \xrightarrow{ \ 2:1 \ } \cP^{\kappa r}\times \cP^{r-1}$.
\end{itemize}
The results in these cases are summarized in \textsc{Table} \ref{tab:list}.

\begin{table}[th]
\centering
\begin{tabular}{|c |c |c | c|}
	\hline
	\textsc{base} & \textsc{nef partition} & $X$ & \textsc{double mirror} \\
	\hline
	$\cP_{\vec{\theta}}^{N-1}$ & $d^{(1)}H+\cdots+ d^{(r)}H$ & $\cP^{r+N-1}_{(\vec{d}, \vec{\theta})} [2d^{(1)}, \dots, 2d^{(r)}]$ & $\begin{matrix}\text{\small yes if }{\scriptstyle \vec{d}=\vec{\theta}=(1^N) } \\ \text{\small no otherwise}\end{matrix}$\\
	$\cP^{\kappa r}\times \cP^{r-1}$ & $\left( \begin{matrix}d^{(1)}_1 H_1+\cdots+ d^{(r)}_1 H_1 \\ d^{(1)}_2 H_2+\cdots+ d^{(r)}_2 H_2\end{matrix} \right)$ & Bl$_{\cP^{\kappa r}} \cP^{n} \left[ \begin{matrix}d^{(1)}_1, \cdots ,d^{(r)}_1 \\ d^{(1)}_2,\cdots ,d^{(r)}_2\end{matrix} \right]$ &  {\small yes} \\
	\hline
\end{tabular}
\caption{For the double cover $Y \ \xrightarrow{ \ 2:1 \ } \ \cB$ we specify the toric base $\cB$ and the part of the nef partition not fixed by \eqref{eq:gaugefixing}. The last column tells whether $\Ync$ and $X$ are double mirrors.}
\label{tab:list}
\end{table}\par

\subsubsection{Double cover of weighted projective spaces}
We consider a Calabi-Yau double cover 
\begin{equation*}
	Y \ \xrightarrow{ \ 2:1 \ } \ \cP^{N-1}_{\vec{\theta}} 
\end{equation*}
of a weighted projective space of weights $\vec{\theta}=(\theta^{(1)}, \dots, \theta^{(N-1)},1) \in \N^{N}$. As in e.g. \cite{Hosono:1993qy}, we take at least one entry of $\vec{\theta}$ equal to 1. In this example, $s=1$, and the nef partition with $\hat{r}=r+N$ is specified by $(d^{(1)},\dots, d^{(r)})$. $X$ is a Calabi--Yau complete intersection of $r$ hypersurfaces of even degree in a (different) weighted projective space $\cP^{r+N-1}_{(\vec{d}, \vec{\theta})}$, shown in the first line of \textsc{Table} \ref{tab:list}.\par
The most basic example in which we can prove Conjecture \ref{conj:prequotient} falls in this setup. Consider $N=1$ and nef partition $2=1+1$, leading to 
\begin{equation*}
    \cP^0 = \C /\!\!/ \C^{\ast} \cong \mathrm{pt} , \qquad X= \cP^{1}[2] .
\end{equation*}
Then $X$ consists of two (generically disjoint) points, which is a Calabi--Yau double cover of $\cP^0$; the group $\Gamma$ is trivial in this example.\par
Furthermore, for singular threefolds, we have the following
\begin{prop}
    Consider a singular double cover $Y \ \xrightarrow{ \ 2:1 \ } \ \cP^3$ branched over a divisor satisfying the hypotheses of \S\ref{sec:LLY}. Conjecture \ref{conj:prequotient} holds in these cases.
\end{prop}

\subsubsection{Double covers of projective spaces branched over lines}
A special instance of the example above is the double cover of $\cP^{N-1}$ with branching locus $D_{\mathrm{br}}$ consisting of $2N$ hyperplanes, 
\begin{equation*}
	D_{\mathrm{br}}= \underbrace{H+ \cdots +H}_{2N}.
\end{equation*}
From this, we get a pair $(\Ync,X)$, where $X=\cP^{2N-1}[2, \dots, 2]$.\par
This setup was considered in \cite{Caldararu:2010ljp,Sharpe:2012ji}, and it provides an example of the double mirror phenomenon \cite{BorisovLi}. It follows from a result of Kuznetsov \cite{Kuznetsov:08} that there exists an equivalence 
\begin{equation*}
	\DC (\Ync) \xrightarrow{ \ \sim \ } D^b \mathrm{Coh}\left( \cP^{2N-1}[2, \dots, 2] \right) .
\end{equation*}\par
We have the following result, which shows that double mirror pairs are a non-generic phenomenon.
\begin{prop}
	Let $\Ync$ be the non-commutative resolution of the singular double cover $Y \ \xrightarrow{ \ 2:1 \ } \ \cP^{N-1}_{\vec{\theta}} $ branched over $N$ hyperplanes and $r$ hypersurfaces of degree $\vec{d}=(d^{(1)}, \dots, d^{(r)})$. It admits a Clifford double mirror Calabi--Yau complete intersection in a toric variety if and only if $\vec{\theta}=(1,\dots,1)$ and $\vec{d}=(1, \dots, 1)$.
\end{prop}
\begin{proof}
	One direction is known. 
	The converse implication is shown by direct computation. Starting on the complete intersection side of the correspondence, one checks that a weighted projective space fails to satisfy the assumptions of \cite{BorisovLi} unless all weights are equal to 1.
\end{proof}\par

\subsubsection{Double cover of products of projective spaces}
For $\Ync$ the non-commutative resolution of the Calabi--Yau double cover 
\begin{equation*}
	Y \ \xrightarrow{ \ 2:1 \ } \ \cP^{\kappa r} \times \cP^{r-1} ,
\end{equation*}
our prescription yields $X$ a smooth Calabi--Yau complete intersection in the blow-up $\mathrm{Bl}_{\cP^{\kappa r}} \cP^{(\kappa +2)r-1}$ of the projective space $ \cP^{(\kappa +2)r -1}$ along $\cP^{\kappa r}$.\par
We now proceed to prove a more general statement, that implies that the pair $(\Ync, X)$ just described is a double mirror pair.\par
Let $s ,r \in \N_{\ge 2}$ with $2r >s$, and $\left\{ \kappa_a \in \N , \ a=1, \dots, s -1\right\}$. We discuss the pair $(\Ync, X)$, where the former is the non-commutative resolution of a singular double cover of a product of $\cP^{\kappa_a r}$ and $\cP^{r-1}$; and $X$ is a certain blow-up of $\cP^n$ along $s-1$ loci.\par
To lighten the notation, we write $m_a := \kappa_a r + 1$ for all $a=1, \dots, s-1$ and $m_s:=2r-s$. Besides, we define $n:= \sum_{a=1}^{s} m_a$ and the set of indices $\mathfrak{M}_a$ of cardinality $\lvert \mathfrak{M}_a \rvert =m_a$ as 
\begin{align*}
	\mathfrak{M}_1 &:= \left\{ 1, \dots, m_1 \right\} , \\
	\mathfrak{M}_2 &:= \left\{ m_1+1, \dots, m_1+m_2 \right\} , \\
	\vdots & \\
	\mathfrak{M}_s &:= \left\{ n-m_s+1, \dots, n \right\} .
\end{align*}
Furthermore, we introduce the following lattice vectors in $\Z^n$:
\begin{align*}
	\nu_i &= (\dots, 0,\underbrace{1}_{i^{\text{th}}} ,0, \dots ) , \qquad i=1, \dots, n , \\
	\nu_{n+1} &= (-1,-1 \dots, -1) \\
	\nu_{n+2} &= (\underbrace{1, \dots, 1}_{m_1}, 0, \dots , 0) \\
	\nu_{n+3} &= (\underbrace{0, \dots, 0}_{m_1},\underbrace{1, \dots, 1}_{m_2}, 0, \dots , 0) \\
	\vdots & \\
	\nu_{n+s} &= (0, \dots, 0,\underbrace{1, \dots, 1}_{m_{s-1}}, \underbrace{0, \dots , 0}_{m_s}) \\
\end{align*}
and embed them into $\Z^r \oplus \Z^n$ as in \eqref{eq:tildenuI}. By construction, 
\begin{equation}
\label{eq:BLblowupnu}
	\sum_{i \in \mathfrak{M}_a} \nu_{i} = \nu_{n+1+a} , \qquad \forall a=1, \dots, s 
\end{equation}
and moreover 
\begin{equation}
\label{eq:doublemirrorcond}
	\sum_{i=n+s+1-2r}^{n+s} \nu_i = \sum_{i=n-m_s+1}^{n} \nu_i  + \nu_{n+1} + \sum_{i=n+2}^{n+s} \nu_i =0 .
\end{equation}\par
In practice, nef partition on is characterized as a way of writing 
\begin{equation*}
	\left\{ n+1, \dots, n+s \right\} \cup \bigsqcup_{a=1}^{s} \mathfrak{M}_a = \bigsqcup_{\alpha=1}^{r} \mathcal{I}_{\alpha} .
\end{equation*}
We require the indices are distributed as follows. Let us start with the collection of indices $\mathfrak{M}_s \sqcup \left\{ n+1, \dots, n+s \right\} $, which contains $m_s +s=2r$ elements; we distribute them pairwise into each set $\mathcal{I}_{\alpha}$. Next, for each $a=1, \dots, s-1$, it is possible to distribute the $\kappa_a r+1$ indices in each $\mathfrak{M}_a$ according to the rule: $\kappa_a+1$ indices go into the set $\mathcal{I}_{\alpha}$ that contains $(n+1+a)$; and into each one of the remaining $(r-1)$ sets $\mathcal{I}_{\alpha} $ for which $(n+1+a) \notin \mathcal{I}_{\alpha}$, we place exactly $\kappa_a$ indices.\par 
In this way, $\sum_{\alpha=1}^r \lvert \mathcal{I}_{\alpha}\rvert = 2r + \sum_{a=1}^{s-1} (\kappa_a r+1)$ by construction, and the Calabi--Yau condition $n+s=\sum_{\alpha=1}^r \lvert \mathcal{I}_{\alpha}\rvert$ simplifies into 
\begin{equation}
\label{eq:CYBLkappa}
	n+1= r \left( 2+\sum_{a=1}^{s-1} \kappa_a \right) .
\end{equation}
The collection of coefficients $\left\{ \kappa_a , \ a=1, \dots, s -1\right\}$ is required to satisfy \eqref{eq:CYBLkappa}, for fixed $r,n,s \in \N$.\par
From this derivation we have
\begin{prop}
	The above construction describes a Calabi--Yau $X$, which is the complete intersection of $r$ hypersurfaces in the blow-up \(\widetilde{\mathbb{P}}^{n}\) of $\cP^n$ along the loci $\cP^{\kappa_a r}$ for $a=1, \dots, s-1$. The $\alpha^{\text{th}}$ hypersurface is defined as the zero-locus of a section $f_{\alpha}$ of degree 
	\begin{equation*}
		d_a^{(\alpha)} = \begin{cases} \kappa_a & \text{ if } 1 \le a \le s-1 \text{ and } n+1+a \notin \mathcal{I}_{\alpha} \\ \kappa_a +1 & \text{ if } 1 \le a \le s-1 \text{ and } n+1+a \in \mathcal{I}_{\alpha} \\ 2 & \text{ if } a=s \end{cases}
	\end{equation*}
	under the action of $\C^{\ast}_a \subset (\C^{\ast})^s$, where the $\C^{\ast}_a$-action on the blow-up is induced by the torus action on the $a^{\text{th}}$ exceptional divisor for $a=1,\dots, s-1$, and by the torus action on the base $\cP^n$ for $a=s$.
\end{prop}
We introduce a Calabi--Yau double cover of the product of projective spaces
\begin{equation*}
	Y \ \xrightarrow{ \ 2:1 \ } \ \cP^{\kappa_1 r} \times \cdots \times \cP^{\kappa_{s-1} r} \times \cP^{r-1} ,
\end{equation*}
with branching locus determined by the same nef partition as above, and lines. Let $\Ync$ be the non-commutative resolution of this double cover.
\begin{prop}
Let $X$ and $\Ync$ be as above. They form a Clifford double mirror pair.
\end{prop}
\begin{proof}
	It suffices to note that the lattice vector 
	\begin{equation*}
		\mathrm{deg}^{\vee} =  (\underbrace{1, \dots, 1}_{r}; \underbrace{0, \dots, 0}_{n}) \in \Z^r \oplus \Z^n
	\end{equation*}
	can be written in two different ways, as 
	\begin{equation*}
		\sum_{\alpha=1}^{r} \tilde{\nu}_{\alpha} = \mathrm{deg}^{\vee} = \frac{1}{2}\sum_{i=n+s+1-2r}^{n+s} \tilde{\nu}_{r+i} .
	\end{equation*}
	The $\Z^n$-part of the second equality is simply \eqref{eq:doublemirrorcond}, whereas the $\Z^r$-part follows by construction, since the $2r$ indices $\left\{ n+s+1-2r, \dots, n+s \right\}$ are distributed pairwise in the $\mathcal{I}_{\alpha}$.\par
	At this stage we pass to the lattice 
	\begin{equation*}
		\overline{\mathbf{N}} := (\Z^r \oplus \Z^n )/ \left(  \sum_{i=n+s+1-2r}^{n+s} \Z \tilde{\nu}_{r+i} + \Z \mathrm{deg}^{\vee} \right) .
	\end{equation*}
	It is generated by the vectors $\left\{\tilde{\nu}_{I}, \ I=1, \dots, n+1+s-r \right\}$, subject to the relations 
	\begin{align*}
		\sum_{\alpha=1}^{r}\tilde{\nu}_{\alpha} & \equiv 0  \\
		\sum_{i \in \mathfrak{M}_a} \tilde{\nu}_{r+i} & =  \tilde{\nu}_{n+1+r+a} + \kappa_a \mathrm{deg}^{\vee} \equiv 0 
	\end{align*}
	where, in the relations in the second line, we have used \eqref{eq:BLblowupnu} and the non-generic choice of nef partition.\par
	From these relations together with the fact that $\mathfrak{M}_a \cap \mathfrak{M}_b = \emptyset$ if $b \ne a$, we deduce that the $r$ vectors $\left\{ \tilde{\nu}_{\alpha}, \ \alpha=1, \dots, r \right\}$ produce the toric diagram of $\cP^{r-1}$, and the $m_a$ vectors $\left\{\tilde{\nu}_{r+i} , \ i \in \mathfrak{M}_a \right\}$ produce the toric diagram of $\cP^{m_a-1} = \cP^{\kappa_a r}$, $\forall a=1, \dots, s-1$.\par
	To establish the derived equivalence between $X$ and $\Ync$, we apply \cite{BorisovLi}.
\end{proof}
\begin{rmk}This example generalizes a result from \cite{Calabrese:2014}, which is recovered setting $s=2,r=2,\kappa=1$. In this specialization, $Y \xrightarrow{ \ 2:1 \ }\cP^2 \times \cP^1$ is branched over a bi-degree $(3,2)$ hypersurface and lines, and $X = \left\{ f_1 =0=f_2 \right\} \subset \mathrm{Bl}_{\cP^2} \cP^5$, for two sections $f_1,f_2 \in \Gamma \left( \varpi^{\ast} \mathcal{O}_{\cP^5} (-3) (-S) \right)$, where $\mathrm{Bl}_{\cP^2} \cP^5 \xrightarrow{ \ \varpi \ } \cP^5$ denotes the blow-down map and $S$ is the exceptional divisor. This particular example was also considered in \cite[\S9.3]{BorisovLi} and \cite[\S6.2]{Katz:2022lyl} as an instance of the double mirror phenomenon.
\end{rmk}

\section{GLSM proposal}
\label{sec:mainGLSM}

The aim of this section is to substantiate Conjecture \ref{conj1} using the techniques of \S\ref{sec:GLSMHS}. 
The main result of this section is Theorem \ref{thm:smoothX}, which states that for each $\Ync$ there exists a Calabi--Yau 
complete intersection $X$ whose $A$-periods are annihilated by the same GKZ system.\par
This result holds for any choice of branching divisor, and reduces to the setup of Conjecture \ref{conj1} 
when it is chosen to satisfy the additional hypotheses of \S\ref{sec:LLY}. The derivation is based on the GLSM realization of the $A$-periods under consideration.\par 
Explicit examples are spelled out in \S\ref{sec:GLSMEx}.

We begin by recalling the result on the GKZ systems and non-commutative resolutions, when the branching locus is chosen according to \S\ref{sec:LLY}.
\begin{prop}[{\cite[Theorem~4.12]{Lee:2023piu}}]
\label{prop:GLSMAYnc}
	Assume $\hat{r}=N+r$ with $0 < r \le N$, and that the branching locus of the double cover $Y \longrightarrow \cB_{\nabla}$ satisfies \eqref{eq:gaugefixing}. For every $\Br \in \mathrm{Obj} \left( \DC (\Ync) \right)$, consider the $A$-period $\pz_{\Ync} (\Br; x)$ in \eqref{eq:Anc}. It satisfies the GKZ system \eqref{eq:GKZCICYBox}-\eqref{eq:GKZCICYEuler} with  
	\begin{equation*}
		\beta = \left( \begin{matrix} -1/2 \\ \vdots \\ -1/2 \\ 0 \\ \vdots \\ 0\end{matrix} \right) \in \mathbb{Q}^{\hat{r}+n} .
	\end{equation*}
\end{prop}
\begin{proof}This was proven in \cite[\S4.3]{Lee:2023piu} by direct computation (with $x_{i,j}$ therein being $x_{\hat{r}+i}$ here, and $x_{\alpha,0}$ therein identified with $x_{\alpha}$ here for $\alpha=1, \dots, r$). In the present statement, we introduce redundant parameters $x_{r+1}, \dots, x_{r+N}$ associated to the fixed part \eqref{eq:gaugefixing}, and work in the larger lattice $\Z^{\hat{r}} \oplus \mathbf{M}$, as explained in \S\ref{sec:doublecover}. We omit the proof of the extension of the result to include these redundant parameters, as it will be recovered later as a particular case.
\end{proof}

We generalize the statement of Proposition \ref{prop:GLSMAYnc} by relaxing the condition on the branching locus, and claim the following.\par
	Consider a branched double cover $Y \longrightarrow \cB_{\nabla}$ of a toric Fano variety, and let $\Ync$ be its non-commutative resolution, as in \S\ref{sec:GLSMnc}.
	There exists a smooth Calabi--Yau complete intersection $X$, such that the $A$-periods of $\Ync$ and $X$ satisfy the same GKZ system, upon a certain identification of variables.\par
The proof, together with a constructive prescription to read off $X$, is in \S\ref{sec:GLSMGKZderivation}. 

\begin{rmk}The existence of such a smooth Calabi--Yau complete intersection $X$ was observed in examples in \cite{Katz:2023zan}, for double covers of $\cP^3$. $X$ was dubbed the `smooth cousin' of $Y$. Here we show that the `smooth cousin' is a universal phenomenon; furthermore, we provide the expression for the smooth cousin of the most general branched double covers of toric Fano varieties.
\end{rmk}

\subsection{GKZ systems, non-commutative resolutions, and smooth complete intersections}
\label{sec:GLSMGKZderivation}
We study the partition function \eqref{eq:Znc}, which we rewrite here:
\begin{equation}
\label{eq:ZYncGLSM}
	\mz_{\Ync} \left( \Br; t \right) = \int_{\gamma_t} \dd^s \sigma~e^{\ii \langle t, \sigma\rangle } ~\prod_{i=1}^{N} \Gamma \left( \ii 2 \langle \underline{\theta}^{(i)}, \sigma \rangle \right) \prod_{\alpha=1}^{\hat{r}} \Gamma \left( \frac{1}{2} - \ii \langle\underline{d}^{(\alpha)}, \sigma \rangle \right) ~f_{\Br} (\sigma) .
\end{equation}
As explained in \S\ref{sec:GLSMHS}, we have set $\varepsilon \to 0^{+}$ without loss of generality, and the admissible contour $\gamma_t$ is chosen to run along the real locus in $\mathfrak{t}_{\C}$, but avoiding $\sigma^{a}=0$ for all $a=1, \dots, s$.\par
The corresponding $A$-period $\pz_{\Ync} \left( \Br; x \right)$ is defined in \eqref{eq:Anc}.\par
\begin{rmk}
    Imposing the condition in \S\ref{sec:LLY}, the integrand becomes 
\begin{equation*}
\begin{aligned}
    & \prod_{i=1}^{N} \Gamma \left( \ii 2 \langle \underline{\theta}^{(i)}, \sigma \rangle \right) \prod_{\alpha=1}^{\hat{r}} \Gamma \left( \frac{1}{2} - \ii \langle\underline{d}^{(\alpha)}, \sigma \rangle \right) \\ 
    \stackrel{\text{gauge fixed}}{=} \ & \prod_{i=1}^{N} \Gamma \left( \ii 2 \langle \underline{\theta}^{(i)}, \sigma \rangle \right) \Gamma \left( \frac{1}{2} - \ii \langle\underline{\theta}^{(i)}, \sigma \rangle \right) \prod_{\alpha=1}^{r} \Gamma \left( \frac{1}{2} - \ii \langle\underline{d}^{(\alpha)}, \sigma \rangle \right) \\
    =& \left(\frac{\sqrt{\pi}}{2} \right)^N \prod_{i=1}^{N} 2^{ \ii 2 \langle \underline{\theta}^{(i)}, \sigma \rangle} \frac{\Gamma \left( \ii \langle \underline{\theta}^{(i)}, \sigma \rangle \right)}{\cosh  \left( \pi \langle\underline{\theta}^{(i)}, \sigma \rangle \right)} \prod_{\alpha=1}^{r} \Gamma \left( \frac{1}{2} - \ii \langle\underline{d}^{(\alpha)}, \sigma \rangle \right) ,
\end{aligned}
\end{equation*}
where we have used the properties of the Gamma function to pass to the second line. Throughout this section, we \emph{do not} impose these simplifications, allowing for the most general non-commutative resolutions. This not only generalizes the product of Gamma functions in the integrand, but also the functions $f_{\Br} (\sigma)$.
\end{rmk}

\subsubsection{Integral manipulations of the \texorpdfstring{$A$}{A}-period}
\label{subsec:manipulations}
To reduce clutter, we introduce the shorthand notation
\begin{equation}
\label{eq:tncshift}
	\tnc := t + 2 \log (2) \left( \sum_{\alpha=1}^{\hat{r}} \underline{d}^{(\alpha)} \right) = t + 4\log (2) \left( \sum_{i=1}^{N} \underline{\theta}^{(i)} \right)
\end{equation}
where the second equality stems from the Calabi--Yau condition.\par
We define the set $\mathcal{J} \subseteq \left\{ 1, \dots, N \right\}$ to be the collection of indices for which $\underline{\theta}^{(i)}= \underline{d}^{(\alpha)}$, for some $\alpha \in \left\{1, \dots, \hat{r} \right\}$, without repetitions. That is, one starts with $\underline{\theta}^{(i)}$ and scans through $\alpha \in \left\{1, \dots, \hat{r} \right\}$ to check whether $\exists \alpha_i$ such that $\underline{\theta}^{(i)}= \underline{d}^{(\alpha_i)}$. If so, then $\{ i\} \subseteq \mathcal{J}$, and for all $i^{\prime} \ne i$, one scans through $\alpha \in \left\{1, \dots, \hat{r} \right\} \setminus \alpha_i$. This guarantees that each $\underline{d}^{(\alpha)}$ is paired with at most one $\underline{\theta}^{(i)}$ for $i \in \mathcal{J}$.\footnote{Under the hypotheses of \S\ref{sec:LLY}, we have $\mathcal{J}=\left\{ 1, \dots, N \right\}$, but this is not necessarily the case in general, and $\mathcal{J}$ is any subset of $\left\{ 1, \dots, N \right\}$, possibly empty.}\par

\begin{lem}
Let $\mathcal{J} \subseteq \left\{ 1, \dots, N \right\}$ defined as above. There exist two sets $\mathfrak{H},\mathfrak{W} \subseteq \left\{ 1, \dots, \hat{r} \right\}$ with $\lvert \mathfrak{H}\rvert =\lvert \mathfrak{W}\rvert - \lvert \mathcal{J} \rvert $ such that \eqref{eq:ZYncGLSM} equals 
\begin{equation}
\begin{aligned}
	\mz_{\Ync} \left( \Br; t \right) &= \pi^{\frac{\hat{r}}{2}}  \int_{\gamma_t} \dd^s \sigma~ e^{\ii \langle \tnc, \sigma\rangle } ~ \prod_{i \in \mathcal{J} } \left( \frac{\pi}{\sin \left( 2 \pi \ii \langle \underline{\theta}^{(i)}, \sigma \rangle\right)} \right) ~f_{\Br} (\sigma) \\
		&\times \prod_{\substack{ 1 \le i \le N \\ i \notin \mathcal{J} }} \Gamma \left( \ii 2 \langle \underline{\theta}^{(i)}, \sigma \rangle \right) \prod_{\alpha \in \mathfrak{W}} \frac{1}{\Gamma \left( 1 - \ii \langle\underline{d}^{(\alpha)}, \sigma \rangle \right)} \prod_{\alpha\in \mathfrak{H} } \Gamma \left( 1 - \ii 2 \langle\underline{d}^{(\alpha)}, \sigma \rangle \right) .
\end{aligned}
\label{eq:ZYnclemma1}
\end{equation}
\end{lem}
\begin{proof}
Begin with \eqref{eq:ZYncGLSM}. To every Gamma function in the integrand with half-integer shift in the argument, we apply the Legendre duplication formula, in the form:
\begin{align}\label{eq:dupliGamma}
	\Gamma \left( \frac{1}{2}- \ii z \right) &= \sqrt{\pi} e^{2 \ii z \log (2) } \frac{  \Gamma \left( 1 - \ii 2 z \right) }{   \Gamma \left( 1 - \ii z  \right)} .
\end{align}
We thus rewrite 
\begin{equation}
\label{eq:ZYncduplic}
	\mz_{\Ync} \left( \Br; t \right) = \pi^{\frac{\hat{r}}{2}} \int_{\gamma_t} \dd^s \sigma~e^{\ii \langle \tnc, \sigma\rangle } ~\prod_{i=1}^{N} \Gamma \left( \ii 2 \langle \underline{\theta}^{(i)}, \sigma \rangle \right) \prod_{\alpha=1}^{\hat{r}} \frac{\Gamma \left( 1 - \ii 2 \langle\underline{d}^{(\alpha)}, \sigma \rangle \right) }{\Gamma \left( 1 - \ii \langle\underline{d}^{(\alpha)}, \sigma \rangle \right)}   ~f_{\Br} (\sigma) ,
\end{equation}
where the exponential terms produced by \eqref{eq:dupliGamma} have been reabsorbed in $\tnc$ by the definition \eqref{eq:tncshift}.\par
To further simplify the expression, define the two sets of integers:
\begin{itemize}
\item $\mathcal{J} \subseteq \left\{ 1, \dots, N \right\}$ the collection of indices defined above;
\item $\mathfrak{J} \subseteq \left\{ 1, \dots, \hat{r} \right\}$ the collection of indices $\alpha$ for which $\exists i \in \mathcal{J}$ such that $\underline{d}^{(\alpha)} = \underline{\theta}^{(i)}$.
\end{itemize}
With this notation, the product of Gamma functions in the numerator is 
\begin{align*}
	&\prod_{i=1}^{N} \Gamma \left( \ii 2 \langle \underline{\theta}^{(i)}, \sigma \rangle \right) \prod_{\alpha =1}^{\hat{r}} \Gamma \left( 1 - \ii 2 \langle\underline{d}^{(\alpha)}, \sigma \rangle \right) \\
 = &\prod_{\substack{ 1 \le i \le N \\ i \notin \mathcal{J} }} \Gamma \left( \ii 2 \langle \underline{\theta}^{(i)}, \sigma \rangle \right) \prod_{\substack{ 1 \le \alpha \le \hat{r} \\ \alpha \notin \mathfrak{J}} } \Gamma \left( 1 - \ii 2 \langle\underline{d}^{(\alpha)}, \sigma \rangle \right) ~ \prod_{i \in \mathcal{J} } \left[ \Gamma \left( \ii 2 \langle \underline{\theta}^{(i)}, \sigma \rangle \right) \Gamma \left( 1 - \ii 2 \langle\underline{\theta}^{(i)}, \sigma \rangle \right) \right] \\
	= &\prod_{\substack{ 1 \le i \le N \\ i \notin \mathcal{J} }} \Gamma \left( \ii 2 \langle \underline{\theta}^{(i)}, \sigma \rangle \right)\prod_{\substack{ 1 \le \alpha \le \hat{r} \\ \alpha \notin \mathfrak{J}} } \Gamma \left( 1 - \ii 2 \langle\underline{d}^{(\alpha)}, \sigma \rangle \right) ~ \prod_{i \in \mathcal{J} } \frac{\pi}{\sin \left( 2 \pi \ii \langle \underline{\theta}^{(i)}, \sigma \rangle\right)} .
\end{align*}
Thus, the product of Gamma functions in the integrand of \eqref{eq:ZYncduplic} becomes 
\begin{equation*}
\begin{aligned}
    & \prod_{i=1}^{N} \Gamma \left( \ii 2 \langle \underline{\theta}^{(i)}, \sigma \rangle \right) \prod_{\alpha=1}^{\hat{r}} \frac{\Gamma \left( 1 - \ii 2 \langle\underline{d}^{(\alpha)}, \sigma \rangle \right) }{\Gamma \left( 1 - \ii \langle\underline{d}^{(\alpha)}, \sigma \rangle \right)}  \\
    & \qquad \qquad = \prod_{\substack{ 1 \le i \le N \\ i \notin \mathcal{J} }} \Gamma \left( \ii 2 \langle \underline{\theta}^{(i)}, \sigma \rangle \right) ~ \prod_{i \in \mathcal{J} } \frac{\pi}{\sin \left( 2 \pi \ii \langle \underline{\theta}^{(i)}, \sigma \rangle\right)} \\
    & \qquad \qquad \times \prod_{\substack{ 1 \le \alpha \le \hat{r} \\ \alpha \notin \mathfrak{J}} } \Gamma \left( 1 - \ii 2 \langle\underline{d}^{(\alpha)}, \sigma \rangle \right) ~ \prod_{\alpha=1}^{\hat{r}} \frac{1}{\Gamma \left( 1 - \ii \langle\underline{d}^{(\alpha)}, \sigma \rangle \right)} .
\end{aligned}
\end{equation*}
It may happen that there exist pairs $(\alpha, \beta)$, with $\alpha \in \left\{1, \dots, \hat{r} \right\}\setminus \mathfrak{J}$ and $\beta \in \left\{1, \dots, \hat{r} \right\}$, for which $2d^{(\alpha)}=d^{(\beta)}$. If this is the case, it leads to further simplifications in the ratio of Gamma functions. After canceling out such pairs, we have 
\begin{equation*}
\begin{aligned}
	&\prod_{\substack{ 1 \le \alpha \le \hat{r} \\ \alpha \notin \mathfrak{J}}} \Gamma \left( 1 - \ii 2 \langle\underline{d}^{(\alpha)}, \sigma \rangle \right) \prod_{\alpha=1}^{\hat{r}}\frac{1}{\Gamma \left( 1 - \ii \langle\underline{d}^{(\alpha)}, \sigma \rangle \right)} \\
	& \qquad \qquad = \prod_{\alpha \in \mathfrak{H}} \Gamma \left( 1 - \ii 2 \langle\underline{d}^{(\alpha)}, \sigma \rangle \right) \prod_{\alpha \in \mathfrak{W}} \frac{1}{\Gamma \left( 1 - \ii \langle\underline{d}^{(\alpha)}, \sigma \rangle \right)}
\end{aligned}
\end{equation*}
for two sets of indices
\begin{equation*}
	\mathfrak{H},\mathfrak{W} \subseteq \left\{ 1, \dots, \hat{r} \right\} , \qquad 1 \le  \lvert \mathfrak{W} \rvert \le \hat{r} , \quad \lvert \mathfrak{H}\rvert =  \lvert \mathfrak{W} \rvert \setminus \lvert \mathfrak{J} \rvert .
\end{equation*}
Explicitly:
\begin{equation*}
\begin{aligned}
    \mathfrak{H} & =  \left( \{1, \dots, \hat{r} \} \setminus \mathfrak{J} \right) \setminus \left\{ \alpha \ : \ 2\underline{d}^{(\alpha)} = \underline{d}^{(\beta)} \ \text{ for some } 1 \le \beta \le \hat{r} \right\} , \\
    \mathfrak{W} & = \{1, \dots, \hat{r} \}  \setminus  \left\{ \beta \ : \ 2\underline{d}^{(\alpha)} = \underline{d}^{(\beta)} \ \text{ for some } 1 \le \alpha \le \hat{r}, \ \alpha \notin \mathfrak{J} \right\}  ,
\end{aligned}
\end{equation*}
counted without repetitions. Putting the pieces together, we arrive at 
\begin{align*}
	& \prod_{i=1}^{N} \Gamma \left( \ii 2 \langle \underline{\theta}^{(i)}, \sigma \rangle \right) \prod_{\alpha=1}^{\hat{r}} \frac{\Gamma \left( 1 - \ii 2 \langle\underline{d}^{(\alpha)}, \sigma \rangle \right) }{\Gamma \left( 1 - \ii \langle\underline{d}^{(\alpha)}, \sigma \rangle \right)} \\
 & \qquad \qquad = \prod_{\substack{ 1 \le i \le N \\ i \notin \mathcal{J} }} \Gamma \left( \ii 2 \langle \underline{\theta}^{(i)}, \sigma \rangle \right) \prod_{\alpha \in \mathfrak{W}} \frac{1}{\Gamma \left( 1 - \ii \langle\underline{d}^{(\alpha)}, \sigma \rangle \right)}\\
	& \qquad \qquad \times \prod_{\alpha\in \mathfrak{H} } \Gamma \left( 1 - \ii 2 \langle\underline{d}^{(\alpha)}, \sigma \rangle \right)  \prod_{i \in \mathcal{J} } \frac{\pi}{\sin \left( 2 \pi \ii \langle \underline{\theta}^{(i)}, \sigma \rangle\right)} .
\end{align*}
Plugging this equality inside the integral \eqref{eq:ZYncduplic} proves the lemma.
\end{proof}

\subsubsection{Explicit Calabi--Yau complete intersection} We provide a prescription to read off an explicit Calabi--Yau complete intersection geometry whose $A$-periods are annihilated by the same GKZ system that annihilates the $A$-periods of $\Ync$.\par
To state the main result of this section, we introduce some notation.
\begin{defin}\label{notation:Xsc}
    Let the sets $\mathcal{J}, \mathfrak{W}, \mathfrak{H}$ be as in \S\ref{subsec:manipulations}. 
    \begin{itemize}
        \item Take the indexed list of vectors
            \begin{equation}
            \label{eq:weightXgen}
    		  \left\{ 2\theta_a^{(i)} , \ i \in \left\{ 1, \dots , N \right\} \setminus \mathcal{J} \right\} \sqcup \left\{  d_a^{(\alpha)} , \ \alpha \in \mathfrak{W} \right\} , \qquad \forall a=1, \dots, s  ;
	        \end{equation}
            and define an ordering on it. We introduce an index $\iota =1, \dots, N + \lvert \mathfrak{W} \rvert - \lvert \mathcal{J} \rvert$. The $\iota^{\text{th}}$ entry of the list will be denoted $\ell_a ^{(\iota)}$, and will belong to the $\iota^{\text{th}}$ new weight vector. Using $\lvert \mathfrak{H}\rvert =\lvert \mathfrak{W}\rvert - \lvert \mathcal{J} \rvert $, there is a total of $N+\lvert \mathfrak{H}\rvert$ weight vectors $\underline{\ell}^{(\iota)}$.
        \item The rescaled degrees 
            \begin{equation}
            \label{eq:degreeXgen}
	   	       2d_a^{(\alpha)} ,  \qquad \alpha \in \mathfrak{H} , \ a=1, \dots, s ,
	        \end{equation}
            will characterize the new nef partitions.
    \end{itemize}
\end{defin}
Furthermore, we consider a map
\begin{equation*}
\begin{aligned}
    \C^{\hat{r}+N} \ & \to \ \C^{N+2 \lvert \mathfrak{H}\rvert} \\
    x_I \ & \mapsto \ c_I ,
\end{aligned}
\end{equation*}
satisfying 
\begin{equation}
\label{eq:xImapstocI}
\begin{aligned}
    x_{\alpha} & \mapsto c_{\alpha}^2 \qquad & \forall \alpha \in \mathfrak{H} , \\
    x_{\alpha} & \mapsto 1 \qquad & \forall \alpha \notin \mathfrak{H} , \\
    x_{\hat{r}+i} & \mapsto c_{2 \lvert \mathfrak{H}\rvert +i } ^2 \qquad & \forall i \notin \mathcal{J} ,
\end{aligned}
\end{equation}
and sending the remaining coordinates $\{ x_{\hat{r}+i}\}_{i \in \mathcal{J}}$ to products of $c_{\lvert \mathfrak{H}\rvert + \iota}$ consistently with 
\begin{equation}
\label{eq:xtocsign}
    \prod_{i=1}^{N} \left( - \frac{1}{4} x_{\hat{r}+i} \right)^{\theta_a^{(i)}} \mapsto \prod_{\iota} ( \epsilon_{\iota} c_{\lvert \mathfrak{H}\rvert + \iota } )^{\ell_a^{(\iota)}} ,
\end{equation}
where the sign $\epsilon_{\iota}=+1$ if the weight $\underline{\ell}^{(\iota)}$ comes from $ \underline{d}^{(\alpha)}$, and $\epsilon_{\iota}=-1$ if it comes from $2\underline{\theta}^{(i)}$.

\begin{thm}\label{thm:smoothX}
The notation is as in Definition \ref{notation:Xsc}. Let $X$ be the Calabi--Yau complete intersection in a toric variety determined, via the quotient construction reviewed in \S\ref{sec:latticeL}, by 
        \begin{itemize}
            \item the weight vectors $\underline{\ell}^{(\iota)}$ from \eqref{eq:weightXgen};
            \item the nef partition specified by $2\underline{d}^{(\alpha)}$ from \eqref{eq:degreeXgen}.
        \end{itemize}
Furthermore, let the moduli $\zeta_a$ of $X$ be related to the moduli $e^{-t_a}$ of $\Ync$ through 
\begin{equation*}
	\zeta_a= \exp \left\{ - t_a  - 4 \log (2) \sum_{i=1}^{N} \theta_a^{(i)} \right\} .
\end{equation*}
Consider the GKZ system \eqref{eq:GKZCICYBox}-\eqref{eq:GKZCICYEuler} associated to such $X$, and 
\begin{equation*}
\beta_X = \left( \begin{matrix} -1 \\ \vdots \\ -1 \\ 0 \\ \vdots \\ 0\end{matrix} \right).
\end{equation*}
Then
\begin{enumerate}[(1)]
    \item The $A$-periods of $X$ are annihilated by this GKZ system.
    \item The $A$-periods of $\Ync$ are annihilated by the same GKZ system, after the map \eqref{eq:xImapstocI}.
\end{enumerate}
\end{thm}
We note that, if $D_{\mathrm{br}}$ is smooth, the variety $X$ output by Theorem \ref{thm:smoothX} is deformation-equivalent to $Y$.\par
    Inserting the replacement \eqref{eq:xImapstocI} in $\pz_{\Ync} (\Br; x)$, we obtain the $A$-period 
    \begin{equation*}
        \pz_{\Ync} (\Br; x (c) ) =  \left( \prod_{\alpha \in \mathfrak{H}} \frac{1}{c_{\alpha}} \right) \mz_{\Ync} (\Br; t) 
    \end{equation*}
    with $e^{-t}$ in the right-hand side replaced in terms of the variables $\{ c_I \}$ associated to $X$. The prefactor is such that the GKZ system in the variables $x$ has $\beta_{\Ync} = \left( \frac{1}{2} , \cdots , \frac{1}{2} , 0 \cdots , 0\right)^{\top}$, while the GKZ system in the variables $c$ has $\beta_X = \left( 1 , \cdots , 1 , 0 \cdots , 0\right)^{\top}$. Observe that $\beta_{\Ync} $ and $\beta_{X}$ in general belong to lattices of different dimensions; the map \eqref{eq:xImapstocI} automatically implements the transformation of the $\beta$.\par
    
We now state a technical result, from which Theorem \ref{thm:smoothX} follows. The proof of the claim is given in \S\ref{app:GKZlemma}.
\begin{lem}\label{lemma:GKZlemma}
	Let $X$ be a Calabi--Yau complete intersection in a toric variety. 
	\begin{enumerate}[(1)]
    \item If $\underline{\ell} \in \mathfrak{t}_{\C}^{\vee} \cap \Z^s$ is a weight of the toric variety, the replacement 
		\begin{equation*}
			 \Gamma \left( \ii \langle \underline{\ell}, \sigma \rangle \right) \mapsto \frac{1}{ \Gamma \left( 1- \ii \langle \underline{\ell}, \sigma \rangle \right)} 
		\end{equation*}
		in the integrand of the $A$-period leads to the same GKZ system, but modifies the identification of parameters shifting $t \mapsto t+\ii \pi \underline{\ell}$ on the right-hand side of \eqref{eq:CICYGKZvar}.
		\item For any $\underline{\ell} \in \mathfrak{t}_{\C}^{\vee} \cap \Z^s$, inserting $\frac{\pi}{\sin \left( \ii \pi \langle\underline{\ell}, \sigma \rangle\right)}$ in the integrand of the $A$-period leads to the same GKZ system.
	\end{enumerate}
\end{lem}
\begin{proof}It is clear that neither modification affects the Euler operators. For the analysis of the box operators, see \S\ref{app:GKZlemma}.\end{proof}
\begin{proof}[Proof of Theorem \ref{thm:smoothX}] 
\begin{enumerate}[(1)]
    \item The first statement is unrelated to double covers, and we defer its proof to \S\ref{app:GKZCICY}, cf. Theorem \ref{thm:CICYGKZ} therein.
    \item Applying Lemma \ref{lemma:GKZlemma} to \eqref{eq:ZYnclemma1}, there exists $\Br^{\prime} \in \MF_{U(1)^s} (W)$, with $\pi_{\ast} \Br^{\prime}  \in D^b \mathrm{Coh} (X)$, such that 
\begin{equation*}
\begin{aligned}
	\mz_{X} \left( \Br^{\prime}; t \right) &= \int_{\gamma_t} \dd^s \sigma~ e^{\ii \langle \tnc , \sigma\rangle } ~f_{\Br^{\prime}} \left( \sigma \right) \\
		\times & \prod_{\substack{ 1 \le i \le N \\ i \notin \mathcal{J} }} \Gamma \left( \ii \langle 2 \underline{\theta}^{(i)}, \sigma \rangle \right) \prod_{\alpha \in \mathfrak{W}} \Gamma \left( \ii \langle \underline{d}^{(\alpha)}, \sigma \rangle \right) \prod_{\alpha\in \mathfrak{H} } \Gamma \left( 1 - \ii \langle  2 \underline{d}^{(\alpha)}, \sigma \rangle \right) 
\end{aligned}
\end{equation*}
leads to an $A$-period annihilated by the same GKZ system. Note that the change of sign due to point (1) of Lemma \ref{lemma:GKZlemma} is accounted for by the sign $\epsilon_{\iota}$ in \eqref{eq:xtocsign}.\par
With the notation of Definition \ref{notation:Xsc}, the latter integral reads
\begin{equation*}
\begin{aligned}
	\mz_{X} \left( \Br^{\prime}; t \right) &= \int_{\gamma_t} \dd^s \sigma~ e^{\ii \langle \tnc , \sigma\rangle } ~\prod_{\iota =1}^{N + \lvert \mathfrak{W} \rvert - \lvert \mathcal{J} \rvert} \Gamma \left( \ii \langle \underline{\ell}^{(\iota)}, \sigma \rangle \right)  \prod_{\alpha\in \mathfrak{H} } \Gamma \left( 1 - \ii \langle 2\underline{d}^{(\alpha)}, \sigma \rangle \right)  ~f_{\Br^{\prime}} \left( \sigma \right) .
\end{aligned}
\end{equation*}
This expression equals the partition function of the Calabi--Yau complete intersection $X$ described in the statement of the theorem.
\end{enumerate}
\end{proof}

\begin{rmk}\label{rem:divide2}
    For certain situations, expression \eqref{eq:ZYnclemma1} may be further reduced. This is the case when, for a given $a$, there exist an integer $k_a \in \N$ such that 
\begin{equation}
\label{eq:kdividesld}
	k_a \vert 2\theta_a^{(i)} \quad \forall i =1, \dots, N , \qquad \text{ and } \qquad k_a \vert d_a^{(\alpha)} \quad \forall \alpha=1, \dots, \hat{r} .
\end{equation}
Typically, one demands that the greatest common divisor of all $\theta_a ^{(i)}$ for given $a$ is $1$. Therefore, scenario \eqref{eq:kdividesld} would only be possible with $k_a \in \left\{1,2\right\} $, and $k_a=2$ only if $d_a^{(\alpha)} \in 2 \N$ for all $\alpha =1, \dots, \hat{r}$. Therefore, for every $a=1, \dots, s$, one can define
	\begin{equation*}
		k_a := \begin{cases} 2 & \text{ if } d_a^{(\alpha)} \in 2 \N \quad \forall \alpha=1, \dots, \hat{r} \\ 1 & \text{ otherwise} \end{cases}
	\end{equation*}
    and rescale 
    \begin{equation*}
        2 \theta^{(i)}_a \mapsto \frac{2}{k_a} \theta^{(i)}_a , \qquad d^{(\alpha)}_a \mapsto \frac{1}{k_a} d^{(\alpha)}_a .
    \end{equation*}
    If $\widetilde{X}$ and $X_2$ are, respectively, the Calabi--Yau before and after rescaling, they yield the same GKZ system. However, the relation between the variables $\{c_I\}$ and $\{\zeta_a\}$ in $X_2$ would differ from the one for a given $X$ that does not arise from rescaling, since the former is inherited from $\widetilde{X}$.
\end{rmk}

To conclude, we cast the result of Theorem \ref{thm:smoothX} into the notation used in \S\ref{intro:results}. We have the indexed lists 
\begin{equation*}
\Theta:=\{ \underline{\theta}^{(1)},\ldots,\underline{\theta}^{(N)}\},\qquad \mathcal{D}:=\{\underline{d}^{(1)},\ldots,\underline{d}^{(\hat{r})}\} ,
\end{equation*}
as in \S\ref{intro:results}. $\mathscr{J}$ is the list of pairs $\left( \underline{\theta}^{(i)}, \underline{d}^{(\alpha)} \right) $ where $i \in \mathcal{J}$ and $\alpha \in \mathfrak{J}$ are the paired-up indices. Then, the lists of vectors defining $X$ in \S\ref{intro:results} are nothing but:
\begin{align*}
    \Theta \setminus p_1 (\mathscr{J}) &=\{ \underline{\theta}^{(i)} , \ i =1, \dots, N , \ i \notin \mathcal{J}\} , \\
    \mathcal{D} \setminus \pi_1 (\mathscr{H}) &=\{ \underline{d}^{(\alpha)} , \ \alpha \in \mathfrak{W}\} , \\
    \mathcal{D} \setminus \left( p_2 (\mathscr{J})  \cup \pi_2 (\mathscr{H})\right) &=\{ \underline{d}^{(\alpha)} , \ \alpha \in \mathfrak{H}\} .
\end{align*}\par
To demystify the notation, here are three simple threefold examples. 
\begin{itemize}
    \item If $Y \longrightarrow \cP^3$ is a double cover branched over a cubic hypersurface and five hyperplanes, $\Theta \cong ( 1,1,1,1)$ and $\mathcal{D} \cong (3,1,1,1,1,1)$. Then $p_1 (\mathscr{J})= \Theta$ and $p_2 (\mathscr{J})$ consists of the last four entries of $\mathcal{D}$. Hence, $\alpha \in \mathfrak{W}$ runs over all the entries in $\mathcal{D}$ while $\alpha \in \mathfrak{H}$ runs only over the first two entries in $\mathcal{D}$. $(\mathcal{D} \setminus p_2 (\mathscr{J})) \cong (3,2)$ in this case. We take the GIT quotient $\C^{0+6} /\!\!/ \C^{*}$ with weights $\mathcal{D}$, and intersect with two hypersurfaces of degrees $2 (\mathcal{D} \setminus p_2 (\mathscr{J})) \cong (6,2)$.\par
    In conclusion, $X$ in this example is the complete intersection of a degree-6 and a degree-2 hypersurfaces in the weighted projective space $\cP^5_{(3,1^5)}$.
    \item On the other hand, if $Y \longrightarrow \cP^3$ is branched over a cubic and a generic quintic hypersurface, $\Theta$ is as above but $\mathcal{D} \cong (5,3)$, whence there are no simplifications; $\mathscr{J}= \emptyset $, and both $\mathfrak{W}$ and $\mathfrak{H}$ label all the entries in $\mathcal{D}$. The GIT quotient $\C^{4+2} /\!\!/ \C^{*}$ is taken with weights $2 \Theta \sqcup \mathcal{D} \cong \{ (2,2,2,2), (5,3) \}$ and is intersected with two hypersurfaces of degrees $2\mathcal{D} \cong (10,6)$.\par
    In conclusion, $X$ in this example is the complete intersection of a degree-6 and a degree-10 hypersurfaces in the weighted projective space $\cP^5_{(5,3,2^4)}$.
    \item As an example of Picard rank 2, consider the double cover $Y \longrightarrow \cP^1 \times \cP^2$ with branching locus characterized by the matrix 
\begin{equation*}
    \mathcal{D} = \left( \begin{matrix} 2 & 1 & 1 & 0 \\ 1 & 3 & 1 & 1 \end{matrix} \right) .
\end{equation*}
The set of vectors $\Theta$ can also be represented as a matrix. In the present case, in which the base is a direct product, it is 
\begin{equation*}
    \Theta = \left( \begin{matrix} 1 & 1 & 0 &0 & 0 \\ 0 & 0 & 1 & 1 & 1\end{matrix} \right) .
\end{equation*}
In this case $\lvert \mathscr{J}\rvert =1$, because the last column of $\mathcal{D}$ can be paired up with the weight of an affine coordinate of the base, in the last column of $\Theta$. Therefore, 
\begin{equation*}
    \Theta \setminus p_1 (\mathscr{J}) \cong  \left( \begin{matrix} 1 & 1 & 0 &0 \\ 0 & 0 & 1 & 1 \end{matrix} \right) , \qquad \mathcal{D} \setminus p_2 (\mathscr{J}) \cong  \left( \begin{matrix} 2 & 1 & 1  \\ 1 & 3 & 1  \end{matrix} \right) .
\end{equation*}
The smooth $X$ sits inside the GIT quotient $\C^{4+4} /\!\!/ (\mathbb{C}^{*})^2$, where the torus acts with weights 
\begin{equation*}
    2 \left( \Theta \setminus p_1 (\mathscr{J}) \right) \sqcup \mathcal{D}  \cong \left( \begin{matrix} 2 & 2 & 0 & 0 & 2 & 1 & 1 & 0 \\ 0 & 0 & 2 & 2 & 1 & 3 & 1 & 1 \end{matrix} \right) ,
\end{equation*}
and $X$ is cut out by three sections of bi-degree
\begin{equation*}
   2 \left(  \mathcal{D} \setminus p_2 (\mathscr{J}) \right) \cong  \left( \begin{matrix} 4 & 2 & 2 \\ 2 & 6 & 2 \end{matrix} \right) .
\end{equation*}
\end{itemize}

\subsection{Examples}
\label{sec:GLSMEx}

\subsubsection{One-modulus family}
Consider the situation $s=1$, and $\dim \Ync = N-1 = \dim X$. This case corresponds to the non-commutative resolution of 
\begin{equation*}
    Y \ \xrightarrow{ \ 2:1 \ } \ \cP^{N-1}_{\vec{\theta}} ,    
\end{equation*}
a branched double cover of a weighted projective space, and has only one modulus $\zeta$, which we take to be in a neighborhood of $\zeta =0$. We further assume for simplicity that $\theta^{(N)}=1$ \cite{Hosono:1993qy}.\par
Write $\Theta := \zeta \frac{\partial \ }{\partial \zeta}$. The box operators derived in \S\ref{sec:GLSMGKZderivation} are equivalent to
\begin{equation*}
    \prod_{\iota } \prod_{k_{\iota}=0}^{\ell^{(\iota)} -1} \left( \ell^{(\iota)} \Theta - k_{\iota} \right) \ - \ \zeta \prod_{\alpha \in \mathfrak{H}} \prod_{k_{\alpha}=0}^{2d^{(\alpha)} -1} \left( -2d^{(\alpha)} \Theta - k_{\alpha} \right) .
\end{equation*}
Factoring this expression one identifies a Picard--Fuchs operator annihilating the $A$-periods of $X$.\par
To compute the holomorphic $A$-periods in this family of examples, we pick any matrix factorization $\mathcal{O}_{\mathrm{pt}}$ such that 
\begin{equation*}
    f_{\mathcal{O}_{\mathrm{pt}}} (\sigma) = \prod_{i=1}^{N}\left( 1+e^{2\pi \theta^{(i)}\sigma } \right) \prod_{i=1}^{N-1}\left( 1-e^{2\pi \theta^{(i)}\sigma } \right) ,
\end{equation*}
The notation $\mathcal{O}_{\mathrm{pt}} $ emphasizes that this matrix factorization is the analogue of a skyscraper sheaf on $\Ync$. We now proceed to show that this choice gives holomorphic $A$-periods, and evaluate them explicitly.\par
We start from expression \eqref{eq:ZYncduplic} specialized to the present case, which reads 
\begin{equation*}
\begin{aligned}
    \mz_{\Ync} \left( \mathcal{O}_{\mathrm{pt}}\right) = \pi^{\frac{\hat{r}}{2}} C_{\Ync} \int_{\gamma} & \zeta^{- \ii \sigma} \prod_{i=1}^{N} \Gamma \left( 2 \ii \theta^{(i)} \sigma\right) \prod_{\alpha=1}^{\hat{r}} \frac{ \Gamma \left( 1-2 \ii d^{(\alpha)} \sigma \right) }{\Gamma \left( 1-\ii d^{(\alpha)} \sigma \right)} \\
    &\times e^{2\pi \sigma \sum_{i=1}^{N-1} \theta^{(i)} } \left( 1+e^{2 \pi \sigma} \right) \prod_{i=1 }^{N-1} \left( -2 \sinh \left( 2 \pi \theta^{(i)} \sigma \right) \right) .
\end{aligned}
\end{equation*}
Here $C_{\Ync} \ne 0$ is a normalization constant, independent of $\Br$, which we introduce to normalize the holomorphic $A$-period to be $1+O(q)$.\par
\begin{figure}[t]
    \centering
    \includegraphics[width=0.4\linewidth]{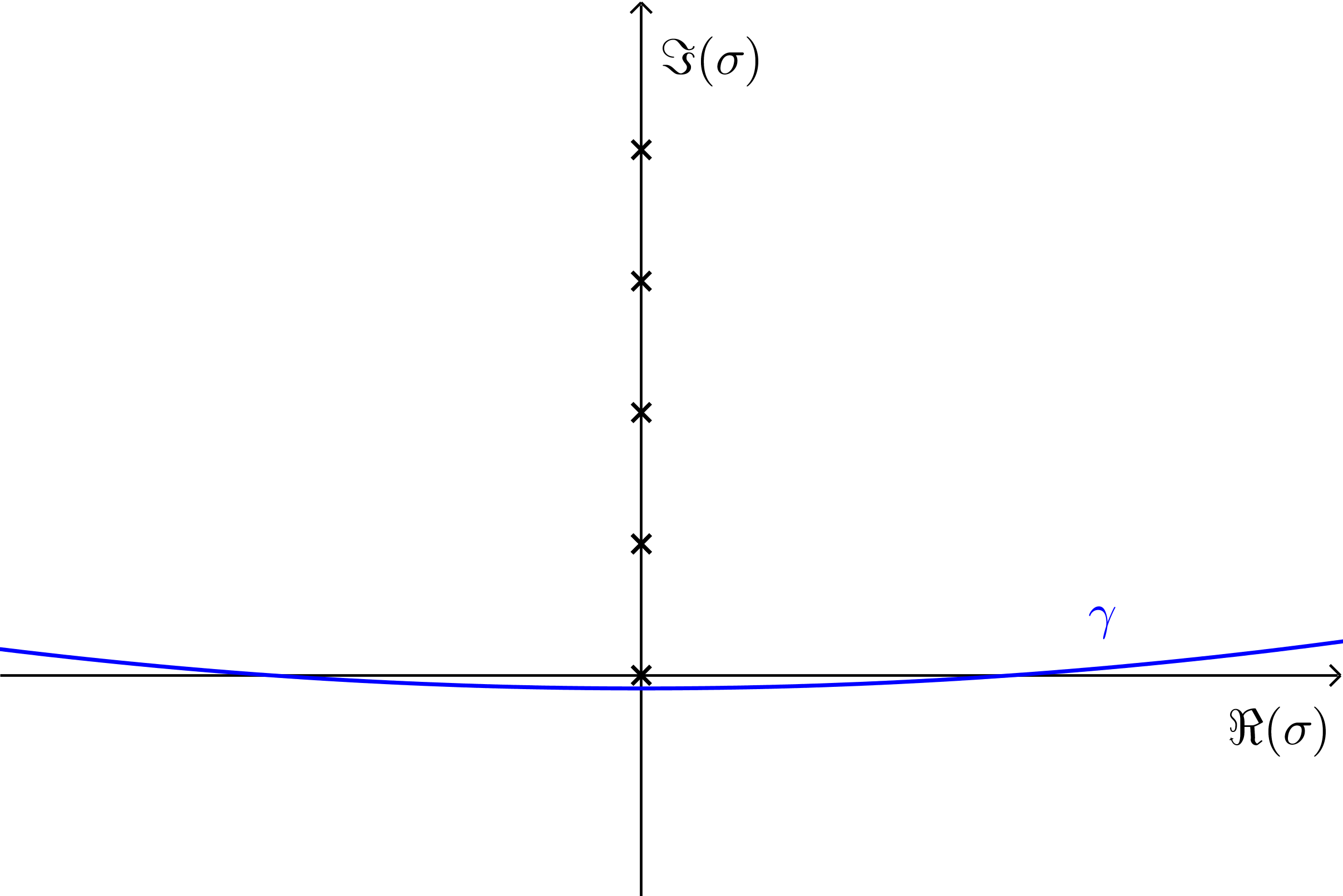}\hspace{0.1\linewidth}\includegraphics[width=0.4\linewidth]{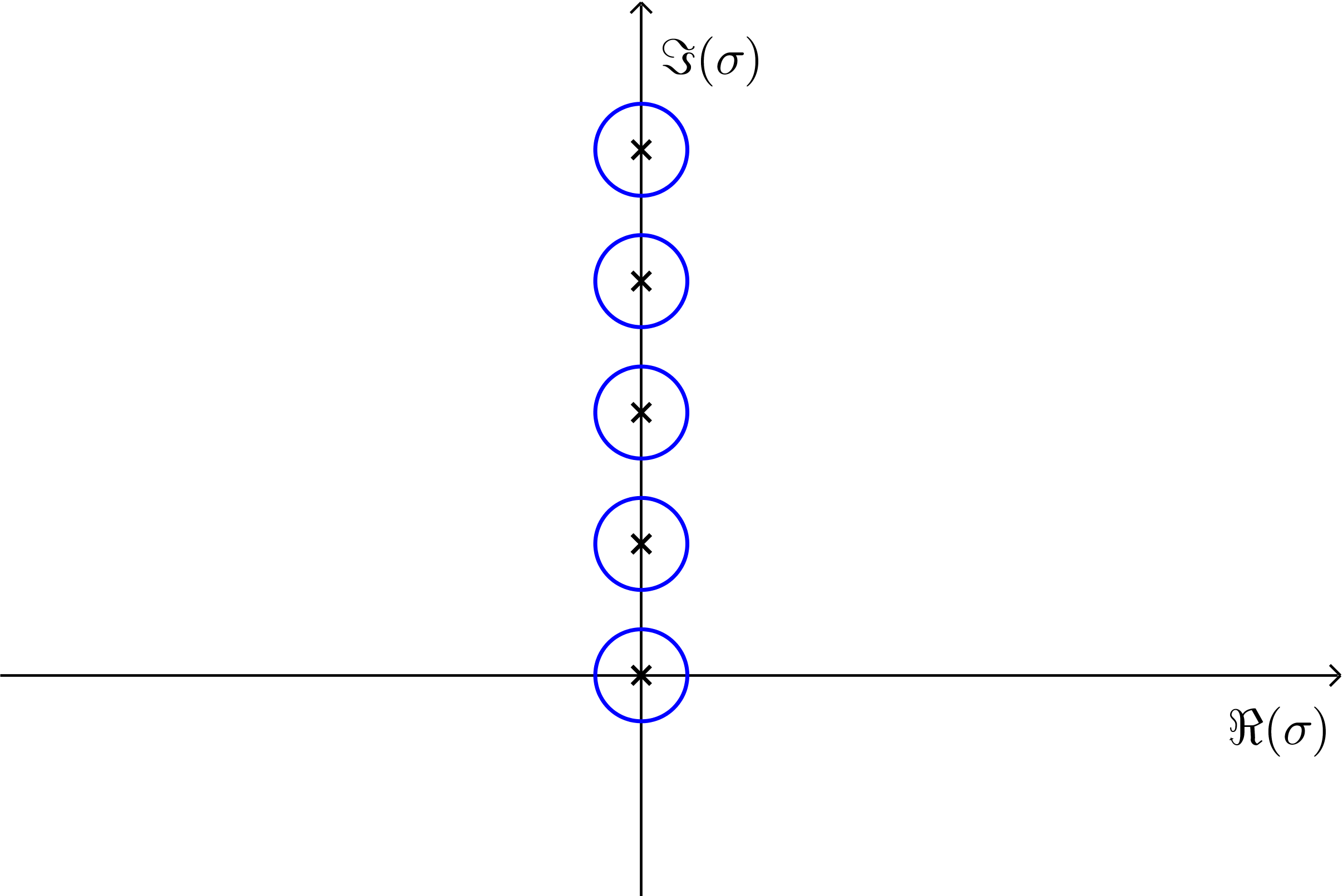}
    \caption{Contour deformation in the one-modulus family. Left: We tilt $\gamma$ in the upper-half plane and close it at $\Im (\sigma) \to + \infty$. Right: We deform the integration cycle to encircle the non-trivial poles.}
    \label{fig:residue-contour}
\end{figure}\par
In the chamber $0<\lvert \zeta \rvert \ll 1$ (equivalently $\Re (t) \gg +1$) it is possible to close the contour $\gamma$ in the upper-half plane. For $\Im (\sigma) \ge 0$ the Gamma functions yield towers of poles at 
\begin{equation*}
    \sigma = \ii \frac{n^{(i)}}{2 \theta^{(i)}} , \qquad n^{(i)} \in \N .
\end{equation*}
However, using 
\begin{equation*}
\begin{aligned}
    \prod_{i=1}^{N} \Gamma \left( 2 \ii \theta^{(i)} \sigma\right) & \prod_{i=1 }^{N-1} \left( -2 \sinh \left( 2 \pi \theta^{(i)} \sigma \right) \right) \\
    & = (2\pi \ii )^{N-1} \left( - \ii \frac{\pi}{\sinh \left( 2 \pi \sigma \right) } \right) \prod_{i=1}^{N}\frac{1}{\Gamma \left( 1-2 \ii \theta^{(i)} \sigma\right)} ,
\end{aligned}
\end{equation*}
we see that only the $N^{\text{th}}$ tower of poles contributes. We thus deform the contour to encircle these poles, as shown in Figure \ref{fig:residue-contour}, and change variables $\sigma = \ii \frac{n^{(N)}}{2 } + u$. After simple manipulations we get 
\begin{equation*}
\begin{aligned}
    \mz_{\Ync} \left( \mathcal{O}_{\mathrm{pt}} \right) & = (- \ii) \pi^{\frac{\hat{r}}{2}} (2\pi \ii )^{N}C_{\Ync} \sum_{n^{(N)}=0}^{\infty}  \zeta^{\frac{n^{(N)}}{2}} (-1)^{n^{(N)} \sum_{i=1}^{N} \theta^{(i)} } \\
     & \times \oint \frac{\dd u}{2\pi \ii} \zeta^{- \ii u}  \left( \frac{\pi}{\sinh \left( 2 \pi u \right) } \right) e^{2\pi u \sum_{i=1}^{N} \theta^{(i)} } \left( (-1)^{n^{(N)}} + e^{-2 \pi u} \right) \\
    & \times \prod_{i=1}^{N}\frac{1}{\Gamma \left( 1+ \theta^{(i)} n^{(N)} -2 \ii \theta^{(i)} u\right)}  \prod_{\alpha=1}^{\hat{r}} \frac{ \Gamma \left( 1+ d^{(\alpha)} n^{(N)}-2 \ii d^{(\alpha)} u \right) }{\Gamma \left( 1+ d^{(\alpha)} \frac{ n^{(N)}}{2} -\ii d^{(\alpha)} u \right)}  .
\end{aligned}
\end{equation*}
When $n^{(N)}$ is odd the integral has vanishing residue, hence only $n^{(N)}=2n \in 2 \N$ contribute. Fixing $C_{\Ync} = \frac{- \ii }{\pi^{\frac{\hat{r}}{2}} (2\pi \ii)^{N}}$ and taking the residue at the simple pole, we arrive at 
\begin{equation}
\label{eq:holoperiod}
     \mz_{\Ync} \left( \mathcal{O}_{\mathrm{pt}} \right) = \sum_{n=0}^{\infty} \zeta^{n}  \prod_{i=1}^{N}\frac{1}{\left(2\theta^{(i)} n \right)!}  \prod_{\alpha=1}^{\hat{r}} \frac{ \left( 2d^{(\alpha)} n \right)! }{\left( d^{(\alpha)} n \right)!} .
\end{equation}
As claimed, \eqref{eq:holoperiod} is holomorphic in $\zeta$ and normalized to $\lim_{\zeta \to 0} \mz_{\Ync} \left( \mathcal{O}_{\mathrm{pt}} \right) =1$.

\subsubsection{Double covers of \texorpdfstring{$\mathbb{P}^3$}{P3}}

\begin{table}[th]
\centering
\begin{tabular}{|c |c |c |c | c| c|}
	\hline
	\textsc{base} & \textsc{nef partition} & $\mathrm{gcd}$ & $X$ & \textsc{double mirror} & \eqref{eq:gaugefixing} \\
	\hline
	$\cP^{3}$ & $(1^8)$ & 1 &$\cP^{7} [2,2,2,2]$ & yes & yes \\
	$\cP^{3}$ & $(2, 1^6)$ & 1 & $\cP^{5} [4,2]$ & no & yes \\
	$\cP^{3}$ & $(2^2, 1^4)$ & 1 & $\cP^{5}_{(2^2,1^4)} [4,4]$ & no & yes \\
	$\cP^{3}$ & $(2^3,1^2)$ & 1 & $\cP^{6}_{(2^5,1^2)} [4,4,4]$ & no & no \\
	$\cP^{3}$ & $(2^4)$ & 2 & $\widetilde{\cP^{7}} [4,4,4,4]$ & no & no \\
	$\cP^{3}$ & $(3,1^5)$ & 1 & $\cP^{5}_{(3,1^5)} [6,2]$ & no & yes \\
	$\cP^{3}$ & $(3,2,1^3)$ & 1 & $\cP^{5}_{(3,2^2,1^3)} [6,4]$ & no & no \\
	$\cP^{3}$ & $(3,2^2,1)$ & 1 & $\cP^{7}_{(3,2^5,1)} [6,4,4]$ & no & no \\
	$\cP^{3}$ & $(3^2,1^2)$ & 1 & $\cP_{(3^2,2^2,1^2)}^{5} [6,6]$ & no & no \\
	$\cP^{3}$ & $(3^2,2)$ & 1 & $\cP^{6}_{(3^2,2^5)} [6,6,4]$ & no & no \\
	$\cP^{3}$ & $(4,1^4)$ & 1 & $\cP^{4}_{(4,1^4)} [8]$ & no & yes \\
	$\cP^{3}$ & $(4,2,1^2)$ & 1 & $\cP_{(4,2^3,1^2)}^{5} [8,4]$ & no & no \\
	$\cP^{3}$ & $(4,2^2)$ & 2 & $\widetilde{\cP^{5}} [8,4]$ & no & no \\
	$\cP^{3}$ & $(4,3,1)$ & 1 & $\cP_{(4,3,2^3,1)}^{5} [8,6]$ & no & no \\
	$\cP^{3}$ & $(4^2)$ & 2 & $\cP^{5}_{(4^2,2^4)} [8,8]$ & no & no \\
	$\cP^{3}$ & $(5, 1^3)$ & 1 & $\cP^{4}_{(5,2,1^3)} [10]$ & no & no \\
	$\cP^{3}$ & $(5,2,1)$ & 1 & $\cP^{5}_{(5,2^4,1)} [10,4]$ & no & no \\
	$\cP^{3}$ & $(5,3)$ & 1 & $\cP^{5}_{(5,3,2^4)} [10,6]$ & no & no \\
	$\cP^{3}$ & $(6,1^2)$ & 1 & $\cP^{4}_{(6,2^2,1^2)} [12]$ & no & no \\	
	$\cP^{3}$ & $(6,2)$ & 2 & $\cP^{5}_{(6,2^5)} [12,4]$ & no & no \\
	$\cP^{3}$ & $(7,1)$ & 1 & $\cP^{4}_{(7,2^3,1)} [14]$ & no & no \\
	$\cP^{3}$ & $(8)$ & 2 & $\cP^{4}_{(8,2^4)} [16]$ & no & no \\
	\hline
\end{tabular}
\caption{List of Calabi--Yau threefold pairs $(\Ync, X)$ when $Y \xrightarrow{ \ 2:1 \ } \cP^3$. The nef partion is expressed as a partition of $8$. The column `gcd' indicates whether the partition is divisible by 2. When gcd is 2, all weights of $X$ are even, and $\widetilde{\cP^{\bullet}}$ indicates that the weights of $\cP^{\bullet}$ are doubled. The last column indicates whether or not the nef partition is compatible with \eqref{eq:gaugefixing}.}
\label{tab:listCY3}
\end{table}\par

We collect in \textsc{Table} \ref{tab:listCY3} all instances of Calabi--Yau threefold pairs when the toric base is the projective space $\cP^3$.\par
\begin{itemize}
\item The nef partition $(1^8)$ leads to a double mirror \cite{Caldararu:2010ljp}. It is the only instance of double mirror in the list.
\item The nef partition $(8)$ leads to a smooth branching locus. In this case, $Y,\Ync$ and $\cP^4_{(4,1^4)} [8]$ describe the same smooth geometry.
\item The five partitions $(1^8)$, $(2,1^6)$, $(2^2,1^4)$, $(3,1^5)$, $(4,1^4)$, corresponding to partitions of 4 completed with $1^4$, satisfy the hypotheses of \S\ref{sec:LLY}. In these cases, $\hat{r}=4+r$, with $r$ the length of the partition of 4. They have $k=1$ and $\left\{ 1, 2,3, 4 \right\}\setminus \mathcal{J}=\emptyset$.
\item The partitions $(2^4)$, $(4,2^2)$, $(4^2)$, $(6,2)$, $(8)$ have even components. Therefore, these entries have $k=2$ and $\mathcal{J}= \emptyset$. These five partitions yield $X$ which has all weights doubled from the five partitions in the previous bullet point. The GKZ is the same for these $X$ and the ones in the previous bullet point.
\item The remaining nef partitions all have $k=1$, and $X$ is a weighted projective space with weights given by the nef partition, and filled in with $4-\lvert \mathcal{J}\rvert$ weights equal to $2$.
\end{itemize}
All holomorphic $A$-periods \eqref{eq:holoperiod} are listed in \textsc{Table} \ref{tab:listPi0}. We have verified case by case that they match the holomorphic $A$-periods of the corresponding $X$.\par

\begin{table}[th]
\centering
\begin{tabular}{|c |c |}
	\hline
	\textsc{partition} & \textsc{holomorphic $A$-period}  \\
	\hline
	$(1^8)$ &  $\sum_{n=0}^{\infty} \zeta^n \left( \frac{ (2n)!}{(n!)^2}\right)^4 $  \\
	$(2, 1^6)$ & $\sum_{n=0}^{\infty} \zeta^n  \frac{(4n)! (2n)!}{(n!)^6}$ \\
	$(2^2, 1^4)$ & $\sum_{n=0}^{\infty} \zeta^n \left(  \frac{(4n)!}{(2n)! (n!)^2} \right)^2 $  \\
	$(2^3,1^2)$ & $\sum_{n=0}^{\infty} \zeta^n  \frac{((4n)!)^3}{((2n)!)^5 (n!)^2} $ \\
	$(3,1^5)$ & $\sum_{n=0}^{\infty} \zeta^n  \frac{(6n)! (2n)!}{(3n)! (n!)^5}$ \\
	$(3,2,1^3)$ & $\sum_{n=0}^{\infty} \zeta^n  \frac{(6n)! (4n)!}{(3n)! ((2n)!)^2 (n!)^3}$ \\
	$(3,2^2,1)$ & $\sum_{n=0}^{\infty} \zeta^n  \frac{(6n)! ((4n)!)^2}{(3n)! ((2n)!)^5 n!}$ \\
	$(3^2,1^2)$ & $\sum_{n=0}^{\infty} \zeta^n \left(  \frac{(6n)!}{(3n)! (2n)! n!} \right)^2 $ \\
	$(3^2,2)$ & $\sum_{n=0}^{\infty} \zeta^n  \frac{((6n)!)^2 (4n)!}{((3n)!)^2 ((2n)!)^5 }$ \\
	$(4,1^4)$ & $\sum_{n=0}^{\infty} \zeta^n  \frac{(8n)!}{(4n)! (n!)^4}$ \\
	$(4,2,1^2)$ & $\sum_{n=0}^{\infty} \zeta^n  \frac{(8n)!}{((2n)!)^3 (n!)^2}$ \\
	$(4,3,1)$ & $\sum_{n=0}^{\infty} \zeta^n  \frac{(8n)! (6n)!}{(4n)! (3n)! ((2n)!)^3 (n!)}$ \\
	$(5, 1^3)$ & $\sum_{n=0}^{\infty} \zeta^n  \frac{(10n)!}{(5n)!(2n)! (n!)^3}$ \\
	$(5,2,1)$ & $\sum_{n=0}^{\infty} \zeta^n  \frac{(10n)! (4n)!}{(5n)! ((2n)!)^4 n!}$ \\
	$(5,3)$ & $\sum_{n=0}^{\infty} \zeta^n  \frac{(10n)! (6n)!}{(5n)! (3n)! ((2n)!)^4 }$ \\
	$(6,1^2)$ & $\sum_{n=0}^{\infty} \zeta^n  \frac{(12n)!}{(6n)! ((2n)!)^2 (n!)^2}$ \\	
	$(7,1)$ & $\sum_{n=0}^{\infty} \zeta^n  \frac{(14n)!}{(7n)! ((2n)!)^3 n!}$ \\
	\hline
\end{tabular}
\caption{Holomorphic $A$-periods of double covers of $\cP^3$.}
\label{tab:listPi0}
\end{table}\par

Furthermore, for every $X$ appearing in \textsc{Table} \ref{tab:listCY3} we compute the topological numbers, and compare with the computation of the $A$-periods, finding perfect agreement. More precisely, to each $X$ we apply the following algorithm:
\begin{enumerate}
    \item\label{step:chern} We compute the Chern classes $c_2 (X), c_3 (X)$ using the adjunction formula. Letting $J$ denote the pullback to $X$ of the hyperplane class of the ambient weighted projective space, we schematically write $c_2 (X) = c_2 J \wedge J $, $c_2 \in \Z$. We also write $\chi (X) = \int_X c_3 (X) =c_3 J^3$, where $J^3 := \int_X J \wedge J \wedge J$ is the triple intersection number and $c_3 \in \Z$.
    \item We fix a basis of $K_0 (X)$ given by $([\mathcal{O}_X], [\mathcal{O}_{\mathrm{D}}], [\mathcal{O}_{\mathrm{C}}], [\mathcal{O}_{\mathrm{pt}}])$ corresponding, respectively, to the structure sheaf, divisor class, curve class, and skyscraper sheaf. We compute the corresponding $A$-periods.
    \item We fix the overall normalization by demanding that $\mz (\mathcal{O}_{\mathrm{pt}}; \tnc) = 1 + O(\zeta)$.
    \item  In this normalization, we write 
    \begin{equation*}
        \frac{\log (\zeta) }{2\pi \ii }= \varphi - \varphi_0 
    \end{equation*}
    for a constant $\varphi_0 \in \Z$, fixed by demanding that $\mz (\mathcal{O}_{X}; \tnc)$ does not contain a quadratic piece in $\varphi$.
    \item In this normalization, the $A$-periods should satisfy (see \cite[\S 3.2]{Lin:2024fpz} for a review) 
    \begin{equation*}
    \begin{aligned}
        \mz_X (\mathcal{O}_{X}) &= \frac{J^3}{3!}\varphi^3 + \frac{c_2 J^3}{24}\varphi + \frac{\ii \zeta (3)}{8 \pi^3}  \chi (X) \ + O(\zeta), \\
        \mz_X (\mathcal{O}_{\mathrm{D}}) &= \frac{J^3}{2}\varphi^2 - \frac{J^3}{2}\varphi +\frac{c_2 +4}{24} J^3 \ + O(\zeta), \\
        \mz_X (\mathcal{O}_{\mathrm{C}}) &= \varphi + 1 \ + O(\zeta), \\
        \mz_X (\mathcal{O}_{\mathrm{pt}}) &= 1 \ + O(\zeta). 
    \end{aligned}
    \end{equation*}
    We read off the topological numbers $c_2, J^3, \chi (X)$ from $\mz (\mathcal{O}_{X})$ and compare them with the ones computed in step \eqref{step:chern}. The results for $\mz (\mathcal{O}_{\mathrm{D}})$ and $\mz (\mathcal{O}_{\mathrm{C}})$ are used as additional consistency checks. 
\end{enumerate}
The results are reported in \textsc{Table} \ref{tab:Xnumbers}. The consistency of the calculation showcases that the complete intersections $X$ output by Theorem \ref{thm:smoothX} are well-behaved.\par

\begin{rmk}[Appearances in other works]
    The results in \textsc{Table} \ref{tab:listCY3} for the partitions $(1^8)$, $(2^4)$, $(3,1^5)$, $(3^2,1^2)$, $(4,2^2)$ and $(5,1^3)$ agree with the findings of \cite[Table 3]{Katz:2023zan}. Other works studied singular degeneration of some $X$, then taking a non-commutative resolution of the latter: $\cP^{4}_{(4,1^4)} [8]$ and $\cP^5_{(3^2,2^2,1^2)}[6,6]$ in \cite{Schimannek:2025cok} and $\cP^5_{(3,2^2,1^3)} [6,4]$ in \cite{Knapp:2025hnf}.
\end{rmk}

\begin{table}[th]
\centering
\begin{tabular}{|c |c |c |c | c|}
	\hline
	$X$ & $c_2$ & $c_3$ & $J^3$ & $\mz_X (\mathcal{O}_{X})$ \\
	\hline
	$\cP^{7} [2,2,2,2]$ & 4 & $-8$ & 16 & $\frac{8}{3} \varphi^3 + \frac{8}{3} \varphi -128 \frac{\ii \zeta (3)}{(2\pi\ii)^3}$ \\
	$\cP^{5} [4,2]$ & 7 & $-22$ & 8 & $\frac{4}{3} \varphi^3 + \frac{7}{3} \varphi - 176\frac{\ii \zeta (3)}{(2\pi\ii)^3}$ \\
	$\cP^{5}_{(2^2,1^4)} [4,4]$ &  10 & $-36$ & 4 & $\frac{2}{3} \varphi^3 + \frac{5}{3} \varphi - 144\frac{\ii \zeta (3)}{(2\pi\ii)^3}$  \\
	$\cP^{6}_{(2^5,1^2)} [4,4,4]$ & 13 & $-50$ & 2 & $\frac{1}{3} \varphi^3 + \frac{13}{12} \varphi - 100\frac{\ii \zeta (3)}{(2\pi\ii)^3}$ \\
	$\cP^{5}_{(3,1^5)} [6,2]$ &  13 & $-64$ & 4 & $\frac{2}{3} \varphi^3 + \frac{13}{6} \varphi - 256\frac{\ii \zeta (3)}{(2\pi\ii)^3}$  \\
	$\cP^{5}_{(3,2^2,1^3)} [6,4]$ & 16 & $-78$ & 2 & $\frac{1}{3} \varphi^3 + \frac{4}{3} \varphi - 156\frac{\ii \zeta (3)}{(2\pi\ii)^3}$  \\
	$\cP^{7}_{(3,2^5,1)} [6,4,4]$ & 19 & $-92$ & 1 & $\frac{1}{6} \varphi^3 + \frac{19}{24} \varphi - 92\frac{\ii \zeta (3)}{(2\pi\ii)^3}$  \\
	$\cP_{(3^2,2^2,1^2)}^{5} [6,6]$ & 22 & $-120$ & 1 & $\frac{1}{6} \varphi^3 + \frac{11}{12} \varphi - 120\frac{\ii \zeta (3)}{(2\pi\ii)^3}$  \\
	$\cP^{6}_{(3^2,2^5)} [6,6,4]$ & 25 & $-134$ & 1 & $\frac{1}{6} \varphi^3 + \frac{25}{24} \varphi - 134\frac{\ii \zeta (3)}{(2\pi\ii)^3}$ \\
	$\cP^{4}_{(4,1^4)} [8]$ & 22 & $-148$ & 2 & $\frac{1}{3} \varphi^3 + \frac{11}{6} \varphi - 296\frac{\ii \zeta (3)}{(2\pi\ii)^3}$  \\
	$\cP_{(4,2^3,1^2)}^{5} [8,4]$ & 25 & $-162$ & 1 & $\frac{1}{6} \varphi^3 + \frac{25}{24} \varphi - 162\frac{\ii \zeta (3)}{(2\pi\ii)^3}$  \\
	$\cP_{(4,3,2^3,1)}^{5} [8,6]$ & 31 & $-204$ & 1 & $\frac{1}{6} \varphi^3 + \frac{31}{24} \varphi - 204\frac{\ii \zeta (3)}{(2\pi\ii)^3}$ \\
	$\cP^{4}_{(5,2,1^3)} [10]$ & 34 & $-288$ & 1 & $\frac{1}{6} \varphi^3 + \frac{17}{12} \varphi - 288\frac{\ii \zeta (3)}{(2\pi\ii)^3}$ \\
	$\cP^{5}_{(5,2^4,1)} [10,4]$ & 37 & $-302$ & 1 & $\frac{1}{6} \varphi^3 + \frac{37}{24} \varphi - 302\frac{\ii \zeta (3)}{(2\pi\ii)^3}$ \\
	$\cP^{5}_{(5,3,2^4)} [10,6]$ & 43 & $-344$ & 1 & $\frac{1}{6} \varphi^3 + \frac{43}{24} \varphi - 344\frac{\ii \zeta (3)}{(2\pi\ii)^3}$  \\
	$\cP^{4}_{(6,2^2,1^2)} [12]$ & 49 & $-498$ & 1 & $\frac{1}{6} \varphi^3 + \frac{49}{24} \varphi - 498\frac{\ii \zeta (3)}{(2\pi\ii)^3}$ \\	
	$\cP^{4}_{(7,2^3,1)} [14]$ & 67 & $-792$ & 1 & $\frac{1}{6} \varphi^3 + \frac{67}{24} \varphi - 792\frac{\ii \zeta (3)}{(2\pi\ii)^3}$  \\
	\hline
\end{tabular}
\caption{Topological data of the Calabi--Yau complete intersections appearing in \textsc{Table} \ref{tab:listCY3}, and the leading term in the $A$-period of the structure sheaf.}
\label{tab:Xnumbers}
\end{table}\par

\begin{appendix}
\section{Summary of notation}
\label{app:Notation}
We list the algebraic varieties involved in the main text:
\begin{itemize}
	\item $\cP_{\nabla}$ is a toric variety;
    \item $\cB_{\nabla}$ is a Fano MPCP desingularization of $\cP_{\nabla}$;
	\item $X$ is a Calabi--Yau complete intersection in a toric variety;
	\item $Y$ is a branched double cover of a toric Fano variety;
	\item $\Ync$ is a non-commutative resolution of $Y$;
	\item $\Xpre$ is the pre-quotient of $Y$.
\end{itemize}
Part of the geometric data is characterized by integers:
\begin{itemize}
	\item $N$ is the number of primitive 1-cones generating the fan of $\nabla$;
	\item $n=\dim \cB_{\nabla}$;
	\item $s$ is the dimension of the algebraic torus acting on $\cB_{\nabla}$;
	\item $r$ is the length of the nef partition $\mathcal{I}_1 \sqcup \cdots \sqcup \mathcal{I}_r $ (complete intersections).
    \item $\hat{r}$ is the length of the general nef partition $\mathcal{I}_1 \sqcup \cdots \sqcup \mathcal{I}_{\hat{r}} $ (double covers).
    \item $\hat{r}=r+N$ if the general nef partition allows for gauge fixing.
\end{itemize}
We adhere to the following notation for the indices:
\begin{itemize}
	\item $i \in \left\{ 1, \dots, N\right\}$;
	\item $a \in \left\{ 1, \dots, s\right\}$;
	\item $\alpha \in \left\{ 1, \dots, \hat{r}\right\}$;
	\item $I \in \left\{1, \dots, N+r \right\}$, with the first $r$ entries referring to the $\alpha$-index and the latter $N$ entries referring to the $i$-index.
\end{itemize}
The weights under the action of the $a^{\text{th}}$ $\C^{\ast}$-factor in $(\C^{\ast})^s$ are:
\begin{itemize}
	\item $\ell_a ^{(i)}$ is the weight of the $i^\text{th}$ homogeneous coordinate (ambient);
	\item $\theta_a ^{(i)}$ is the weight of the $i^\text{th}$ homogeneous coordinate (base);
	\item $d_a ^{(\alpha)}$ is the weight of the $\alpha^\text{th}$ section.
\end{itemize}
These weights are collected into arrays of $s$ entries $\underline{\ell}^{(i)}$, $\underline{\theta}^{(i)}$, and $\underline{d}^{(\alpha)}$, respectively. On the other hand, $\vec{\ell}_a$ denotes the array of $N$ entries containing the weights of all homogeneous coordinates under the action of the $a^{\text{th}}$ $\C^{\ast}$, for fixed $a$.\par
$\underline{\ell}^{(i)}$ and $\underline{\theta}^{(i)}$ play the same role, but we use $\underline{\theta}^{(i)}$ to emphasize a toric Fano base of double covers is considered.\par
\begin{itemize}
    \item $\cP^{n}_{\vec{w}} [d^{(1)}, \dots, d^{(r)}]$ denotes the complete intersection of $r$ hypersurfaces of degrees $d^{(\alpha)}$ in the $n$-dimensional weighted projective space of weight vector $\vec{w} \in \N^{n+1}$.
\end{itemize}
Moreover, we use $\ii:= \sqrt{-1}$ to avoid confusion with the index $i$, and $\dd$ for the differential.\par

\section{\texorpdfstring{$A$}{A}-periods of Calabi--Yau complete intersections}
\label{app:CICYperiod}

This appendix is devoted to the analysis of Calabi--Yau complete intersections in toric varieties, using GLSMs.\par
In \S\ref{sec:CICYsetup} we review the generalities of the GLSM and the $A$-period. These are well known and extensively studied, thus we shall be brief. In \S\ref{app:AperiodVGLSM} we introduce the $A$-periods of non-compact Calabi-Yau varieties closely related to the complete intersections of interest.\par 
In \S\ref{app:GKZCICY} we devise a method for systematically reading off the GKZ system that annihilates the $A$-periods of the Calabi-Yau complete intersection, from their integral representation. Closely related computations appeared previously in \cite{Gu:2020ana}, albeit without analysis of the GKZ system. The analysis in \S\ref{app:GKZCICY} is used to substantiate the results in the main text.

\subsection{GLSMs and \texorpdfstring{$A$}{A}-periods for Calabi--Yau complete intersections}
\label{sec:CICYsetup}

The setup is as in \S\ref{sec:toric}. We briefly recall the notation: $X$ is a Calabi--Yau complete intersection in a toric ambient space of complex dimension $n$. The number of affine coordinates is $N$, with weights $\ell^{(i)}_a$ under the action of $\C^{\ast} _a \subset (\C^{\ast})^s$, $i=1, \dots, N$; whence $N=n+s$. The Calabi--Yau $X$ is defined as the complete intersection of $r$ polynomials $f^{(\alpha)}$ of weights $d_a^{(\alpha)}$ under $\C^{\ast} _a$.\par

\subsubsection{GLSMs for complete intersections} 
It is possible to realize this geometry by a GLSM \cite{Witten:1993yc,Morrison:1994fr}, in the chamber $\Re (t_a) \gg +1$ for all $a=1, \dots, s$. Fields and charges of the GLSMs are summarized in \textsc{Table} \ref{tab:GLSMcharges}.\par

\begin{table}[th]
\centering
\begin{tabular}{l |c |c | c}
	field & \# & $U(1)_a$ gauge charge & $U(1)_V$ R-charge \\
	\hline
	$\phi_i$ & $i=1, \dots, N$ & $\ell_a ^{(i)}$ & $2\varepsilon$ \\
	$P_{\alpha}$ & $\alpha=1, \dots, r$ & $- d_a^{(\alpha)}$ & $2-2\varepsilon$\\
	\hline
\end{tabular}
\caption{Field content and charges of the GLSM for Calabi--Yau complete intersections in toric varieties.}
\label{tab:GLSMcharges}
\end{table}\par
For later convenience, we introduce the notation
\begin{equation}
\label{eq:tildeLCICY}
	\tilde{\ell}_a ^{(I)} := \begin{cases} - d_a ^{(\alpha=I)}  & \ 1 \le I \le r \\ \ell_a ^{(i=I-r)} &  \ r+1 \le I \le r+N . \end{cases}
\end{equation}
The superpotential reads 
\begin{equation*}
	W_X = \sum_{\alpha=1}^{r} P_{\alpha} f^{(\alpha)} (\phi) .
\end{equation*}

\subsubsection{B-branes and \texorpdfstring{$A$}{A}-periods of complete intersections}
We consider the category of $U(1)^s$-equivariant matrix factorizations of $W_X$, denoted $ \MF_{U(1)^s} (W_X)$.\par
We define the partition function of an object $\Br^{\prime} \in \MF_{U(1)^s} (W_X)$ as:
\begin{equation}
\label{eq:ZCICY}
	\mz_{X} \left( \Br^{\prime} ; t \right) = \int_{\gamma_t} \dd^s \sigma~e^{\ii \langle t, \sigma\rangle} ~\prod_{i=1}^{N} \Gamma \left( \varepsilon + \ii \langle \underline{\ell}^{(i)}, \sigma \rangle \right) \prod_{\alpha=1}^{r} \Gamma \left( 1 - \varepsilon - \ii \langle\underline{d}^{(\alpha)}, \sigma \rangle \right) ~f_{\Br^{\prime}} (\sigma) .
\end{equation}
An admissible contour $\gamma_t$ in this setting is a middle-dimensional, Weyl-invariant cycle in 
\begin{equation*}
	\mathfrak{t}_{\C} \setminus \left( \bigcup_{i=1}^{N} \bigcup_{n^{(i)} \ge 0} \left\{ \langle \underline{\ell}^{(i)}, \sigma \rangle =\ii \left( n^{(i)} +\varepsilon\right) \right\}\cup \bigcup_{\alpha=1}^{r} \bigcup_{n^{(\alpha)} \ge 1} \left\{ \langle \underline{d}^{(\alpha)}, \sigma \rangle = -\ii \left( n^{(\alpha)} -\varepsilon\right) \right\} \right) ,
\end{equation*}
with asymptotic directions such that \eqref{eq:ZCICY} converges.\par
In the chamber $\Re (t_a) \gg +1$, there exists a derived projection functor \cite{Herbst:2008jq,Ballard:2016ncw,HalpernLeistner:2014}
\begin{equation*}
	\pi : \MF_{U(1)^s} (W) \ \longrightarrow \ D^{b} \mathrm{Coh} (X).
\end{equation*}
This allows us to extract the $A$-period of any complex of coherent sheaves $\pi (\Br^{\prime})$ on $X$ from \eqref{eq:ZCICY}.

\subsubsection{GLSMs and \texorpdfstring{$A$}{A}-periods for line bundles}
\label{app:AperiodVGLSM}

The toric data in \S\ref{sec:toric} define a direct sum of line bundles over the toric Fano variety $\cB_{\nabla} $. One then introduces the non-compact variety 
\begin{equation*}
	V := \mathrm{Tot} \left( F_1 \oplus \cdots \oplus F_r \ \longrightarrow \ \cB_{\nabla}  \right) ,
\end{equation*}
which is Calabi--Yau by construction.\par 
It is possible to describe $V$ with a GLSM, whose fields and charges of are summarized in \textsc{Table} \ref{tab:GLSMambient}. The latter GLSM is obtained from the one for $X$ (\textsc{Table} \ref{tab:GLSMcharges}) upon replacing $W_X$ with $W_V \equiv 0$ \cite{Morrison:1994fr}, and adapting the values of the $U(1)_V$ charges accordingly.\par
\begin{table}[th]
\centering
\begin{tabular}{l |c |c | c}
	field & \# & $U(1)_a$ gauge charge & $U(1)_V$ R-charge \\
	\hline
	$\phi_i$ & $i=1, \dots, N$ & $\ell_a ^{(i)}$ & $0$ \\
	$P_{\alpha}$ & $\alpha=1, \dots, r$ & $- d_a^{(\alpha)}$ & $0$\\
	\hline
\end{tabular}
\caption{Field content and charges of the GLSM for line bundles on toric varieties.}
\label{tab:GLSMambient}
\end{table}\par
Being $W_V$ trivial, a trivial matrix factorization exists, denoted as the object $\varnothing \in \MF_{U(1)^s} (0)$, for which $f_{\varnothing} (\sigma)=1$. The corresponding partition function is 
\begin{equation}
\label{eq:Zambient}
	\mz_{V} \left( t \right) := \mz_{V} \left( \varnothing ; t \right) = \int_{\gamma_t} \dd^s \sigma~e^{\ii \langle t, \sigma\rangle} ~\prod_{i=1}^{N} \Gamma \left(\ii \langle \underline{\ell}^{(i)}, \sigma \rangle \right) \prod_{\alpha=1}^{r} \Gamma \left( -\ii \langle\underline{d}^{(\alpha)}, \sigma \rangle \right) .
\end{equation}
The contour $\gamma_t$ runs along the real locus in $\mathfrak{t}_{\C}$, with a small detour avoiding $\sigma^{a}=0$ for every $a=1, \dots, s$.

\subsection{GKZ system for \texorpdfstring{$A$}{A}-periods of Calabi--Yau complete intersections}
\label{app:GKZCICY}
We now present an operational definition of $A$-periods. Consider a 1-parameter family of partition functions of the form
\begin{equation}
\label{eq:Zvaux}
    \mz^{(v)} \left( \Br^{\prime}; t \right)  = \int_{\gamma_t} \dd^s \sigma~e^{\ii \langle t, \sigma\rangle} ~\prod_{i=1}^{N} \Gamma \left(\ii \langle \underline{\ell}^{(i)}, \sigma \rangle \right) \prod_{\alpha=1}^{r} \Gamma \left( v -\ii \langle\underline{d}^{(\alpha)}, \sigma \rangle \right) f_{\Br^{\prime}} (\sigma) ,
\end{equation}
labeled by $0 \le v \le 1$. Introduce a redundant set of variables $\left\{ c_I , \ I=1 \dots, N+r \right\}$, related to the true moduli $e^{-t_a}$ via 
\begin{equation}
\label{eq:CICYGKZvar}
	\prod_{I=1}^{N+r} \left( - c_I \right)^{\tilde{\ell}_a ^{(I)}} = e^{-t_a} , \qquad \forall a=1, \dots, s .
\end{equation}
\begin{defin}
    With the notation as above, we define the $A$-period to be
    \begin{equation*}
        \pz^{(v)} \left( \Br^{\prime}; c \right)  = \left. \left( \prod_{\alpha=1}^{r} c_{\alpha}^{-v} \right) ~\mz^{(v)} \left( \Br^{\prime}; t \right)  \right\rvert_{\text{\eqref{eq:CICYGKZvar}}} ,
    \end{equation*}
    with the right-hand side in terms of \eqref{eq:Zvaux} subject to the replacement \eqref{eq:CICYGKZvar}.
\end{defin}\par
\medskip
We are interested in the following three situations:
\begin{itemize}
    \item Let $X$ be a Calabi--Yau complete intersection in a toric ambient space, as in \S\ref{sec:CICYsetup}, and recall \eqref{eq:tildeLCICY}. It has $v=1$, whence  
    \begin{equation}\label{eq:AperiodCICY}
        \pz_X \left( \Br^{\prime} ; c \right) = \left( \prod_{\alpha=1}^{r} \frac{1}{c_{\alpha}} \right) \mz_{X} \left( \Br^{\prime} ; t \right)
    \end{equation}
    with the right-hand side proportional to \eqref{eq:ZCICY} subject to the replacement \eqref{eq:CICYGKZvar}.
    \item Let $V$ be the non-compact Calabi--Yau constructed in \S\ref{app:AperiodVGLSM}. It has $v=0$, hence $\pz (\varnothing; c)$ is directly obtained plugging \eqref{eq:CICYGKZvar} in \eqref{eq:Zambient}, with unit coefficient.
    \item Let $Y$ be a Calabi--Yau double cover, and $\Ync$ its non-commutative resolution. It has $v=\frac{1}{2}$, and this definition recovers \eqref{eq:Anc}.
\end{itemize}\par

\begin{thm}\label{thm:CICYGKZ}For every $\Br^{\prime} \in D^{b} \mathrm{Coh} (X)$, \eqref{eq:AperiodCICY} satisfies the GKZ system \eqref{eq:GKZCICYBox}-\eqref{eq:GKZCICYEuler}
with 
\begin{equation}
\label{eq:beta=1CICY}
\beta = \left( \begin{matrix} -1 \\ \vdots \\ -1 \\ 0 \\ \vdots \\ 0\end{matrix} \right) \in \mathbb{Q}^{r+n} .
\end{equation}
\end{thm}
The rest of this section contains the proof of Theorem \ref{thm:CICYGKZ}, broken down in various steps for readability.
\begin{enumerate}[(i)]
    \item The fist step consists in relating $\pz_X$ to $\mz_V$. This task is performed in \S\ref{app:sec:mzV}.
    \item The second step derives the box operators that annihilate $\mz_V$, and is performed in \S\ref{app:sec:BoxOp}.
    \item The third step consists in showing that both $\mz_V$ and $\mz_X$ satisfy the Euler operators with $\beta=0$. This task is performed in \S\ref{app:sec:EuOp}.
    \item The last step consists in showing that if $\mz_X$ satisfies the GKZ system with $\beta=0$, then $\pz_X$ satisfies the GKZ system with $\beta$ as in \eqref{eq:beta=1CICY}. This is done in \S\ref{app:sec:prefactor}.
\end{enumerate}
Putting all the steps together completes the proof. We remark that these manipulations are parallel to the ones used in \cite{Hori:2000kt}, and also in \cite{Hori:2013ika} to relate $A$-periods and $B$-periods of mirror pairs. However, here we keep track of the integration contours in a detailed way. Moreover, we put a concentrated focus on the GKZ system annihilating the $A$-periods.\par

\subsubsection{\texorpdfstring{$A$}{A}-periods of a line bundle}
\label{app:sec:mzV}
We begin by setting $\varepsilon \to 0^{+}$ in \eqref{eq:ZCICY}. This can be done without loss of generality, as we will only be interested in a neighborhood of $e^{-t_a}=0$ for all $a=1, \dots, s$.\par
We then write $f_{\Br^{\prime}} (\sigma)$ in the generic form \eqref{eq:fBsum}, so that 
\begin{equation*}
	\mz_{X} \left( \Br^{\prime} ; t \right) =  \sum_{k \in \mathfrak{N}_{\Br^{\prime}}} (-1)^{\mathsf{r}_{k}} \mz_{X} (\varnothing ; t-2 \pi \ii q_k) ,
\end{equation*}
where we have defined
\begin{equation*}
	 \mz_{X} (\varnothing ;t) := \int_{\gamma_t} \dd^s \sigma~e^{\ii \langle t, \sigma\rangle} ~\prod_{i=1}^{N} \Gamma \left(  \ii \langle \underline{\ell}^{(i)}, \sigma \rangle \right) \prod_{\alpha=1}^{r} \Gamma \left( 1 - \ii \langle\underline{d}^{(\alpha)}, \sigma \rangle \right) .
\end{equation*}
The identity 
\begin{align*}
	e^{\ii \langle t, \sigma\rangle} \Gamma \left( 1- \ii \langle\underline{d}^{(\alpha)}, \sigma \rangle \right) & = e^{\ii \langle t, \sigma\rangle}  \left(- \ii \langle\underline{d}^{(\alpha)}, \sigma \rangle \right) \Gamma \left(- \ii\langle\underline{d}^{(\alpha)}, \sigma \rangle \right) \\
	&= - \sum_{a=1}^{s} d_{a}^{(\alpha)} \frac{ \partial \ }{\partial t_a } e^{\ii \langle t, \sigma\rangle} \Gamma \left(- \ii\langle\underline{d}^{(\alpha)}, \sigma \rangle \right) 
\end{align*}
implies that 
\begin{equation}
\label{eq:ZXtoZVrel}
	\mz_{X} (\varnothing ;t) = \prod_{\alpha=1}^{r} \left( - \sum_{a=1}^{s} d_{a}^{(\alpha)} \frac{ \partial \ }{\partial t_a } \right) \mz_{V} \left( t \right) ,
\end{equation}
with the right-hand side given in terms of \eqref{eq:Zambient}.

\subsubsection{Box operators from integral manipulations}
\label{app:sec:BoxOp}

Using the integral representation of the $\Gamma$-function 
\begin{equation}
\label{eq:gammaint}
	\Gamma (z) = \int_0 ^{\infty} \frac{\dd y}{y} e^{-y + z \log (y)} ,
\end{equation}
we write 
\begin{equation*}
\begin{aligned}
	\mz_{V} \left( t \right)  &= \int_{\gamma_t} \dd^s \sigma~ \int_{(0,\infty)^{N}} \left( \prod_{i=1}^{N} e^{-y_i}  \frac{\dd y_i}{y_i} \right) ~ \int_{(0,\infty)^{r}} \left( \prod_{\alpha=1}^{r} e^{-\eta_{\alpha}}  \frac{\dd \eta_{\alpha}}{\eta_{\alpha}} \right) \\
	& \times \exp \left( \ii \langle t + \sum_{i=1}^{N} \underline{\ell}^{(i)} \log (y_i) - \sum_{\alpha=1}^{r} \underline{d}^{(\alpha)} \log (\eta_{\alpha}) , \sigma\rangle  \right) ,
\end{aligned}
\end{equation*}
where the variables $y_i$ (respectively $\eta_{\alpha}$) come from the integral representation of $\Gamma \left(  \ii \langle \underline{\ell}^{(i)}, \sigma \rangle \right)$ (respectively of $\Gamma \left(- \ii \langle\underline{d}^{(\alpha)}, \sigma \rangle \right)$). We collect the two sets of variables into a unique set $\left\{ y_{I}^{\prime}, \ I=1, \dots, N+r\right\}$ with 
\begin{equation}
\label{eq:changeetaytoyI}
	y_{I}^{\prime} = \begin{cases} \eta_{I} , & \text{ if } 1 \le I \le r \\ y_{I-r}  , & \text{ if } r+1 \le I \le r+N \end{cases}
\end{equation}
and use \eqref{eq:tildeLCICY} to write 
\begin{equation*}
	\mz_{V} \left( t \right)  = \int_{\gamma_t} \dd^s \sigma~ \int_{(0,\infty)^{N+r}} \left( \prod_{I=1}^{N+r} e^{-y_I^{\prime}}  \frac{\dd y_I^{\prime}}{y_I^{\prime}} \right) ~ \exp \left( \ii \langle t + \sum_{I=1}^{N+r} \underline{\tilde{\ell}}^{(I)} \log (y_I^{\prime}) , \sigma\rangle  \right) .
\end{equation*}
Integrating over $\sigma$ (and dropping the ${}^{\prime}$) yields 
\begin{equation*}
	\mz_{V} \left( t \right)  = \int_{(0,\infty)^{N+r}} \left( \prod_{I=1}^{N+r} e^{-y_I}  \frac{\dd y_I}{y_I} \right) ~ \prod_{a=1}^{s} \delta \left( t_a + \sum_{I=1}^{N+r} \tilde{\ell}_a^{(I)} \log (y_I) \right) .
\end{equation*}
Scaling the variables $y_I = - c_I \tilde{y}_I$, with the coefficients chosen to satisfy \eqref{eq:CICYGKZvar}, the integral becomes 
\begin{equation}
\label{eq:ZhatBCICYint}
	\mz_{V} \left( t \right)  = \int_{\mathcal{R}_{N+r}} \left( \prod_{I=1}^{N+r} e^{c_I\tilde{y}_I}  \frac{\dd \tilde{y}_I}{\tilde{y}_I} \right) ~ \prod_{a=1}^{s} \delta \left( \log  \left[ \prod_{I=1}^{N+r} (\tilde{y}_I)^{\tilde{\ell}_a^{(I)}} \right] \right) .
\end{equation}
The integration contour is 
\begin{equation*}
	\mathcal{R}_{N+r} := e^{- \ii \Arg (c_1)} \R_{< 0} \times \cdots \times e^{- \ii \Arg (c_{N+r})} \R_{< 0} .
\end{equation*}
Note that shifting $t \mapsto t - \ii 2 \pi q_k$ does not affect the relation \eqref{eq:CICYGKZvar}.\par
\medskip
The $\delta$-distributions enforce constraints 
\begin{equation*}
	 1= \left( \prod_{I \ : \ \tilde{\ell}_a ^{(I)} >0} ( \tilde{y}_I)^{ \tilde{\ell}_a ^{(I)}} \right) \cdot \left( \prod_{I \ : \ \tilde{\ell}_a ^{(I)} <0} ( \tilde{y}_I)^{ \tilde{\ell}_a ^{(I)}} \right)
\end{equation*}
for each $a=1, \dots, s$, which are conveniently rewritten in the form 
\begin{equation}
\label{eq:ZhatBCICYconstraint}
	\prod_{I \ : \ \tilde{\ell}_a ^{(I)} >0} ( \tilde{y}_I)^{ \tilde{\ell}_a ^{(I)}} - \prod_{I \ : \ \tilde{\ell}_a ^{(I)} <0} ( \tilde{y}_I)^{ - \tilde{\ell}_a ^{(I)}} =0 .
\end{equation}
In other words, the integrand is supported on a $(n+s)$-dimensional locus inside $\mathcal{R}_{N+r}$, cut out by \eqref{eq:ZhatBCICYconstraint}.\par
It follows that the differential operators 
\begin{equation*}
	\prod_{I \ : \ \tilde{\ell}_a ^{(I)} >0}  \left( \frac{\partial \ }{\partial c_I} \right)^{ \tilde{\ell}_a ^{(I)}} - \prod_{I \ : \ \tilde{\ell}_a ^{(I)} <0} \left( \frac{\partial \ }{\partial c_I} \right)^{ - \tilde{\ell}_a ^{(I)}  } 
\end{equation*}
annihilate \eqref{eq:ZhatBCICYint}, as a consequence of \eqref{eq:ZhatBCICYconstraint}.

\subsubsection{Euler operators}
\label{app:sec:EuOp}
We now demonstrate that, upon substituting \eqref{eq:CICYGKZvar}, $\mz_V$ is annihilated by the Euler operator \eqref{eq:GKZCICYEuler} with $\beta=0$. In fact, it is easy to prove a more general version of this statement.
\begin{lem}\label{lem:EuOpCICY}
    Let $g: \mathfrak{t}_{\C} \longrightarrow \C$ and consider 
    \begin{equation*}
        \mz [g] = \int_{\gamma_t} \dd^s \sigma~e^{\ii \langle t, \sigma\rangle} ~g(\sigma)
    \end{equation*}
    with $\gamma_t \subset \mathfrak{t}_{\C}$ a half-dimensional cycle such that $g\vert_{\gamma_t}$ has no singularities and the integral converges. With the replacement \eqref{eq:CICYGKZvar}, $\mz [g]$ is annihilated by the Euler operator \eqref{eq:GKZCICYEuler} with $\beta=0$.
\end{lem}
\begin{proof}
    Plugging \eqref{eq:CICYGKZvar} in the integrand, we have that 
    \begin{equation*}
        \mz [g] = \int_{\gamma_t} \dd^s \sigma~\prod_{I=1}^{N+r} c_I^{-\ii \langle \underline{\tilde{\ell}}^{(I)}, \sigma\rangle} ~g(\sigma) .
    \end{equation*}
    We apply $\sum_{J=1}^{N+r} \tilde{\nu}_{J,j} \left( c_J \frac{\partial \ }{\partial c_J} \right) $ to this expression and obtain 
    \begin{equation*}
    \begin{aligned}
        \sum_{J=1}^{N+r} \tilde{\nu}_{J,j} \left( c_J \frac{\partial \ }{\partial c_J} \right) \mz [g] &= \int_{\gamma_t} \dd^s \sigma~\left[ \sum_{J=1}^{N+r} \tilde{\nu}_{J,j} \left( -\ii \langle \underline{\tilde{\ell}}^{(J)}, \sigma\rangle \right)\prod_{I=1}^{N+r} c_I^{-\ii \langle \underline{\tilde{\ell}}^{(I)}, \sigma\rangle} \right] ~g(\sigma) \\
        &=  -\ii \int_{\gamma_t} \dd^s \sigma~ \langle \left[ \sum_{J=1}^{N+r} \tilde{\nu}_{J,j} \underline{\tilde{\ell}}^{(J)}  \right] , \sigma\rangle \prod_{I=1}^{N+r} c_I^{-\ii \langle \underline{\tilde{\ell}}^{(I)}, \sigma\rangle}  ~g(\sigma) .
    \end{aligned}
    \end{equation*}
    However, by construction (recall \S\ref{sec:CICYsetup}) we have that 
    \begin{equation*}
         \sum_{J=1}^{N+r} \tilde{\nu}_{J} \underline{\tilde{\ell}}^{(J)} = 0 ,
    \end{equation*}
    thus proving the claim.
\end{proof}
Let $\mz [g]$ be as in Lemma \ref{lem:EuOpCICY}, and fix $v \in \C$. We immediately have that
\begin{equation*}
    \widehat{\mz}^{(v)}[g] := \left( \prod_{\alpha=1}^{r} c_{\alpha} \right)^{-v} \mz [g]
\end{equation*}
is annihilated by the Euler operator \eqref{eq:GKZCICYEuler} with $\beta_j=-v$ for $j=1, \dots, r$ and, and $\beta_j=0$ for $j>r$. Indeed, 
\begin{equation*}
\begin{aligned}
    \sum_{J=1}^{N+r} \tilde{\nu}_{J,j} \left( c_J \frac{\partial \ }{\partial c_J} \right)  \widehat{\mz}^{(v)}[g] &= \left( \prod_{\alpha=1}^{r} c_{\alpha} \right)^{-v}  \sum_{J=1}^{N+r} \tilde{\nu}_{J,j} \left( c_J \frac{\partial \ }{\partial c_J} \right)  \mz [g] \\
    & \quad + \left[ \sum_{J=1}^{N+r} \tilde{\nu}_{J,j} \left( c_J \frac{\partial \ }{\partial c_J} \right) \left( \prod_{\alpha=1}^{r} c_{\alpha} \right)^{-v} \right] \mz [g] \\
    &= -v  \left[ \sum_{J=1}^{r} \tilde{\nu}_{J,j} \right]  \widehat{\mz}^{(v)}[g] \\
    &= \begin{cases} -v ~ \widehat{\mz}^{(v)}[g] & \ 1 \le j \le r \\ 0 & \ \text{otherwise} . \end{cases}
\end{aligned}
\end{equation*}
The first piece on the right-hand side vanishes by Lemma \ref{lem:EuOpCICY}, and in the last line we have used that $ \tilde{\nu}_{J,j} = \delta_{J,j}$ if $1 \le J \le r$.

\subsubsection{Compact Calabi--Yau from non-compact Calabi--Yau}
\label{app:sec:prefactor}

To conclude the proof, we observe that, from the definition \eqref{eq:CICYGKZvar}, we have 
\begin{equation*}
    \frac{\partial \ }{\partial c_{\alpha}} = \sum_{a=1}^{s} \frac{\partial t_a } {\partial c_{\alpha}} \frac{\partial \ }{\partial t_a } = - \frac{1}{c_{\alpha}} \sum_{a=1}^{s} d_a^{(\alpha)} \frac{\partial \ }{\partial t_a } . 
\end{equation*}
Thus, identity \eqref{eq:ZXtoZVrel} is equivalent to 
\begin{equation*}
	\left( \prod_{\alpha=1}^{r} \frac{1}{c_{\alpha}}\right) \mz_{X} (\varnothing ;t) = \prod_{\alpha=1}^{r} \left(\frac{\partial \ }{\partial c_{\alpha}}  \right) \mz_{V} \left( t \right) .
\end{equation*}
Therefore, if a differential operator annihilates $\mz_{V} $ and commutes with $\frac{\partial \ }{\partial c_{\alpha}}$, then it annihilates $\widehat{\mz}_X$ as defined in \eqref{eq:AperiodCICY}.\par
Applying this statement to the box operators from \eqref{eq:GKZCICYBox}, and combining with the result from \S\ref{app:sec:EuOp}, concludes the proof of Theorem \ref{thm:CICYGKZ}.

\subsection{Modifications preserving the GKZ system}
\label{app:GKZlemma}
This appendix contains the proof of Lemma \ref{lemma:GKZlemma}.\par
We begin with point (1) of the Lemma. Choose an arbitrary subset $\mathfrak{M}_{\downarrow} \subseteq \left\{1, \dots, N \right\} $, and denote for shortness $\mathfrak{M}_{\uparrow} := \left\{1, \dots, N \right\} \setminus \mathfrak{M}_{\downarrow}$. We want to show that the derivation in \S\ref{app:GKZCICY} is unchanged under the replacement 
\begin{equation*}
    \prod_{i=1}^{N} \Gamma \left(  \ii \langle \underline{\ell}^{(i)}, \sigma \rangle \right)  \ \mapsto \ \frac{ \prod_{i \in \mathfrak{M}_{\uparrow} } \Gamma \left(  \ii \langle \underline{\ell}^{(i)}, \sigma \rangle \right)  }{ \prod_{i \in \mathfrak{M}_{\downarrow} } \Gamma \left(  1-\ii \langle \underline{\ell}^{(i)}, \sigma \rangle \right)  } ,
\end{equation*}
possibly up to a shift of $t$ in $\ii \pi \Z^s$ in the identification of parameters \eqref{eq:CICYGKZvar}.\par 
This modification does not affect the relation between $\mz_X$ and $\mz_V$. Therefore, we only have to study the replacement of Gamma functions in \eqref{eq:Zambient}, leading to the consideration of 
\begin{equation*}
	\mz_{V}^{\text{\rm mod}} \left( t \right) := \int_{\gamma_t} \dd^s \sigma~e^{\ii \langle t, \sigma\rangle} ~\frac{ \prod_{i \in \mathfrak{M}_{\uparrow} } \Gamma \left(  \ii \langle \underline{\ell}^{(i)}, \sigma \rangle \right)  }{ \prod_{i \in \mathfrak{M}_{\downarrow} } \Gamma \left(  1-\ii \langle \underline{\ell}^{(i)}, \sigma \rangle \right)  } \prod_{\alpha=1}^{r} \Gamma \left( -\ii \langle\underline{d}^{(\alpha)}, \sigma \rangle \right) .
\end{equation*}

\subsubsection{Integral manipulations}
\begin{figure}
    \centering
    \includegraphics[width=0.8\linewidth]{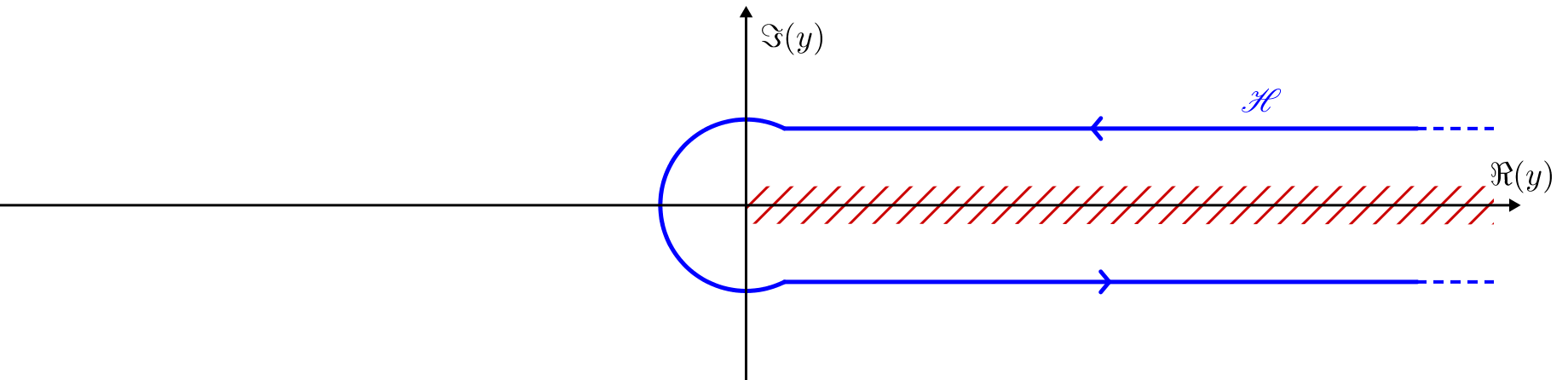}
    \caption{Hankel contour $\mathscr{H}$, in blue. The integrand in \eqref{eq:invgamma} has a branch cut along the positive real axis.}
    \label{fig:Hankel-contour}
\end{figure}

We utilize the integral representation \eqref{eq:gammaint} for all Gamma functions in the numerator. For the Gamma functions in the denominator, we use the integral representation 
\begin{equation}
\label{eq:invgamma}
    \frac{1}{\Gamma (1-z)} = \int_{\mathscr{H}} \frac{\dd y}{y} e^{-y + z \log (-y)} ,
\end{equation}
where the integration cycle is a Hankel contour encircling the positive real line, avoiding the branch cut in $\Re (y) \ge 0$, depicted in Figure \ref{fig:Hankel-contour}. This gives:
\begin{equation*}
\begin{aligned}
	\mz_{V}^{\text{\rm mod}} \left( t \right)  &= \int_{\gamma_t} \dd^s \sigma~ \int_{(0,\infty)^{\lvert \mathfrak{M}_{\uparrow}\rvert } }\int_{\mathscr{H}^{\lvert \mathfrak{M}_{\downarrow}\rvert } } \left( \prod_{i=1}^{N} e^{-y_i}  \frac{\dd y_i}{y_i} \right)~ \int_{(0,\infty)^{r}} \left( \prod_{\alpha=1}^{r} e^{-\eta_{\alpha}}  \frac{\dd \eta_{\alpha}}{\eta_{\alpha}} \right) \\
	& \times \exp \left( \ii \langle t + \ii \pi \sum_{i\in \mathfrak{M}_{\downarrow}} \underline{\ell}^{(i)} + \sum_{i=1}^{N} \underline{\ell}^{(i)} \log (y_i) - \sum_{\alpha=1}^{r} \underline{d}^{(\alpha)} \log (\eta_{\alpha}) , \sigma\rangle  \right) .
\end{aligned}
\end{equation*}
We fix a convention for $\log (-1)$, so that the minus sign in the argument of $\log (-y_i)$ for all $i\in \mathfrak{M}_{\downarrow}$ results in a $(\ii \pi \Z)^s$-shift of $t$.\par
Following \S\ref{app:GKZCICY}, we make the change of variables \eqref{eq:changeetaytoyI} followed by $y_I=-c_I \tilde{y}_I$, now subject to 
\begin{equation}
\label{eq:modGKZrelCI}
    \prod_{\alpha=1}^{r} \left( - c_{\alpha} \right)^{-d_a ^{(\alpha)}} \prod_{i \in \mathfrak{M}_{\uparrow} } \left( - c_{r+i} \right)^{\ell_a ^{(i)}} \prod_{i \in \mathfrak{M}_{\downarrow} } \left( c_{r+i} \right)^{\ell_a ^{(i)}} = e^{-t_a  } , \qquad  a=1, \dots, s .
\end{equation}
This yields 
\begin{equation*}
	\mz_{V}^{\text{\rm mod}} \left( t \right)  = \int_{\mathcal{R}^{\text{\rm mod}}_{N+r}} \left( \prod_{I=1}^{N+r} e^{c_I\tilde{y}_I}  \frac{\dd \tilde{y}_I}{\tilde{y}_I} \right) ~ \prod_{a=1}^{s} \delta \left( \log  \left[ \prod_{I=1}^{N+r} (\tilde{y}_I)^{\tilde{\ell}_a^{(I)}} \right] \right) ,
\end{equation*}
which differs from the unmodified case only in the relation \eqref{eq:modGKZrelCI} and in the integration cycle. The integration contour in this case is 
\begin{equation*}
	\mathcal{R}^{\text{\rm mod}}_{N+r} := \prod_{i=1}^{r} e^{- \ii \Arg (c_i)} \R_{< 0} \prod_{i \in \mathfrak{M}_{\uparrow}} e^{- \ii \Arg (c_{r+i})} \R_{< 0} \prod_{i \in \mathfrak{M}_{\downarrow}} e^{- \ii \Arg (c_{r+i}) - \ii \pi} \mathscr{H}.
\end{equation*}
The derivation of the box operators in not affected by these changes.

\subsubsection{Second statement in the lemma}
We now prove point (2) in Lemma \ref{lemma:GKZlemma}.\par
We observe that inserting 
\begin{equation*}
    \frac{\pi}{ \sin (\ii \pi \langle \underline{\ell}, \sigma \rangle )} = \Gamma \left(\ii \langle \underline{\ell}, \sigma \rangle \right) \Gamma \left( 1-\ii \langle \underline{\ell}, \sigma \rangle \right)
\end{equation*}
in the integrand, is tantamount to consider the $A$-periods of a different Calabi--Yau complete intersection $X_0$, constructed by 
\begin{enumerate}[(i)]
    \item adding a homogeneous coordinate $\phi_{0}$ of weight $\underline{\ell}$ to the toric ambient space, and 
    \item intersecting with an additional hypersurface of multi-degree $\underline{\ell}$.
\end{enumerate}
We then use the invariance of the $A$-periods of $X_0$ under deformations of the complex structure. In particular, the $A$-periods of $X_0$ only depend on the K\"ahler moduli $\zeta_a$, and not on the complex structure moduli. It is thus possible to choose the additional hypersurface from (ii) in a non-generic way, without affecting the $A$-periods. Selecting it to be $\{ \phi_0=0 \}$, we get an isomorphism between $X_0$ (on this positive-codimensional locus of its complex structure moduli space) and $X$.\par
Concerning the toric data, passing to $X_0$ introduces to additional vertices $\tilde{\nu}_0, \tilde{\nu}_{N+r+1}$ and two new variables, namely $c_0,c_{N+r+1}$ associated to the new coordinate $\phi_0$ and the corresponding section $s_0(\phi)=\phi_0$. The GKZ for $X_0$ thus has more redundancy than the GKZ for $X$, which is recovered by fixing $c_0\equiv 1\equiv c_{N+r+1}$.\par
We thus conclude that the GKZ system that annihilates the $A$-periods of $X$ also annihilates the $A$-periods of $X_0$, showing point (2) of Lemma \ref{lemma:GKZlemma}.

\section{GKZ moduli from GLSM moduli}
\label{app:FImoduli}

In this appendix, we clarify the meaning of the moduli $c_I$ and $x_I$ appearing in the GKZ systems (\S\ref{sec:GKZdef}) from the GLSM perspective. In particular, our goal is to show that the GLSM automatically encodes:
\begin{itemize}
    \item The relations \eqref{eq:xtotnc} and \eqref{eq:CICYGKZvar} between the GKZ moduli and the moduli $e^{-t_a}$;
    \item The factors in the definition of the $A$-periods \eqref{eq:Anc} and \eqref{eq:AperiodCICY}.
\end{itemize}\par
We begin with a GLSM describing an affine space $\mathbb{A}:=\C^{N_+}$, where $N_+=N+\hat{r}$ or $N_+=N+r$ in the cases of interest in the main text. There is a natural $(\C^{\ast})^{N_+}$-action on $\mathbb{A}$, with moduli provided by a choice of character 
\begin{equation*}
    (c_{1}, \dots, c_{N_{+}} ) \in \mathrm{Hom} \left( \pi_1 \left( (\C^{\ast})^{N_+} \right), \C^{\ast} \right) \cong  \mathrm{Hom} \left( \Z^{N_+},\C^{\ast} \right) \cong  (\C^{\ast})^{N_+} .
\end{equation*}\par
\begin{table}[th]
\centering
\begin{tabular}{l |c |c | c |c}
	field & \# & $U(1)_I$ charge & $U(1)_a$ gauge charge & $U(1)_V$ R-charge\\
	\hline
	$\phi_I$ & $I=1, \dots, N_+$ & $+1$ & $\tilde{\ell}_a ^{(I)}$ & $2v^{(I)}$\\
	\hline
\end{tabular}
\caption{Field content and charges of the GLSM for affine varieties $\mathbb{A}=\C^{N_+}$ and symplectic quotients $X_+$.}
\label{tab:GLSMaffine}
\end{table}\par
Next, we are interested in a variety $X_{+}$ obtained as the symplectic quotient of $\mathbb{A}$ by a selected $(\C^{\ast})^{s} \subset (\C^{\ast})^{N_+}$, which acts on the $I^{\text{th}}$ affine coordinate with weights $\underline{\ell}^{(I)}$.
More precisely, we define 
\begin{equation*}
    X_+ = \dd W^{-1} (0) \cap \mathbb{A} / (\C^{\ast})^{s} 
\end{equation*}
for a $(\C^{\ast})^{s}$-invariant, holomorphic function $W : \mathbb{A} \longrightarrow \C$. $X_+$ inherits a Calabi--Yau structure if $\sum_{I=1}^{N_+}\tilde{\underline{\ell}}^{(I)} = \underline{0}\in \R^s$. We further demand that $W$ is quasi-homogeneous of weight $2$ under the action of $U(1)_V$, and assign $U(1)_V$ charges $2 v^{(I)}$ accordingly. The GLSM is summarized in \textsc{Table} \ref{tab:GLSMaffine}.\par
\begin{rmk}
    In the physics language, $(\C^{\ast})^{N_+}$ is the maximal torus of the flavor symmetry $GL (N_+, \C)$, and we gauge a rank-$s$ subgroup of it. The equivariant parameters for the $(\C^{\ast})^{N_+}$-action are background twisted chiral fields, and the gauging operation promotes the corresponding combination of background fields to dynamical fields.
\end{rmk}\par
The $(\C^{\ast})^{N_+}$-equivariant $A$-period of $\mathbb{A}$, with equivariant parameters 
$$\tilde{\sigma}=(\tilde{\sigma}^{(1)}, \dots, \tilde{\sigma}^{(N_+)})$$ 
is given by 
\begin{equation*}
    \mz_{\mathbb{A}} (\Br; c) = \prod_{I=1}^{N_+} c_I^{- \ii \tilde{\sigma}^{(I)}} \Gamma \left( \ii \tilde{\sigma}^{(I)} \right) ~f_{\Br} (\tilde{\sigma}) ,
\end{equation*}
for any $\Br$. To obtain the $A$-period of $X_+$, we identify $\tilde{\sigma}^{(I)}= \langle \tilde{\underline{\ell}}^{(I)} , \sigma \rangle - \ii v^{(I)}$ and integrate over $\sigma$. This change yields
\begin{equation*}
    \pz_{X_+} (\Br_+ ; c) = \left( \prod_{I=1}^{N_+} c_I^{-v^{(I)}} \right) \int_{\gamma} \dd \sigma \prod_{I=1}^{N_+} c_I^{- \ii \langle \tilde{\underline{\ell}}^{(I)} , \sigma \rangle} \Gamma \left( v^{(I)} + \ii\langle \tilde{\underline{\ell}}^{(I)} , \sigma \rangle \right) ~f_{\Br_+} (\sigma)
\end{equation*}
where we have denoted $\Br_+ $ any object in $ \mathrm{MF}_{U(1)^s} (W)$ such that 
\begin{equation*}
    f_{\Br_+} (\sigma) = f_{\Br} (\tilde{\sigma})\Big\vert_{\tilde{\sigma}^{(I)} = \langle \tilde{\underline{\ell}}^{(I)} , \sigma \rangle - \ii v^{(I)}}
\end{equation*}
obtained from a suitable choice of matrix factorization $\Br$. Using  
\begin{equation*}
     \prod_{I=1}^{N_+}c_I^{-  \tilde{\ell}_a^{(I)}} = e^{-t_a} , \qquad a=1, \dots, s ,
\end{equation*}
we arrive at the expected form of the partition function of $X_+$, multiplied by a prefactor 
\begin{equation*}
    \pz_{X_+} (\Br_+ ; c) = \left( \prod_{\substack{I=1 \\ I \ : \ v^{(I)} \ne 0 }}^{N_+} c_I^{-v^{(I)}} \right) \int_{\gamma} \dd \sigma e^{\ii \langle t , \sigma \rangle} \prod_{I=1}^{N_+} \Gamma \left( v^{(I)} + \ii\langle \tilde{\underline{\ell}}^{(I)} , \sigma \rangle \right) ~f_{\Br_+} (\sigma) .
\end{equation*}
This expression is consistent with the definition of $A$-period given in \eqref{eq:Anc} and \eqref{eq:AperiodCICY}.\par
In conclusion, tracking the parameters through the symplectic quotient construction, the GLSM automatically yields not only the partition function, but the full $A$-period, equipped with the $c$-dependent prefactor and the maps \eqref{eq:xtotnc} and \eqref{eq:CICYGKZvar} between parameters.

\end{appendix}
\bibliography{GLSM}
\end{document}